\documentclass[11pt]{article}
\usepackage{eurosym}
\usepackage{natbib} 
\usepackage{amsmath}
\usepackage{amssymb}
\usepackage{amsthm}
\usepackage{mathtools}
\usepackage{dutchcal}
\usepackage[scr=rsfs]{mathalpha}
\usepackage{bm}
\usepackage{graphicx}
\usepackage[singlespacing]{setspace}
\usepackage[bottom,multiple]{footmisc}
\usepackage[justification=centering]{caption}
\usepackage{subcaption}
\usepackage{setspace}
\usepackage{color}
\usepackage{ragged2e}
\usepackage{bigstrut}
\usepackage{lscape}
\usepackage{arydshln}
\usepackage{booktabs}                                         
\usepackage{multirow}                                          
\usepackage[dvipsnames]{xcolor}
\usepackage{comment}
\usepackage{relsize}
\usepackage{soul}

\definecolor{cranberry}{HTML}{DC143C}
\DeclareMathOperator*{\argmin}{arg\,min}

\RequirePackage[breaklinks, colorlinks,citecolor=Mahogany,urlcolor=gray,linkcolor=cranberry]{hyperref}

\newtheorem{theorem}{Theorem}

\newtheorem{proposition}[theorem]{Proposition}

\newcommand{\vm}{\vspace{.1in}}

\def\toprule{\hline\hline}
\def\midrule{\hline}
\def\bottomrule{\hline\hline}
\def\tablenotes{\\ \noindent \justifying \small}

\def\toprule{\hline\hline\addlinespace[0.2cm]}
\def\midrule{\addlinespace[0.1cm]\hline\addlinespace[0.1cm]}
\def\bottomrule{\addlinespace[0.2cm]\hline\hline}
\def\tablenotes{\\ \vm \noindent \justifying \footnotesize Notes: }

\renewcommand{\today}{\ifcase \month \or January \or February \or March \or %
April \or May \or June \or July \or August \or September \or October \or November \or %
December \fi \number \year}

\usepackage{titling}

\begin{document}

\title{Which Green Technology to Subsidize? \\ Evidence from Electric Vehicles in South Korea\thanks{We thank Soren Anderson, Luming Chen, Charles Murry, and Jingyuan Wang for their valuable comments. We are also grateful to seminar participants at the University of Michigan IO Lunch, Korea University, Yonsei University, MSU/UM EEE day, North America Summer Meeting of the Econometric Society, and IIOC. Sukyung Joo, Beomsu Kim, Yongrok Kim, and Chaeeun Shin provided outstanding research assistance. We gratefully acknowledge the 2024 Data Voucher Program from the Ministry of Science and ICT of South Korea and the BK21 FOUR program, funded by the Ministry of Education and the National Research Foundation of South Korea. All authors equally contributed to this paper.}} 
\author{Youngjin Hong\thanks{%
Department of Economics, University of Michigan, E-mail: \texttt{yjinhong@umich.edu}.} \and
In Kyung Kim\thanks{%
Department of Economics, Sogang University, E-mail: \texttt{inkim@sogang.ac.kr}.}
\and
Frank Verboven\thanks{%
Department of Economics, KU Leuven and CEPR, E-mail: \texttt{frank.verboven@kuleuven.be}.}}
\maketitle
\thispagestyle{empty}

\begin{abstract}
We develop a framework to compare the relative effectiveness of subsidizing alternative emission-reducing technologies. We show that an intermediate technology may reduce emissions more effectively than the cleanest technology if it induces sufficiently greater substitution away from the prevailing high-emission technology. We apply the framework to the South Korean passenger vehicle market using a demand model that incorporates mileage heterogeneity, an important determinant of fuel-type choice. First, reallocating existing subsidies from battery electric vehicles (BEVs), the cleanest technology, to hybrid electric vehicles (HEVs), an intermediate technology, would reduce total greenhouse gas emissions by an additional 47\%. Second, for a BEV-focused subsidy policy to outperform an HEV-focused policy, the carbon intensity of electricity generation would need to fall by approximately 45\%. Our findings suggest that HEV subsidies remain more effective than BEV subsidies until consumers become sufficiently willing to switch to BEVs or electricity generation becomes sufficiently decarbonized.

\bigskip
\bigskip

\end{abstract}

\onehalfspacing

\clearpage

\section{Introduction}
\setcounter{page}{1}

The transition from high-emission to low-emission technologies rarely occurs in a single step. The cleanest technologies typically coexist with both dirty and intermediate technologies. To advance progress toward net-zero greenhouse gas (GHG) emissions, governments often concentrate policy support on the technological frontier with the lowest emissions. Although ranking technologies by their stand-alone emissions provides a simple and salient criterion for allocating policy support, a technology's environmental effectiveness also depends on the extent to which it displaces dirtier technologies. A subsidy for an intermediate technology can therefore generate larger emission reductions than a subsidy for a frontier technology if it induces substantially greater substitution away from prevailing dirty technologies.

To evaluate the relative effectiveness of subsidies for alternative emission-reducing technologies, this paper focuses on the automobile market. This market offers considerable scope for emission savings because of multiple emission technologies and substantial substitution possibilities among differentiated products. In particular, battery electric vehicles (BEVs), powered entirely by rechargeable batteries, are widely viewed as the long-run future of the automotive industry, replacing conventional internal combustion engine vehicles (ICEVs). Accordingly, many governments have promoted BEV adoption through generous purchase subsidies, tax incentives, and investments in charging infrastructure.\footnote{For an overview of policies in EU countries and the US, see \citet{ACEA2025} and \citet{AFDC2026_US}, respectively.} By contrast, hybrid electric vehicles (HEVs), which combine an internal combustion engine with an electric motor, have received less policy attention.
 
Despite substantial government incentives, the expansion of charging networks, and technological improvements such as faster charging and longer driving ranges, BEV adoption has fallen short of policy targets. Explanations include high purchase prices even net of subsidies, continued consumer concerns about charging convenience and battery safety, and regulatory uncertainty.\footnote{For example, in March 2025, the EU amended its CO$_2$ regulation to provide automakers with additional compliance flexibility by allowing them to average emission performance over the period 2025-2027 rather than requiring compliance in 2025 alone. In the United States, qualified BEVs were eligible for a federal tax credit of up to \$7\,500 for vehicles registered before October 2025, but eligibility expired under the \emph{One Big Beautiful Bill Act}.} At the same time, despite receiving substantially less policy support than BEVs, HEVs have continued to gain market share and outsell BEVs across most major economies. For example, in the United States, BEV market share growth has recently slowed, with the share reaching 7.8\% in 2024, whereas the HEV market share has grown more rapidly, reaching 10.6\%. The gap between HEV and BEV market shares is even more pronounced in several European countries, including Germany, France, Italy, and the United Kingdom, as well as in Asian countries such as Japan, South Korea, and India. China is a notable exception, where HEVs capture only a small market share while BEVs account for nearly 30\% of new vehicle sales.\footnote{Figure \ref{fig: BEV and HEV share} in Appendix \ref{sec: Additional figures and tables} shows the evolution of BEV and HEV sales shares across these economies.}

The growing importance of HEVs raises the question of whether policies promoting HEV adoption may reduce GHG emissions more effectively than policies focused exclusively on BEVs. Relative to BEVs, HEVs are typically less expensive and may better align with current consumer preferences, making them potentially more effective at inducing substitution away from higher-emission ICEVs. Moreover, under current production technologies and electricity generation mixes across different energy sources, the lifecycle GHG emission advantage of BEVs over HEVs may be smaller than often presumed. While BEVs produce no tailpipe emissions, electricity generation remains carbon intensive in many countries, and battery production is itself highly carbon intensive. These considerations suggest that subsidizing HEVs rather than BEVs may, under certain conditions, lead to larger overall reductions in GHG emissions.

We develop and empirically implement a framework to compare the relative effectiveness of subsidies for alternative emission-reducing technologies (in our application, HEVs versus BEVs). We begin with a theoretical model that highlights a possible tradeoff: while one technology may generate lower stand-alone emissions per vehicle, the alternative technology may yield greater diverted emissions by inducing more substitution away from the prevailing high-emission technology. 

Next, we implement the model to examine the South Korean new passenger vehicle market, where BEVs have received substantial government subsidies, whereas HEVs have received little policy support. We estimate a random coefficients (RC) logit demand system (\citealp{Berry94, Berry95}) using product-province-year-level data from 2012 to 2023 on vehicle sales, prices, and characteristics. We combine these data with individual-level data on annual mileage and second-choice survey responses by adding micro-moments \citep[e.g.,][]{conlon2025incorporating,grieco2024evolution} to account for rich sources of observed and unobserved consumer heterogeneity. The estimates show that consumers perceive vehicles with the same fuel type as closest substitutes. More interestingly, across fuel types, HEVs are perceived as closer substitutes for ICEVs than BEVs are. Mileage heterogeneity plays an important role in the estimated substitutability across fuel types: consumers who drive more place greater value on vehicles with low operating costs such as BEVs and HEVs. 

We subsequently compare the emission effects of subsidizing HEVs instead of BEVs, accounting for price responses in an equilibrium model of Bertrand price-setting firms. We construct comprehensive product-year-level lifecycle GHG emission measures that vary across individuals because of differences in vehicle usage (mileage). We incorporate two types of emissions: fuel-cycle emissions, relating to (heterogeneous) usage, and vehicle-cycle emissions, relating to production and recycling. Under South Korea's current electricity generation mix, lifecycle emissions from BEVs are on average 38\% lower than those from ICEVs and 15\% lower than those from HEVs.

Despite the lower per-vehicle emissions of BEVs relative to HEVs, we find that reallocating the existing subsidy budget from BEVs to a uniform HEV subsidy would reduce total emissions by an additional 47\%. This advantage arises because HEVs induce greater substitution away from conventional ICEVs than BEVs do. In contrast, BEV subsidies primarily generate within-BEV substitution, thereby attenuating their overall emission benefits.

We then consider alternative BEV emission levels under electricity generation mixes matching those observed in other countries. While South Korea's electricity generation mix is currently comparable to those of the United States and Germany, it relies much less on fossil fuels than the mixes in China and India. We find that, for BEV subsidies to become more effective than HEV subsidies, South Korea's electricity generation mix would need to become at least 45\% cleaner through decarbonization of electricity supply, approximately reaching the level observed in Portugal. Only a few countries, such as Norway, Brazil, or France, currently have even cleaner electricity generation mixes. Conversely, our findings imply that BEV subsidies would generate only limited emission reductions, or could even increase emissions, if electricity generation were as carbon-intensive as in countries such as India and Kazakhstan.

In sum, our analysis establishes that subsidies should not necessarily target the technology with the lowest per-vehicle emissions. Subsidies for intermediate emission-reducing technologies may be more effective if they induce greater substitution away from the prevailing high-emission technology. Furthermore, this tradeoff may shift in favor of BEVs only when a country's electricity generation mix becomes sufficiently less carbon-intensive. The fact that many countries continue to prioritize BEVs suggests that these policies may reflect industrial policy objectives in addition to environmental concerns. We return to this in the concluding section.


\paragraph{Related literature}

Our paper relates to several strands in the literature. First, a large literature studies the effects of government incentives to promote the demand for electric vehicles. Early work examines policies targeting HEVs.\footnote{See in particular \citet{chandra2010green,beresteanu2011gasoline,gallagher2011giving,sexton2014conspicuous,heutel2015consumer,gulati2017tax}. Most recently, \citet{langford2023quantifying} conduct a comprehensive examination of the economic benefits of HEVs in a flexible demand framework for California.} More recent work studies policies targeting BEVs.\footnote{Part of this literature focuses on the impact of these policies on BEV adoption \citep[e.g.,][]{deshazo2017designing, clinton2019providing, muehlegger2022subsidizing, barwick2024attribute}, as well as the role of heterogeneity in adoption \citep[e.g.,][]{davis2025political, archsmith2022future, gillingham2023has, NBERw29842}. Other studies focus more on the implications of these policies for emissions \citep[e.g.,][]{holland2016there, xing2021does, muehlegger2023correcting, fournel:hal-04971697, allcott2024effects}. Finally, several studies examine the complementary role of charging networks \citep[e.g.,][]{li2017market, zhou2018technology, li2019compatibility, springel2021network, fournel:hal-04971697}.} 
This literature focuses on the impact of promoting a single technology presumed to be the cleanest available one, initially HEVs and subsequently BEVs.

Our paper instead highlights the importance of comparing the relative effectiveness of subsidies for the cleanest technology (BEVs) and an intermediate technology (HEVs). We show theoretically that this comparison depends not only on the technologies' stand-alone emission savings, but also on their capacity to induce substitution away from the prevailing high-emission technology (conventional ICEVs). We then demonstrate empirically for the South Korean passenger vehicle market that a subsidy for the intermediate technology (HEVs) reduces emissions more effectively than a subsidy for the cleanest technology (BEVs), because the intermediate technology diverts more sales from the prevailing high-emission technology (ICEVs).

Second, a related literature evaluates the environmental effects of policies promoting BEVs, with particular attention to the measurement of emissions.\footnote{See, among others, \citet{holland2016there, xing2021does, muehlegger2023correcting, guo2023welfare, meng2026pollution, fournel:hal-04971697}.} Emissions measurement has become increasingly comprehensive, as in the recent work of \citet{allcott2024effects} and \citet{heid2025equilibrium}. Notably, \citet{holland2016there} account for the carbon intensity of electricity generation. But across U.S. states it plays only a limited role in explaining differences in local environmental benefits, compared with differences in local pollution damages. We show that carbon intensity becomes critical when considering the broader range of electricity generation mixes across countries and over the course of the energy transition. In particular, BEV subsidies become more effective than HEV subsidies only when the electricity generation mix is sufficiently clean.

Third, there is a broader literature on technical change and the transition from dirty to clean energy inputs \citep[e.g.,][]{acemoglu2012environment, acemoglu2016transition, papageorgiou2017substitution}. This work focuses on dynamic innovation incentives and shows that the gains from environmental policy depend on the input elasticity of substitution and the productivity gap between clean and dirty technologies. 
Our work is complementary to this literature: we highlight the importance of demand-side substitution in consumer markets characterized by imperfect competition. The automobile market provides a particularly suitable setting because it has been a major focus of environmental policy, with well-established empirical methodologies to recover detailed substitution patterns across multiple technologies.

\vspace{1em}
\noindent
The remainder of this paper is organized as follows. Section \ref{sec: theory} presents the theoretical model highlighting the mechanisms underlying the comparison of subsidies for two emission-reducing technologies. Section \ref{sec: background} covers relevant industry background, while Section \ref{sec: data} describes the data used in the empirical analysis. Section \ref{sec: model and empiric} introduces the demand model and discusses the estimation results, and Section \ref{sec: environmental evaluation} conducts the counterfactual analysis comparing emissions under BEV and HEV subsidies. Finally, Section \ref{sec: conclusion} concludes.

\section{Theoretical framework}\label{sec: theory}

This section develops a theoretical framework to compare the relative effectiveness of subsidizing alternative emission-reducing technologies, in our case HEVs and BEVs. The framework serves as a guide to identify the key channels that operate in our more elaborate empirical model. We start from a simple setup with three products, homogeneous vehicle usage, and full subsidy pass-through, and subsequently generalize it to multiple products, heterogeneous vehicle usage, and equilibrium pass-through under imperfect competition.

\subsection{Basic setup}
Consider three substitute products, namely an ICEV, an HEV, and a BEV, indexed by $j \in \{I,H,B\}$, in a market of size $M$. Each product generates a negative externality due to emissions from production and usage. Let $e_j$ denote emissions per unit of product $j$.

Assume that the ICEV represents the dirtiest technology, so that $e_j < e_I$ for $j \in \{ H,B \}$. The relative cleanliness of the HEV and the BEV depends on the phase of the energy transition. Since the BEV entails higher emissions during production than the HEV, it is possible that $e_H < e_B$ in the early phase of the energy transition, when electricity generation still relies heavily on fossil fuels. By contrast, in the later phase, when the energy mix relies more heavily on renewables, we have $e_B < e_H$.

To reduce emissions, the government considers a subsidy policy $\bm{\tau}=(\tau_I,\tau_H,\tau_B)$, where $\tau_j$ denotes the per-unit subsidy for product $j$. Consequently, the net consumer price of product $j$ is $p_j=p^G_j-\tau_j$ where $p^G_j$ is the gross price set by firms. Let $s_j(\mathbf{p})$ denote the market share of product $j$ as a function of the consumer price vector $\mathbf{p}$, with $\frac{\partial s_j}{\partial p_j}<0$ and $\frac{\partial s_k}{\partial p_j}>0$ for $k \neq j$. Total emissions and total government spending are given, respectively, by
\begin{equation*}
   E(\bm{\tau}) = M\sum_{k\in\{I,H,B\}}e_k s_k(\mathbf{p}(\bm{\tau}))  \quad \text{and} \quad G(\bm{\tau}) = M\sum_{k\in\{I,H,B\}} \tau_k s_k(\mathbf{p}(\bm{\tau})),
\end{equation*}
where $\mathbf{p}(\bm{\tau})$ denotes the vector of consumer prices under $\bm{\tau}$. For simplicity, assume that the subsidy policy does not affect firms' equilibrium gross prices $p^G_j$, so that there is complete subsidy pass-through and no price response of non-subsidized products to the subsidy policy (as under perfect competition with constant marginal costs).

First, consider the stand-alone impact of a subsidy $\tau_j$ for either the HEV or the BEV, $j \in \{ H,B \} $, on total emissions. To highlight the role of substitution, recall the commonly used diversion ratio from $j$ to $k$, $D_{j\to k} \equiv  \frac{\partial s_k/\partial p_j}{-\partial s_j / \partial p_j}$; this is the fraction of sales diverted from $j$ to $k$ after a price increase of product $j$. Now, define the diverted emissions of product $j$ as:
\begin{equation*}
    e_j^D \equiv  \sum_{k\neq j}D_{j\to k} e_k.
\end{equation*}
This is the weighted average of emissions from alternative products $k\neq j$, where the diversion ratios serve as weights. Intuitively, $e^D_j$ measures the expected emissions that would be generated by the marginal consumer of product $j$ absent the subsidy increase.

Applying these definitions, the impact of a subsidy for product $j \in \{ H,B \}$ on total emissions is given by:
\begin{eqnarray}
    \frac{\partial E}{\partial \tau_j} & = & - M \sum_{k\in\{I,H,B\}} e_k\frac{\partial s_k}{\partial p_j}  \nonumber \\ 
    & = & M \frac{\partial s_j}{\partial p_j} \left( e^D_j - e_j \right).  \label{eqn: partial E partial tau}
\end{eqnarray} 
This implies that a subsidy for product $j \in \{ H,B \}$ reduces total emissions if its net diverted emissions are positive, $e^D_j-e_j>0$, i.e., if product $j$'s own emissions are lower than its diverted emissions.

We now compare the \textit{relative} effectiveness of subsidies for the HEV and the BEV. To this end, assume that the ICEV is neither taxed nor subsidized, and consider a budget-neutral subsidy reallocation from the BEV to the HEV. Formally, let $\bm{\tau}=(0,\tau_H,\tau_B)$ and let $\tau_B=\tau_B(\tau_H)$ be implicitly defined by the government budget constraint $G\left(0,\tau_H,\tau_B \right)=\bar{G}$. 

In general, a budget-neutral subsidy reallocation from the BEV to the HEV reduces emissions if the marginal abatement return of an HEV subsidy, $\frac{-\partial E/\partial \tau_H}{\partial G/\partial\tau_H}$, exceeds that of a BEV subsidy; see Appendix \ref{app: prop 1 derive}. Under a baseline policy with no subsidies, this comparison amounts to the condition in the following proposition.


\begin{proposition}\label{p: per-unit subsidy}
Consider the baseline policy with no subsidies, $\bm{\tau} = (0,0,0)$. A subsidy for product $j \in \{H,B\}$, financed by a tax on product $k \in \{H,B\}$ with $k \neq j$, reduces emissions if and only if
\begin{equation}\label{eqn: main}
     -\frac{1}{s_j}\frac{\partial s_j}{ \partial p_j} (e_j^D-e_j)  > -\frac{1}{s_k}\frac{\partial s_k}{ \partial p_k} (e_k^D-e_k) . 
\end{equation}

\end{proposition}
\begin{proof}
See Appendix \ref{app: prop 1 derive}.
\end{proof}

Proposition \ref{p: per-unit subsidy} implies that subsidizing the HEV, while taxing the BEV to maintain a zero-net budget, may be more effective than the reverse policy. This occurs if the HEV subsidy generates sufficiently large demand responses, measured by the own-price semi-elasticity $(-\frac{\partial s_j}{\partial p_j}\frac{1}{s_j}$), and/or net diverted emissions ($e_H^D-e_H$) relative to the corresponding BEV subsidy. 

In the special case of symmetric logit demand, inequality \eqref{eqn: main} reduces to the simple condition $e_j-e_k<0$. For example, an HEV subsidy is more effective than a BEV subsidy whenever HEV emissions are lower than BEV emissions (as in the early phase of the energy transition). Under more flexible substitution patterns, an HEV subsidy may be more effective even when HEV emissions are higher than BEV emissions, provided that it generates more diversion away from ICEVs.

To further isolate the role of diversion under flexible substitution patterns, suppose that demand responses are identical across technologies, $-\frac{\partial s_j}{\partial p_j}\frac{1}{s_j} = -\frac{\partial s_k}{\partial p_k}\frac{1}{s_k}$. In this case, an HEV subsidy is more effective than a BEV subsidy if and only if $e_H - e_B<e_H^D - e_B^D$. This inequality holds when the emissions gap $e_H - e_B$ is sufficiently small or negative and/or when the HEV diverts relatively more demand from the ICEV than does the BEV.\footnote{To make this link more explicit, we derive an equivalent condition that is expressed directly in terms of the diversion ratios from the HEV and the BEV to the ICEV in Appendix \ref{app: diversion emission derive}.} In our empirical analysis, we show that these conditions are currently satisfied in the South Korean new passenger vehicle market.

\subsection{General framework}\label{sec: theory extension}

Proposition \ref{p: per-unit subsidy} imposes simplifying assumptions to highlight the key factors that determine the relative effectiveness of alternative subsidy policies. Here, we extend the framework along the four dimensions below.

First, we allow for multiple products within each fuel type, which generates scope for within-type substitution in response to subsidies. However, such substitution may contribute relatively little to emissions changes, because products of the same fuel type share similar emissions. Second, we include additional fuel-type categories, allowing for multiple ICE vehicle types such as gasoline, diesel, and LPG. Third, we allow for nontrivial subsidy pass-through under imperfect competition between multi-product price-setting firms. Fourth, we allow usage emissions to vary across consumers through heterogeneous annual mileage.

Accordingly, we extend the basic setup as follows. Products are partitioned into $G$ groups based on their fuel types, indexed by $g$. Let $\mathscr{J}_g$ denote the set of products with fuel type $g$, and let $S_g=\sum_{j\in\mathscr{J}_g} s_j$ denote the aggregate market share of fuel type $g$. Let $\tau_g$ denote the per-unit subsidy common to all products with fuel type $g$, and let $\Theta_{j,g} = - \frac{\partial p_j}{\partial \tau_g}$ denote the pass-through rate of subsidy $\tau_g$ to the price of product $j$.

To account for individual-specific emissions, $e_{ij}$, stemming from heterogeneous consumer mileage, we use the following notation. Let $D_{j\to k}^i = -\frac{\partial s_{ik} / \partial p_j}{\partial s_{ij} / \partial p_j}$ denote the individual diversion ratio, and let $e_{ij}^D = \sum_{k\neq j}D_{j\to k}^ie_{ik}$ denote the individual diverted emissions of product $j$. Finally, let $F(i)$ denote the consumer distribution, defined over observed demographics (e.g., income and mileage) as well as unobserved heterogeneity.

Product-level emissions and product-level diverted emissions are then defined as weighted averages of individual emissions and individual diverted emissions, respectively:
\begin{equation}\label{eqn: hetero emission}
    e_j = \int \omega_{ij} e_{ij} \, dF(i) 
    \quad \text{and} \quad
    e_j^D = \int \omega_{ij} e_{ij}^D \, dF(i),
\end{equation}
where the weights are given by
\begin{equation}\label{eqn: hetero weights}
    \omega_{ij} = \frac{\partial s_{ij}/\partial p_j}{\int \frac{\partial s_{ij}}{\partial p_j} \, dF(i)} = \frac{\partial s_{ij}/\partial p_j}{\partial s_j/\partial p_j}.
\end{equation}
Intuitively, consumers who are more responsive to changes in the price of product $j$, relative to the aggregate response, receive greater weight in the construction of these emission metrics; see Appendix~\ref{app: prop 2 derive} for the formal derivation. Note that the weights in \eqref{eqn: hetero weights}, which define product-level emissions and diverted emissions, closely resemble those defining product-level diversion ratios in \citet{conlon2021empirical}. This is because both are derived in the context of a price change for product $j$.\footnote{Furthermore, in the special case of homogeneous mileage, $e_{ij} = e_j$ for all $i$ and $j$, it can be verified that 
\begin{equation*}
    e_j^D = \sum_{k\neq j} \left(\int \omega_{ij}D_{j\to k}^i dF(i)\right) e_k =  \sum_{k \neq j} D_{j\to k} e_k.
\end{equation*} 
This implies that the aggregate diverted emissions of product $j$ are separable in the diversion ratios $D_{j\to k}$ and the emissions $e_k$ (even under consumer heterogeneity in other dimensions than mileage).}

This setup yields the following generalization of Proposition \ref{p: per-unit subsidy}:

\begin{proposition}\label{p: per-unit subsidy generalized}
Consider the baseline policy with no subsidies. A subsidy for fuel type $g_1 \in \{H,B\}$, financed by a tax on fuel type $g_2 \in \{H,B\}$ with $g_2 \neq g_1$, reduces emissions if and only if
\begin{equation}\label{eqn: multiproduct}
   -\frac{1}{S_{g_1}} \sum_{g=1}^G\sum_{j\in\mathscr{J}_g} \Theta_{j,g_1} \frac{\partial s_j}{\partial p_j}(e_j^D-e_j) > -\frac{1}{S_{g_2}} \sum_{g=1}^G \sum_{j\in\mathscr{J}_g} \Theta_{j,g_2} \frac{\partial s_j}{\partial p_j}(e_j^D-e_j).
\end{equation}
\end{proposition}
\begin{proof}

See Appendix~\ref{app: prop 2 derive}.
\end{proof}

Similar to condition~\eqref{eqn: main}, condition~\eqref{eqn: multiproduct} compares the marginal abatement returns of subsidies for two fuel types, $g_1$ and $g_2$. The marginal abatement return of a subsidy for fuel type $g$, $\frac{- \partial E / \partial \tau_g}{\partial G / \partial \tau_g}$, still depends on demand responses and net diverted emissions. However, these are now weighted by the pass-through rates within and across fuel types. One may expect high pass-through rates, $\Theta_{j,g}$, for products $j\in \mathscr{J}_{g}$ (within-fuel-type pass-through), and lower pass-through rates for products $j\notin \mathscr{J}_{g}$ (cross-fuel-type pass-through).\footnote{In the special case where own-fuel-type pass-through rates are all equal to one and cross-fuel-type pass-through rates are all equal to zero, condition~\eqref{eqn: multiproduct} simplifies to
\begin{equation*}
   -\frac{1}{S_{g_1}} \sum_{j \in \mathscr{J}_{g_1}} \frac{\partial s_j}{\partial p_j} (e_j^D - e_j) > -\frac{1}{S_{g_2}} \sum_{j \in \mathscr{J}_{g_2}} \frac{\partial s_j}{\partial p_j} (e_j^D - e_j).
\end{equation*}
In this case, the marginal abatement return of a subsidy for fuel type $g$ depends only on the demand responses and net diverted emissions within group $g$.} Furthermore, net diverted emissions ($e_j^D - e_j$) now explicitly account for heterogeneity in emissions, demand responses, and substitution patterns across consumers. If these sources of heterogeneity are systematically correlated, net diverted emissions may differ between the heterogeneous-mileage model and its homogeneous counterpart.\footnote{For example, if high-mileage consumers, who have greater scope for emission reductions (higher $e_{ij}^D-e_{ij}$), are more price-sensitive (higher $\frac{\partial s_{ij}}{\partial p_j}$) than low-mileage consumers, aggregate net-diverted emissions exceed those predicted by the homogeneous model.}

Proposition \ref{p: per-unit subsidy generalized} is based on a baseline policy of no subsidies ($\bm{\tau}=\mathbf{0}$) to simplify algebraic expressions. Extending the analysis to incorporate existing subsidies for HEVs or BEVs is straightforward and yields a generalization of condition~\eqref{eqn: multiproduct} that includes additional terms reflecting existing subsidy expenditures, as shown in Appendix~\ref{app: prop 2 derive}.

\subsection{Empirical implementation}

The empirical model implements the general framework discussed above and quantifies its implications. Two key empirical components are the measurement of emissions and the characterization of demand substitution patterns across fuel types.

On the emissions side, we draw on the environmental literature to measure lifecycle emissions as the sum of production and usage emissions. Usage emissions depend not only on vehicle characteristics but also on annual mileage, generating heterogeneity in emissions within fuel types and amplifying differences across fuel types.

On the demand side, capturing flexible substitution patterns requires accounting for heterogeneity in consumer characteristics, in particular annual mileage. High-mileage consumers are more likely to substitute away from ICEVs toward subsidized alternatives with lower operating costs (HEVs and BEVs), as they benefit more from lower fuel expenditures.

Finally, to highlight the key mechanisms, our theoretical framework has so far focused on comparing the \emph{marginal} effects of changing per-unit subsidies across fuel types. In practice, however, we compare the \emph{total} effects of alternative budget-neutral subsidy policies. Specifically, we ask which of the two fuel-specific policies (BEVs vs. HEVs) is more effective at reducing emissions. To answer this question, we now turn to our empirical analysis.\footnote{Appendix \ref{app: total condition} provides a theoretical comparison of the total effects of two subsidy policies.}

\section{Industry background}\label{sec: background}

\subsection{Fuel type segments in the South Korean car market}\label{sec: south korean car market}

Automobile ownership in South Korea exceeds one vehicle per household, reaching 1.13 in 2023. In that same year, nearly 1.5 million new passenger vehicles were sold, ranking South Korea among the largest automotive markets globally \citep{StatistaVehicleSales2024}. 

The South Korean market features three primary categories of conventional ICEVs: gasoline, diesel, and LPG. Additionally, several categories utilizing new technologies are available: BEVs, HEVs, plug-in hybrid vehicles (PHEVs), and hydrogen fuel cell vehicles (HFCVs). BEVs are powered entirely by an onboard rechargeable battery and do not contain an ICE. By contrast, HEVs combine a conventional ICE with an electric motor and a small battery pack to enhance fuel efficiency through three key mechanisms: regenerative braking, engine optimization, and a stop-start system (see Appendix \ref{sec: HEV and PHEV} for details). PHEVs can be viewed as an extension of HEVs: like HEVs, they combine an ICE with an electric drivetrain, but they also include a larger battery that can be plugged into an external power source (the grid). Finally, HFCVs generate electricity through a fuel cell stack that converts hydrogen gas into electrical energy.

Table~\ref{tab: sales 2023} reports sales volumes by fuel type for 2023. Gasoline-powered vehicles constituted the majority of newly registered passenger cars, capturing a 55\% sales share. Hybrids (including both HEVs and PHEVs) were the second most popular category, with a share exceeding 25\%. This segment was dominated by HEVs, as fewer than 10\,000 PHEVs were sold. While BEVs and diesel cars each accounted for approximately 7\% of total sales, sales of HFCVs totaled only around 4\,000 units. Additionally, the sales share of the Hyundai Motor Group, which comprises three brands (Hyundai, Kia, and Genesis), was approximately 75\%, while the remainder was distributed among other domestic manufacturers and imported brands.

\begin{table}[!t]
\small
  \centering
  \caption{New passenger vehicle sales in 2023}
    \begin{tabular}{lrrrrrrrrr}
    \hline
    \hline
          & \multicolumn{6}{c}{Fuel type}                 &  \bigstrut\\
\cline{2-7}    Manufacturer & \multicolumn{1}{c}{Gasoline} & \multicolumn{1}{c}{Diesel} & \multicolumn{1}{c}{LPG} & \multicolumn{1}{c}{Hybrid$^{+}$} & \multicolumn{1}{c}{BEV} & \multicolumn{1}{c}{HFCV} & \multicolumn{1}{c}{Total} \bigstrut\\
    \hline
    Hyundai Motor Group &   621\,548  &    83\,205  &   45\,451  &   278\,525  &        70\,359  &       4\,326  &   1\,103\,414  \bigstrut[t]\\
    Others  &   199\,162  &    21\,312  &    7\,710  &   112\,368  &        45\,180  &            -  &     385\,732  \\
    \hdashline[1pt/1pt]
    Total &   820\,710  &   104\,517  &   53\,161  &   390\,893  &       115\,756  &       4\,326  &   1\,489\,363  \bigstrut\\
    \hline
    \hline
    \multicolumn{10}{l}{\footnotesize $^{+}$ Hybrid refers to the combined sales of HEVs and PHEVs.}
    \end{tabular}%
  \label{tab: sales 2023}%
\tablenotes This table reports sales (i.e., the number of vehicles sold) by each fuel type in 2023. Source: Korea Automobile \& Mobility Association (KAMA).
\end{table}%

It is instructive to consider the evolution of sales shares by fuel type during our sample period (2012--2023); for further information, see Figure~\ref{fig: fuel share trend} in Appendix \ref{sec: Additional figures and tables}. Gasoline vehicles experienced an initial decline followed by a recovery, thereby maintaining their dominance. By contrast, diesel cars, which were favored until the mid-2010s because of their relatively high fuel economy, have seen their popularity plummet in recent years due to the Dieselgate scandal and increasing regulatory pressure from CAFE standards.\footnote{Since the initial implementation of the Corporate Average Fuel Economy (CAFE) standards in South Korea in 2012 (17~km/L for fuel economy and 140~g/km for GHG emissions), the requirements have been gradually tightened annually. By 2021, these standards had reached 24.3~km/L for fuel economy and 97~g/km for GHG emissions.} As manufacturers phased out diesel variants, the HEV sales share rose more than eightfold, from 2.7\% in 2015 to 23.1\% in 2023, whereas the BEV sales share stagnated at only 6.5\% in 2023, 0.3 percentage points lower than in the previous year. The growth in the HEV share and the stagnation of the BEV share occurred despite the removal of HEV subsidies and substantial subsidies for BEVs (discussed further below).

 \subsection{Greenhouse gas emissions of vehicles}\label{sec: greet}

The GREET model is a standard, comprehensive analytical tool used in the transportation sector to measure life-cycle GHG emissions, also referred to as cradle-to-grave emissions.\footnote{GREET stands for Greenhouse gases, Regulated Emissions, and Energy use in Technologies. The GREET model was developed by Argonne National Laboratory \citep{elgowainy2016cradle, kelly2023cradle}.} As Figure~\ref{fig: greet} illustrates, a vehicle's life-cycle GHG emissions consist of two components: vehicle-cycle emissions and fuel-cycle emissions, the latter also referred to as well-to-wheels emissions.

\begin{figure}[!t]
	\centering
	\caption{The life-cycle GHG emissions of a vehicle}
	\includegraphics[width=\textwidth]{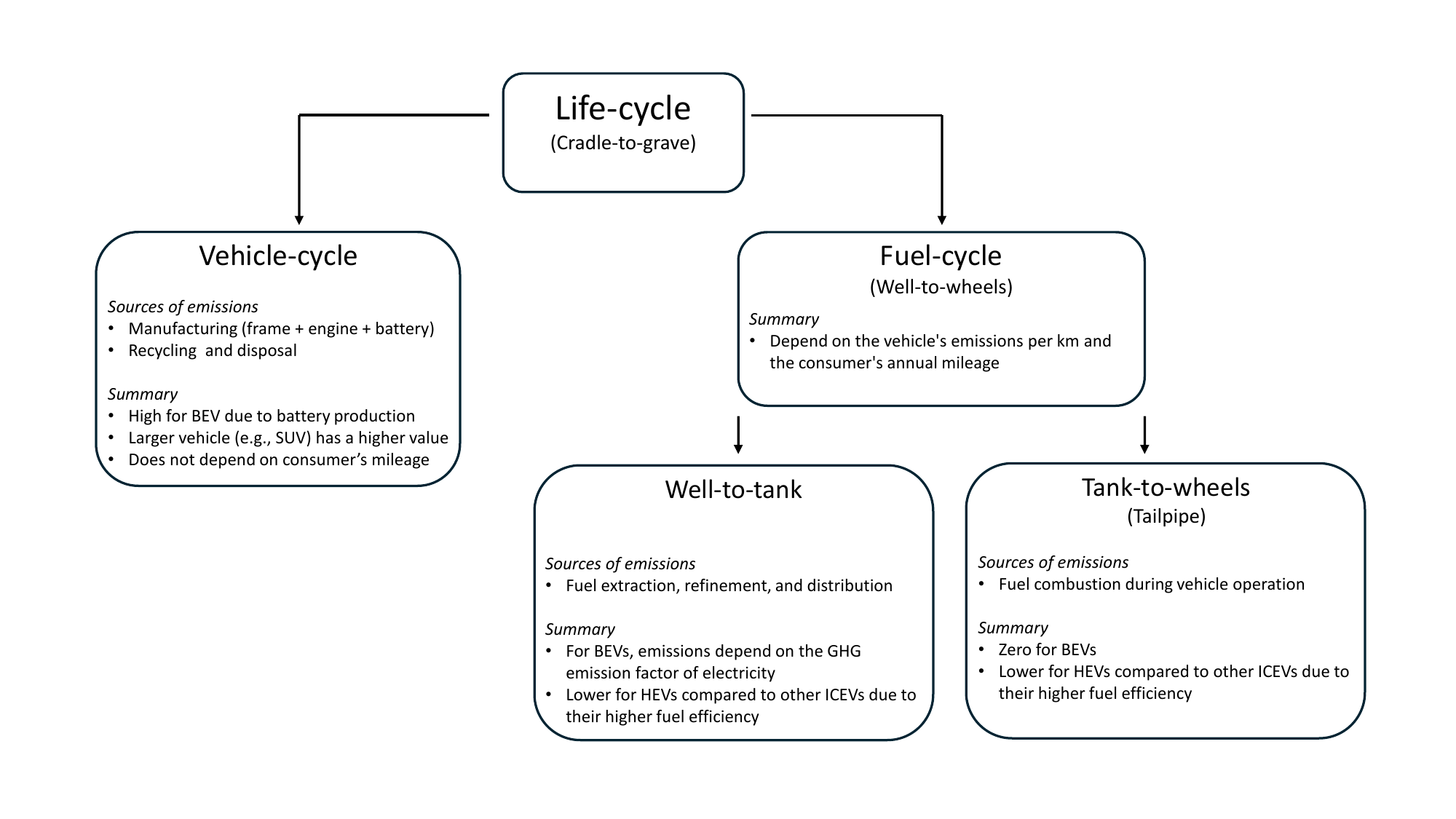}
\label{fig: greet}
\tablenotes Based on the GREET model in \citet{kelly2023cradle}, this figure decomposes the life-cycle GHG emissions of a vehicle into three parts and provides the sources of emissions for each part.
\end{figure}

The first component, vehicle-cycle GHG emissions, is generated during the production, disposal, and recycling of the vehicle and its battery. This component is independent of the vehicle owner's mileage. The second component, fuel-cycle GHG emissions, is associated with producing and consuming the fuel required to drive 1~km, multiplied by the vehicle owner's total mileage. The emissions per km can be further divided into well-to-tank emissions and tank-to-wheels (or tailpipe) emissions. The former are associated with the extraction, processing, transportation, refining, and distribution of fuels for ICEVs (including the more efficient HEVs), and with electricity generation and transmission for BEVs. The latter arise from actual fuel combustion during vehicle operation. Since BEVs generate no tailpipe emissions, their fuel-cycle emissions per km (in g~CO$_2$e/km) consist solely of well-to-tank emissions, calculated by dividing the emission factor at final consumption (in g~CO$_2$e/kWh) by the vehicle's fuel economy (in km/kWh).\footnote{The GHG emission factor at final consumption represents the GHG emissions per unit of electricity, accounting for energy losses incurred during transmission and distribution through the grid.} 

As an illustration, consider a 2020 Toyota Camry HEV whose vehicle-cycle emissions are 6.3 tonnes of CO$_2$e and fuel-cycle emissions per km are 130.14 g of CO$_2$e.\footnote{Appendix~\ref{sec: Greenhouse gas emissions of vehicles} describes in detail the vehicle-cycle and fuel-cycle emission calculations for each product-year combination in our sample.} Assuming a 15-year holding period and an annual mileage of 12\,000~km, the cumulative fuel-cycle emissions amount to 23.4~tonnes of CO$_2$e ($15 \times 12\,000 \times 130.14 \times 10^{-6}$). The total 15-year GHG emissions are therefore 29.7~tonnes of CO$_2$e ($6.3 + 23.4$).

 \subsection{BEV and HEV subsidy policy}\label{sec: subsidy background}

In accordance with the Paris Agreement, the South Korean government established its Nationally Determined Contribution (NDC) in June 2015. According to the ``2050 Carbon Neutrality Scenario'' and the upgraded NDC announced in October 2021, the government aims to reduce GHG emissions in the transportation sector by approximately 40\% by 2030. To achieve this, the government has been bolstering its support for BEVs.

More specifically, it directly installs electric vehicle chargers and subsidizes private firms to install them; in 2023, 300~billion KRW (260~million USD) was budgeted for this purpose.\footnote{The average exchange rate between 2012 and 2023, 0.86~USD per 1\,000~KRW, is used throughout the paper.} It also determines the national BEV purchase subsidies. The amount of these subsidies was initially more or less uniform across nameplates, but began to vary in 2018 based on vehicle attributes such as fuel economy (energy efficiency), maximum range, size, battery efficiency, and price. Subsidies also depend on whether automakers actively participate in installing charging stations and apply advanced, up-to-date technologies in their vehicles. For example, the 2023 formula for the national BEV purchase subsidy is given by:
\begin{equation*}
\begin{aligned}
	& \Big[ \text{attribute part } + \text{infrastructure part} + \text{innovation part} \Big]  \times \text{price multiplier},
\end{aligned}
\end{equation*}
where the price multiplier equals 1 for vehicles priced below 57~million KRW (49~thousand USD), 0.5 for those priced between 57~million and 85~million KRW (73~thousand USD), and 0 for vehicles priced above 85~million KRW. As a result, some BEVs---mostly German models---have not been eligible for purchase subsidies due to their high prices. In 2023, for example, 25 out of 36 BEV nameplates (e.g., Kia EV6 and Tesla Model~3) were eligible for subsidies, whereas the remaining 11 nameplates were not; see Table~\ref{tab: year sub nameplate} in Appendix \ref{sec: Additional figures and tables}.

Local governments provide additional purchase subsidies, which vary across provinces and years, generating cross-province variation in the total (national + local) subsidy per vehicle.\footnote{Provinces are the first-tier administrative divisions in South Korea. Figure~\ref{fig: province illustration} in Appendix \ref{sec: Additional figures and tables} illustrates how they are geographically defined.} Table~\ref{tab: year sub} summarizes the evolution of BEV subsidies. The per-vehicle subsidy has declined from almost 33~million KRW in 2012 to around 7~million in 2023; on average, the purchase subsidy covered at least 50\% of the vehicle price up until 2015, but 15\% or less since 2022.\footnote{Section \ref{sec: data} provides details on how we construct the BEV purchase subsidies at the year-province-nameplate level used in this study.} In contrast, total subsidy expenditure has increased rapidly due to the rising number of registered BEVs.

\begin{table}[!t]
\small
  \centering
  \caption{BEV purchase subsidies and price}
    \begin{tabular}{rrrrrr}
    \toprule
          & Total subsidy & Number of & Per-vehicle & Acquisition & Subsidy \\
          & value & BEVs sold & subsidy & price & share \\
    Year  & (Billion KRW) &  & (Million KRW) & (Million KRW) &     \\
    \midrule
    2012  & 0.13  &                       4  & 32.67 & 49.01 & 0.67 \\
    2013  & 3.16  &                    125  & 25.30 & 44.72 & 0.57 \\
    2014  & 11.66 &                    507  & 23.00 & 46.45 & 0.51 \\
    2015  & 45.54 &                  2\,042  & 22.30 & 45.96 & 0.50 \\
    2016  & 50.33 &                  2\,612  & 19.27 & 43.35 & 0.45 \\
    2017  & 140.53 &                  7\,149  & 19.66 & 42.38 & 0.49 \\
    2018  & 374.99 &                22\,183  & 16.90 & 46.64 & 0.38 \\
    2019  & 337.61 &                23\,163  & 14.58 & 49.72 & 0.31 \\
    2020  & 251.83 &                19\,100  & 13.19 & 55.07 & 0.26 \\
    2021  & 346.57 &                35\,684  & 9.71  & 58.83 & 0.18 \\
    2022  & 512.94 &                62\,172  & 8.25  & 59.72 & 0.15 \\
    2023  & 464.70 &                65\,851  & 7.06  & 58.06 & 0.14 \bigstrut[b] \\
    \hdashline[1pt/1pt]
    Total & 2\,540.00 &              240\,592  & 10.56 & 55.75 & 0.21 \bigstrut[t] \\
    \bottomrule
    \end{tabular}
  \label{tab: year sub}%
\tablenotes This table reports the total subsidy value, the number of BEVs registered, the sales-weighted average per-vehicle purchase subsidy, the sales-weighted average acquisition price, and the sales-weighted average subsidy share for each year between 2012 and 2023. All monetary values are in 2020 constant KRW (1\,000~KRW $\approx$ 0.86~USD). Source: official electric vehicle policy websites of the Korean government, official policy documents from the Ministry of Environment, and news articles for figures from the early 2010s.
\end{table}%

The Korean government introduced HEV purchase subsidies of 1 million KRW per vehicle in 2009 (less than 3\% of the purchase price), halved them in 2018, and discontinued them the following year (with PHEV subsidies ending in 2021). Consequently, subsidies for HEVs have been significantly lower than those for BEVs; the sales-weighted average per-vehicle subsidy for HEVs (170\,000 KRW or 147 USD) between 2012 and 2023 is negligible compared to that for BEVs (10.56 million KRW or 9\,100 USD).

In addition to direct purchase subsidies, buyers of BEVs and HEVs have also been eligible for tax credits since 2012, ranging from 1~million KRW (860~USD) to 4~million KRW (3\,440~USD), with BEVs receiving roughly two to three times the amount provided for HEVs. Tax credits for HEVs have been gradually reduced since 2020, as the Korean government plans to phase them out by 2026. In contrast, tax credits for BEVs are expected to remain relatively stable at least until the end of 2026.

\section{Data}\label{sec: data}

\subsection{Main dataset}

We combine datasets from multiple sources. The primary dataset is provided by ConsumerInsight, a leading automotive marketing and research firm in South Korea. For each of the 17 provinces between 2012 and 2023, it contains the annual number of registrations and sales revenues for each unique combination of \textit{brand}~$\times$~\textit{nameplate}~$\times$~\textit{model name}~$\times$~\textit{model year}~$\times$~\textit{fuel type}~$\times$~various car attributes (e.g., size, defined by the combination of width, height, and length; fuel economy; engine displacement; and seating capacity). One advantage of our data is that sales revenues account for manufacturer promotions at the point of sale as well as the costs of miscellaneous options selected by consumers. Therefore, the vehicle price, obtained by dividing sales revenue by the number of registrations and then deflating it using the CPI (2020~=~100), is closer to the actual transaction price than the manufacturer's suggested retail price (MSRP) often used in the literature.

We subtract the total subsidy from the price of a BEV or HEV to obtain the vehicle's net price. In doing so, we make two assumptions. First, we assume that all consumers of BEVs and HEVs applied for and received the purchase subsidies, although subsidies are reimbursed only to purchasers upon request. Second, we assume that the province- and year-specific cap on the total number of BEVs eligible for subsidies in each province and year is not binding. This assumption is reasonable, given that local governments typically relax or extend the cap during the year when the demand for subsidized vehicles exceeds the initially set cap. According to Table~\ref{tab: attributes by fuel type}, the average net price of BEVs is 45~million~KRW (approximately 40~thousand~USD), which is 22\% and 54\% higher than the average prices of HEVs and gasoline cars, respectively.

The sales data include six domestic brands (Hyundai, Kia, KG~Mobility, GM~Korea, Renault~Korea, and Genesis) and seven foreign brands (Mercedes-Benz, BMW, Volkswagen, Audi, Toyota, Lexus, and Tesla). Together, these 13~brands account for 93\% of total sales during the sample period. The data also report basic vehicle attributes such as fuel type, fuel economy, and body style, among others. We calculate fuel cost per km (1{,}000~KRW per~km) by dividing the province/year-level fuel price by fuel economy (km per~liter for ICEVs and km per~kWh for BEVs).\footnote{Fuel prices for ICEVs, disaggregated by province and year, are obtained from \href{https://www.opinet.co.kr/user/dopospdrg/dopOsPdrgAreaView.do}{OPINET} (Oil Price Information Network). The annual charging rates for BEVs (expressed in KRW per kWh) are sourced from official announcements by the Ministry of Environment \citep{ME_Press_2022}.} We supplement the dataset with additional vehicle specifications, including horsepower, curb weight, and driving range for BEVs, collected from two comprehensive South Korean automotive databases.\footnote{The two sources are \href{https://www.carisyou.com/car/}{CARISYOU} and \href{https://auto.danawa.com/}{Danawa}.} 

We define a product as a unique combination of \textit{brand}~$\times$~\textit{nameplate}~$\times$~\textit{fuel type}. For instance, Toyota (brand)–Camry (nameplate)–HEV (fuel type) represents a product. We also define a market as a unique combination of province and year. Accordingly, we aggregate the data to the product-market level and use average attributes, weighted by the number of registrations, as the representative product characteristics. According to the National Travel Survey conducted by the Ministry of Land, Infrastructure, and Transport in 2014,\footnote{Source: \url{https://www.ktdb.go.kr/www/selectBbsNttView.do?key=42&bbsNo=7&nttNo=2688}} the average replacement cycle for passenger cars in South Korea is around six years. Therefore, we assume that the market size associated with province~$r$ is one-sixth of the average annual number of registered passenger cars in that province during the sample period. The outside option includes buying used vehicles or continuing to drive currently owned vehicles.\footnote{Rather than interpreting public transportation as part of the outside option, we let its availability influence consumer mileage.}

From the data, we remove products with cumulative sales below 100 units during the sample period. We also exclude PHEVs and HFCVs, given their negligible market shares in South Korea. Finally, we omit observations from Sejong City (a special self-governing city established in July 2012) as reliable income distribution data are unavailable for this region. The final sample consists of 34\,629~observations from 366~products, among which 40~and~34 are BEVs and HEVs, respectively. 

Table~\ref{tab: attributes by fuel type} presents the sales-weighted average product attributes by fuel type. There is a clear negative correlation between acquisition prices and driving costs. BEVs are the most expensive, with an average acquisition price of 55.7 million KRW (48 thousand USD), yet they offer the lowest fuel cost at 50 KRW per km. HEVs are significantly cheaper than BEVs (by 33\% in terms of acquisition price and by 17.8\% in terms of net price), and exhibit the highest fuel economy (16.8 km/L) among ICEVs. Consequently, consumers may view HEVs as viable substitutes for BEVs. The table also shows that BEVs are the heaviest due to their batteries, whereas diesel cars are the largest since they are mostly SUVs.

\begin{table}[!t]
\small
  \centering
  \caption{Vehicle attributes by fuel type}
    \begin{tabular}{lrrrrrr}
    \toprule
          & \multicolumn{5}{c}{Fuel type}         &  \\
\cmidrule{2-6}    Attribute & \multicolumn{1}{l}{Gasoline} & \multicolumn{1}{l}{Diesel} & \multicolumn{1}{l}{LPG} & \multicolumn{1}{l}{HEV} & \multicolumn{1}{l}{BEV} & \multicolumn{1}{l}{All} \\
    \midrule
    Price (Mil. KRW) & 29.37 & 35.87 & 25.48 & 37.27 & 55.71 & 32.33 \\
    Net price (Mil. KRW) & 29.37 & 35.87 & 25.48 & 37.10 & 45.14 & 32.11 \\
    Fuel cost per km (1\,000 KRW/km) & 0.14  & 0.11  & 0.15  & 0.09  & 0.05  & 0.12 \\
    Fuel economy (km/L, km/kWh) & 12.53 & 13.27 & 9.28  & 16.84 & 5.16  &  \\
    Size ($m^3$) & 12.29 & 14.96 & 13.18 & 13.68 & 13.21 & 13.27 \\
    Power (hp) & 168.14 & 180.73 & 152.50 & 218.47 & 233.57 & 176.61 \\
    Weight (kg) & 1\,371.43 & 1\,775.42 & 1\,473.53 & 1\,628.90 & 1\,823.79 & 1\,528.04 \\
    Acceleration (hp/kg) & 0.12  & 0.10  & 0.10  & 0.13  & 0.13  & 0.11 \\
    GHG emissions  &       &       &       &       &       &  \\
    -Fuel-cycle (g CO$_2$e/km) & 199.87 & 212.88 & 204.45 & 134.85 & 89.21 & 196.60 \\
    -Vehicle-cycle (tonnes CO$_2$e) & 6.99  & 7.70  & 7.11  & 7.16  & 10.61 & 7.29 \bigstrut[b]\\
    \hdashline[1pt/1pt]
    Observations & 15\,551 & 11\,854 & 2\,096  & 3\,147  & 1\,981  & 34\,629 \bigstrut[t]\\
    \bottomrule
    \end{tabular}%
  \label{tab: attributes by fuel type}%
  \tablenotes This table reports sales-weighted average vehicle attributes by fuel type in our sample data (2012--2023, 16 provinces). Fuel economy is expressed in km/kWh for BEVs and km/L for all other fuel types. Fuel-cycle GHG emissions are measured on a per-kilometer basis, while vehicle-cycle GHG emissions are calculated for products sold during the 2020--2023 period.    
\end{table}%

\subsection{Supplementary datasets}\label{sec: supplementary data}

\paragraph{GHG emissions}

We construct data on fuel-cycle emissions per km for each product-year combination in our sample, as well as vehicle-cycle emissions for products sold between 2020 and 2023, by aggregating information from various sources. Details of these sources and a full description of the GHG emission calculations are provided in Appendix~\ref{sec: Greenhouse gas emissions of vehicles}. 

Table \ref{tab: attributes by fuel type} shows that sales-weighted average fuel-cycle emissions per km are lowest for BEVs (89 g CO$_2$e/km), as they generate no tailpipe emissions. HEVs also exhibit significantly lower fuel-cycle emissions per km (135 g CO$_2$e/km on average) than conventional ICEVs (200--213 g CO$_2$e/km) through their superior fuel economy, as explained in Section \ref{sec: greet}. In contrast, the vehicle-cycle emissions of BEVs are 37.7\%--51.7\% higher than those of HEVs and conventional ICEVs. This is largely because battery production is emissions-intensive and batteries have become heavier over time to support longer driving ranges.

Consequently, the emissions gap between HEVs and BEVs is relatively small given current technology. Assuming six years of ownership and an annual mileage of 12\,000~km (the average replacement cycle and annual mileage in South Korea, respectively), the emissions from an HEV (17~tonnes of $\text{CO}_2\text{e}$) are approximately the same as those from a BEV in our sample between 2020 and 2023. Extending the operating period to 15 years, the average passenger vehicle lifespan in South Korea, results in life-cycle emissions for BEVs of 26.7~tonnes of $\text{CO}_2\text{e}$, compared with 31.4~tonnes for HEVs, a difference of only 15\%.

\paragraph{BEV purchase subsidies}

We collect BEV purchase subsidies at the year–province–nameplate level since 2020 from the official electric vehicle policy website of the Korean government.\footnote{Source: \href{https://ev.or.kr/nportal/buySupprt/initSubsidyPaymentCheckAction.do}{Zero-Emission Vehicle Integrated Portal}} In some provinces, local subsidies vary across counties within the province, although such variation is not substantial. In these cases, we proxy the local subsidy using figures from the county where the majority of the population resides or where most BEVs are registered within the province. For years prior to 2020, we rely on official policy documents, when available, issued by the Ministry of Environment, as well as news articles. Figures~\ref{fig: subsidy dispersion} and \ref{fig: scatter localsub} in Appendix \ref{sec: Additional figures and tables} show substantial variation in total (national + local) BEV subsidies across nameplates and across provinces, respectively, for each year (2016--2023).

\paragraph{Second-choices}
Each of the two consumer surveys conducted by ConsumerInsight in 2018 and 2023 contains the automobile purchase decisions of respondents who purchased new vehicles within the two years preceding the survey. Importantly, the survey data record the second choice of each respondent, that is, the vehicle they would have purchased if the primary choice had not been available. After excluding respondents whose primary vehicles are not listed in our sales data or who did not report a second choice, we retain 4{,}035 and 4{,}145 consumers for 2018 and 2023, respectively. The two surveys oversample consumers who purchased vehicles with low market shares such as BEVs, while undersampling gasoline and diesel vehicle buyers. We use the inverse sampling weights for each survey respondent provided by ConsumerInsight to correct for sampling imbalances in the estimation process, as detailed in Appendix \ref{sec: micro moment description}.

\paragraph{Mileage information}

In South Korea, all registered passenger vehicles are required to undergo a regular inspection every two years. The inspection interval is four years for newly purchased vehicles and one year for vans, trucks, and similar vehicles. The Korea Transportation Safety Association (KTSA) keeps records of all vehicle inspections. From these publicly available data, we obtain the following information: (i) inspection year, (ii) annual mileage, (iii) fuel type, and (iv) province in which the vehicle is registered.

According to Figure~\ref{fig: mileage trend by fuel type}, HEVs tend to have higher mileage than vehicles with other fuel types, indicating that consumers who drive longer distances value the high fuel economy of HEVs. Interestingly, the average mileage of BEVs increased sharply during the sample period. Potential explanations include growth in battery capacity and the expansion of the electric charging network. The exceptionally high average mileage observed across all fuel types in 2021 is likely attributable to the increase in individual travel demand during the COVID-19 pandemic (2020 to 2021). Finally, note that people in metropolitan cities tend to drive less than residents of less densely populated provinces; see Table~\ref{tab: mileage by province and fuel type} in Appendix \ref{sec: Additional figures and tables}.

\begin{figure}[!t]
	\centering
	\caption{Mileage trends}
	\includegraphics[width=0.8\textwidth]{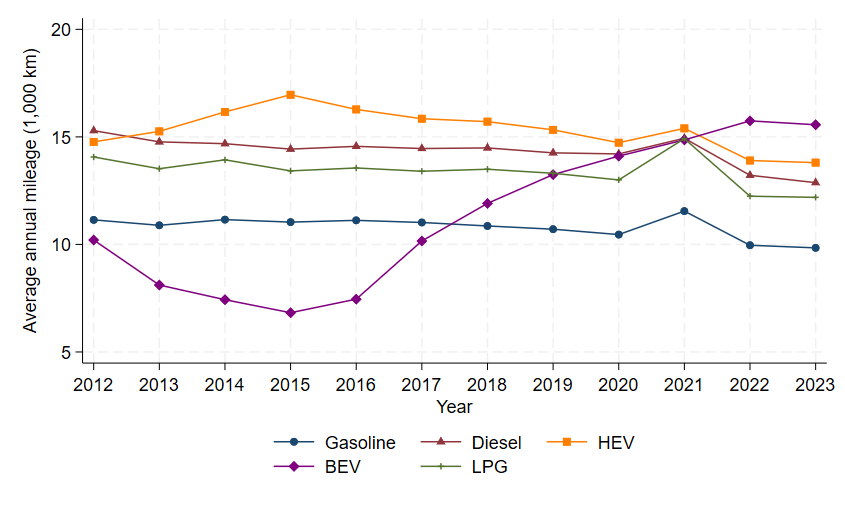}
\label{fig: mileage trend by fuel type}
\tablenotes This figure presents the average annual mileage trend for each fuel type.
\end{figure}

\paragraph{Additional datasets}

We obtain household income data from the Korean Labor and Income Panel Study (KLIPS)  to estimate the corresponding log-normal income distribution for each market (province--year). We then simulate 1\,000 individuals per market from these estimated distributions. We also obtain a charger dataset from the Korean Environment Corporation. Following \cite{springel2021network}, we define the interaction between the log of the number of chargers in a market and a BEV dummy as the infrastructure network variable. A full description of these datasets is provided in Appendix~\ref{sec: Data description}.

\section{Demand estimation}\label{sec: model and empiric}

\subsection{Model for consumer demand}

Recall that the market $m$ refers to the province and year. The indirect utility of consumer $i$ in market $m$ from purchasing product $j$ is given by
\begin{equation}\label{eqn: indirect utility}
	u_{ijm} = x_{jm}\beta_i - \alpha_i p_{jm} + \xi_{jm} + \varepsilon_{ijm},
\end{equation}
where $p_{jm}$ denotes the price net of subsidies, and $x_{jm}$ is a vector of product attributes, including cost per km, acceleration, size, infrastructure network, market entry and exit indicators, and fixed effects for nameplates, provinces, and fuel type$\times$years. The latter account for time-varying, fuel-type-specific demand shocks, such as the Dieselgate scandal or non-subsidy benefits for HEVs and BEVs (e.g., parking fee discounts and designated parking zones). Remaining unobserved factors are captured by the mean-zero structural error term $\xi_{jm}$, while $\varepsilon_{ijm}$ represents an i.i.d.\ extreme value distributed idiosyncratic taste shock.

We specify the individual-specific price coefficient and valuation of attribute $k$ as:
\begin{equation}\label{eqn: indirect utility2}
    \begin{aligned}
     & \alpha_i = \frac{\alpha}{y_i}, \\
     & \beta_{ik} = \beta_k + \pi_k \,\text{mileage}_i + \sigma_k v_{ik}, 
    \end{aligned}	
\end{equation}
where $y_i$ and $\text{mileage}_i$ denote consumer $i$'s annual income and daily mileage, respectively, and $v_{ik}$ captures remaining unobserved heterogeneity in the valuation of attribute $k$. We allow $\pi_k$ to be non-zero for the four main fuel-type dummy variables: diesel, gasoline, BEV, and HEV. Our specification therefore captures how tastes for these four fuel types compared to the outside option vary with mileage. We also allow $\sigma_k$ to be non-zero for the following variables: the four fuel-type dummy variables, the SUV indicator, and driving cost per km. 

Suppressing the market subscript $m$ for brevity, the indirect utility in \eqref{eqn: indirect utility} can then be decomposed into three terms:
\begin{equation}\label{eqn: indirect utility3}
    \begin{aligned}
     & u_{ij}     = \delta_{j} + \mu_{ij} + \varepsilon_{ij}, \\
     & \delta_{j} = x_{j}\beta + \xi_{j}, \\
     & \mu_{ij}   = -\alpha\frac{p_j}{y_i} + \sum_{k} \left( \pi_k \, \text{mileage}_{i} + \sigma_k v_{ik} \right) x^{(2)}_{jk},
    \end{aligned}	
\end{equation}
where $x_{jk}^{(2)}$ denotes the vector of attributes for which we specify heterogeneity.

We normalize the mean utility of the outside good to zero ($\delta_0=0$) and define $\mathcal{D}$ as a vector of demographic variables (income and mileage). Assuming independence among $\mathcal{D}$, $v$, and $\varepsilon$, consumer $i$'s choice probability for product $j$ ($\mathcal{s}_{ij}$) and the market share of product $j$ ($\mathcal{s}_{j}$) are
\begin{equation}\label{eqn: market share}
\begin{aligned}
\mathcal{s}_{ij} &= \frac{\exp \left(\delta_{j} + \mu_{ij} \right)} {1 + \sum_{\ell \in \mathscr{J}} \exp \left(\delta_{\ell} + \mu_{i\ell} \right)} \\
\mathcal{s}_{j} &= \int_{\mathcal{D}} \int_{v} \mathcal{s}_{ij} \, dF(v) \, dF(\mathcal{D}),
\end{aligned}
\end{equation}
where $\mathscr{J}$ denotes the set of available inside goods, $F(v)$ represents the (normal) distribution of taste shocks ($v$), and $F(\mathcal{D})$ denotes the empirical distribution function of demographic characteristics ($\mathcal{D}$). Because the empirical distributions of income and mileage come from separate sources, we further assume independence between these two demographic variables.

\subsection{Estimation and identification}

Given the vector of nonlinear utility parameters $\theta = (\alpha, \{\pi_k\}, \{\sigma_k\})$, there is a unique vector of mean utilities, $\delta(\theta)$, that equates the observed market shares with the simulated ones \citep{Berry94, Berry95}:
\begin{equation*}
{s}_{j} = \frac{1}{N}\sum_{i=1}^{N} \frac{\exp \left(\delta_{j} + \mu_{ij}(\theta)\right)} {1+\sum_{\ell \in \mathcal{J}} \exp \left(\delta_{\ell} + \mu_{i\ell}(\theta) \right)}, ~\forall j \in \mathscr{J}
\end{equation*}
where we set $N=1\,000$, i.e., 1\,000 draws from the income distribution, 1\,000 draws from the mileage distribution, and 1\,000 scrambled Halton draws in each market. By solving the equations using a fixed point procedure \citep{Berry95}, we obtain the structural error term as a function of the utility parameters: $\xi_j (\theta) = \delta_j(\theta) - x_j \hat{\beta}(\theta)$, where $\hat{\beta}(\theta)$ is recovered from a linear GMM regression of $\delta_j(\theta)$ on $x_j$.  

Our generalized method of moments (GMM) estimator is derived from two sets of sample moments. First, we construct a set of aggregate moments. Including the market subscript again, assume that the structural error term is mean independent of a vector of instrumental variables $Z_{jm}$ for product $j$ in market $m$ to obtain the following moments:
\begin{equation}
  G_1(\theta) = \frac{1}{N_p} \sum_{j,m} \xi_{jm}(\theta) Z_{jm},
\end{equation}
where $N_p$ is the number of products across markets. In addition to the vector of (non-price) product attributes, we include two cost shifters as instruments for price: (i) the import tariff rate and (ii) the raw material price interacted with vehicle weight. Due to the free trade agreement with the European Union in the early 2010s, import tariffs on German cars progressively declined from 2010 to 2016 with variation by fuel type and engine displacement. Tariffs on Japanese HEVs also declined slightly in the early 2020s. Based on market reports indicating that the production of a car requires approximately five times as much iron as aluminum, we calculate the second instrument as $(0.2 \times \text{aluminum price} + \text{iron ore price}) \times \text{weight}$. Because the transmission of a cost shock takes time, we use one-year lagged tariff rates and raw material prices. 

Second, for the identification of mileage and fuel-type interaction parameters and the standard deviations of tastes, $(\{\pi_k\}, \{\sigma_k\})$, we construct micro moments, $G_2(\theta)$ (e.g., \citealp{berry2004differentiated, Petrin02, conlon2025incorporating, grieco2024evolution}), using ConsumerInsight's 2018 and 2023 consumer surveys and KTSA's vehicle inspection data. More specifically, we compute the correlation between the product attributes of the first and second choices from the survey data and match it to the corresponding moment from the model. We also compute the average daily mileage of each fuel type from the vehicle inspection data and match it to the corresponding model-predicted moment. Appendix~\ref{sec: micro moment description} provides a full description of the construction of the micro moments using our micro datasets.

We form the objective function by stacking the sample moments. The GMM estimator is  
\begin{equation}\label{eqn: GMM}
    \hat{\theta} = \argmin_{\theta} G(\theta)' W G(\theta),
\end{equation}
where
\begin{equation}
   G(\theta) =
\begin{bmatrix}
G_1(\theta)  \\
G_2(\theta)
\end{bmatrix},
\quad
W =
\begin{bmatrix}
W_1 & 0 \\
0 & W_2
\end{bmatrix},
\end{equation}
and $W_1$ and $W_2$ are the weighting matrices corresponding to the first and second sets of moment conditions, respectively. We carry out the estimation using \texttt{PyBLP} \citep{conlon2020best}.

\subsection{Results}

Table~\ref{tab: rc demand estimation results} presents the demand parameter estimates. We first consider a simple logit demand model with homogeneous valuations for product attributes, $\beta_i = \beta$ and $\alpha_i = \alpha / \overline{\text{y}}$, where $\overline{\text{y}}$ is the average income in the market. The first column is based on an OLS regression, while the second column uses our two cost shifters as instruments for price (i.e., the import tariff rate and the raw material price interacted with vehicle weight).\footnote{The first-stage regression results, reported in Table~\ref{tab: first stage results} in Appendix \ref{sec: Additional figures and tables}, show that the instrumental variables are jointly significant and have the expected signs.} As expected, the IV estimate of the price coefficient $\alpha$ is larger than the OLS estimate (in absolute value). OLS implies inelastic product-level own-price elasticities, whereas the IV estimates imply plausible median and average elasticities of, respectively, -5.5 and -4.3. The coefficients of other attributes are mostly significant and of the expected signs.

To appreciate the role of consumer heterogeneity in the valuation of fuel types, we also estimate a nested logit demand model, as in \citet{Berry94}. We partition the products in six fuel type groups: gasoline, diesel, LPG, HEV, BEV, and the outside option. The model includes a nesting parameter $\rho$, capturing taste correlation between products of the same group. We estimate the model using the same cost shifters as price instruments, and the number of products per group as an instrument for the nesting parameter. We estimate a statistically significant nesting parameter of 0.677, indicating that consumers perceive vehicles of the same fuel type as much closer substitutes than vehicles of different fuel types. The median and average own-price elasticities are comparable with those of the IV logit model. For more details, see Table~\ref{tab: results of restricted models} in Appendix \ref{sec: Additional figures and tables}.

The next columns of Table~\ref{tab: rc demand estimation results} show the results for two random coefficients logit specifications, after incorporating micro moments into our GMM estimation. The third and fourth columns assume that mileage does not affect consumers' valuation of fuel type and set $\pi_k = 0$. The estimates of the price coefficient and the mean coefficients ($\alpha, \beta$) are consistent with the benchmark results reported in the first two columns. Moreover, the estimates of the standard deviations ($\{\sigma_k\}$) are all statistically significant, indicating that a consumer with a higher taste draw for fuel type $k$ ($v_{ik}$) places a higher value on vehicles of that fuel type relative to others. This finding is consistent with the statistically significant nesting parameter, and suggests that vehicles with the same fuel type are closer substitutes for one another.

\begin{table}[!t]
  \small
  \centering
  \caption{Estimation results of demand models}
    \begin{tabular}{lccccccrccr}
    \toprule
          & OLS logit &       & IV logit &       & \multicolumn{2}{c}{RC logit 1} &       & \multicolumn{3}{c}{RC logit 2} \bigstrut\\
\cmidrule{2-2}\cmidrule{4-4}\cmidrule{6-7}\cmidrule{9-11}    Variable & $\alpha,~\beta$ &       & $\alpha,~\beta$ &       & $\alpha,~\beta$ & \multicolumn{1}{c}{$\sigma$} &       & $\alpha,~\beta$ & \multicolumn{1}{c}{$\pi$} & \multicolumn{1}{c}{$\sigma$} \\
    \midrule
    Price / income & -0.406 &       & -6.013 &       & -12.864  &       &       & -14.569  &       &  \\
          & (0.062) &       & (0.739) &       & (2.299) &       &       & (2.323) &       &  \\
    Cost per km & -2.522 &       & -0.668 &       & -3.162  & 1.557  &       & -3.041  &       & \multicolumn{1}{c}{1.453 } \\
          & (0.090) &       & (0.227) &       & (0.388) & (0.331) &       & (0.357) &       & \multicolumn{1}{c}{(0.341)} \\
    Acceleration & -0.271 &       & 18.033 &       & 10.636  &       &       & 11.460  &       &  \\
          & (0.461) &       & (2.659) &       & (1.507) &       &       & (1.488) &       &  \\
    Size  & 0.324 &       & 0.859 &       & 1.082  &       &       & 1.147  &       &  \\
          & (0.048) &       & (0.134) &       & (0.115) &       &       & (0.118) &       &  \\
    Network & 0.032 &       & 0.248 &       & 0.362  &       &       & 0.213  &       &  \\
          & (0.069) &       & (0.069) &       & (0.178) &       &       & (0.171) &       &  \\
    Entry year & -0.623 &       & -0.457 &       & -0.367  &       &       & -0.345  &       &  \\
          & (0.032) &       & (0.041) &       & (0.051) &       &       & (0.052) &       &  \\
    Exit year & -1.707 &       & -1.720 &       & -1.975  &       &       & -2.043  &       &  \\
          & (0.034) &       & (0.041) &       & (0.059) &       &       & (0.061) &       &  \\
    BEV   &       &       &       &       &       & 4.992  &       &       & 7.423  & \multicolumn{1}{c}{5.025 } \\
          &       &       &       &       &       & (0.379) &       &       & (0.450) & \multicolumn{1}{c}{(0.376)} \\
    HEV   &       &       &       &       &       & 1.811  &       &       & 2.681 & \multicolumn{1}{c}{1.778 } \\
          &       &       &       &       &       & (0.114) &       &       & (0.108) & \multicolumn{1}{c}{(0.118)} \\
    Diesel &       &       &       &       &       & 1.663  &       &       & 2.495 & \multicolumn{1}{c}{1.722 } \\
          &       &       &       &       &       & (0.125) &       &       & (0.128) & \multicolumn{1}{c}{(0.121)} \\
    Gasoline &       &       &       &       &       & 1.830  &       &       & -3.513 & \multicolumn{1}{c}{1.568 } \\
          &       &       &       &       &       & (0.082) &       &       & (0.118) & \multicolumn{1}{c}{(0.087)} \\
    SUV   &       &       &       &       &       & 3.514  &       &       &       & \multicolumn{1}{c}{3.521 } \\
          &       &       &       &       &       & (0.105) &       &       &       & \multicolumn{1}{c}{(0.105)} \\
    \midrule
    \multicolumn{9}{l}{Own price elasticity}                              &       &  \\
    \hspace{0.1in}Median & -0.371 &       & -5.490 &       & \multicolumn{2}{c}{-5.588} &       &       & -6.125 &  \\
    \hspace{0.1in}Weighted average & -0.291 &       & -4.302 &       & \multicolumn{2}{c}{-5.377} &       &       & -5.906 &  \\
    \bottomrule
    \end{tabular}%
  \label{tab: rc demand estimation results}%
   \tablenotes This table presents the estimation results for our demand model. The utility specification includes nameplate, province, and year-interacted fuel-type fixed effects. The first two columns report the results from the logit model, where income refers to the average market income ($\bar{y}_m$). The next two columns report random-coefficients estimation results under the assumption that mileage does not affect consumer valuations of fuel types (RC logit 1). The last three columns report the results for the full model (RC logit 2), which allows for heterogeneous valuations of fuel types by mileage. In both random-coefficients specifications, income refers to an individual income draw ($y_i$). Micro-moments are incorporated into the GMM estimation for the random-coefficients specifications. Robust standard errors, clustered by market, are reported in parentheses.
\end{table}%

The last three columns present the results from the full model, which adds the interaction between mileage and each fuel type ($\{\pi_k\}$). This addition does not substantially affect the parameter estimates of the more restrictive specification in the previous two columns ($\alpha$, $\beta$, and $\{\sigma_k\}$).\footnote{Table~\ref{tab: micromoment values} in Appendix~\ref{sec: micro moment description} lists the observed and simulated values of the micro moments.} The magnitudes and significance of the interaction parameters reveal that consumers' fuel-type valuations depend on driving distance: the higher the mileage, the higher the value placed on the more energy-efficient HEVs, diesel vehicles, and especially BEVs. In Section \ref{sec: environmental evaluation}, we show that failing to account for this consumer mileage heterogeneity leads to an underestimation of the emission reduction effects of a vehicle subsidy policy.

The sales-weighted average own-price elasticities are $-5.38$ and $-5.91$ for the specifications without and with mileage heterogeneity, respectively, which are larger than the earlier reported average of $-4.3$ for the restricted logit model. These estimates are broadly comparable to those reported in previous studies, for example, $-5.32$ for the Korean market during 1996–2009 \citep{ohashi2017effects} and $-5.06$ for the U.S. market during 1980–2018 \citep{grieco2024evolution}. A breakdown of the average own-price elasticities by fuel type does not suggest noticeable differences. Cross-price elasticities for HEVs seem to be stronger with respect to gasoline vehicles than with respect to BEVs. See Table \ref{tab: elasticities by group} in Appendix \ref{sec: Additional figures and tables} for further information on own- and cross-price elasticities by fuel type.

To examine the extent of substitution across different fuel types more systematically, we calculate diversion ratios at the fuel-type level, reported in Table \ref{tab: diversion} for 2023. These measure the fraction of a product's lost sales after a price increase that is diverted to other products with the same or different fuel type. Not surprisingly, the diversion ratios are larger for products of the same fuel type. More interestingly, across fuel types, HEVs exhibit higher substitutability with gasoline and diesel vehicles than with BEVs. For example, following an HEV price increase, gasoline vehicles capture on average 28.8\% of the lost market share of HEVs, while BEVs capture only 2.9\%. Similarly, when the price of a gasoline vehicle increases, HEVs capture 9.8\% of the lost sales, while diversion to BEVs is limited to only 1.7\%. This stronger substitutability between HEVs and ICEVs than between BEVs and ICEVs is observed throughout the sample period, as shown in Figure \ref{fig: diversion by fuel type} in Appendix \ref{sec: Additional figures and tables}.

\begin{table}[!t]
  \centering
  \caption{Diversion ratios in 2023}
    \begin{tabular}{lcccccc}
    \hline
    \hline
          & BEV   & HEV   & Diesel & Gasoline & LPG   & Outside \bigstrut\\
    \hline
    BEV   & 0.574 & 0.101 & 0.038 & 0.168 & 0.009 & 0.110 \bigstrut[t]\\
    HEV   & 0.029 & 0.432 & 0.061 & 0.288 & 0.020 & 0.171 \\
    Diesel & 0.034 & 0.194 & 0.290 & 0.340 & 0.026 & 0.116 \\
    Gasoline & 0.017 & 0.098 & 0.038 & 0.639 & 0.015 & 0.193 \\
    LPG   & 0.023 & 0.178 & 0.078 & 0.442 & 0.043 & 0.236 \bigstrut[b]\\
    \hline
    \hline
    \end{tabular}%
  \label{tab: diversion}%
  \tablenotes 
  This table presents the diversion ratios from the row fuel type to the column fuel type in 2023. In a given market, the diversion ratio from product $j$ to fuel type $G$ is defined as $D_{j\to G} \equiv \sum_{k \in \mathcal{J}_{G}\setminus\{j\}} D_{j\to k} \equiv \sum_{k \in \mathcal{J}_{G}\setminus\{j\}} \frac{\partial s_k / \partial p_j}{-\partial s_j / \partial p_j}$, where $\mathcal{J}_{G}$ denotes the set of products with fuel type $G$. The sales-weighted average diversion ratio from fuel type $G_1$ to $G_2$ is then $D_{G_1\to G_2} \equiv \sum_{j \in \mathcal{J}_{G_1}} w_j D_{j\to G_2}$, where $w_j$ is product $j$'s sales share within fuel type $G_1$. Finally, we average $D_{G_1\to G_2}$ across the 16 markets in 2023 to obtain $\bar{D}_{G_1 \to G_2}.$

  \end{table}%

The full demand model suggests that mileage heterogeneity is an important determinant of consumer choice and substitution across fuel types. To further explore the role of mileage heterogeneity, we predict each fuel type's market share in 2023 for consumers in each of five average daily mileage intervals: 0--20, 20--40, 40--60, 60--80, and 80--100~km. We find that the market shares of vehicles with lower per-km fuel costs (HEVs, BEVs, and diesel cars) tend to be higher among consumers with greater mileage. For instance, the BEV market share doubles from 2.22\% to 4.44\% as daily mileage increases from the 20--40~km range to the 60--80~km range. In contrast, the market share of gasoline vehicles, which typically have higher per-km fuel costs, declines sharply in consumer mileage. See Figure~\ref{fig: mean_ci_2023} in Appendix~\ref{sec: Additional figures and tables} for further information.


\section{Comparing BEV and HEV subsidies}\label{sec: environmental evaluation}

In this section, we compare the emission-reducing benefits of the BEV and HEV subsidy policies over the final four years of the sample period (2020--2023). As detailed in Section \ref{sec: subsidy background}, HEV subsidies were discontinued in 2019; consequently, only BEVs were subsidized during this period.

\subsection{Counterfactual approach}\label{sec: counterfactual_design}

\paragraph{Supply and equilibrium}

To evaluate the impact of subsidies on emissions, we integrate our demand model with a supply-side model of multi-product Bertrand competition. In Section \ref{sec: theory}, we defined the net consumer price of product $j$ as $p_j = p_j^{G} - \tau_j$, i.e., the difference between the gross price and the subsidy. Let $\mathbf{p}$ be the vector of net consumer prices for all products in the market. Then, the variable profit of firm $f$ selling a set of products $\mathscr{J}_f$ is given by
\begin{equation*}
    \Pi_f = \sum_{j \in \mathscr{J}_f} \big(p_j^{G} - mc_j\big) \, s_j(\mathbf{p}) \cdot M,
\end{equation*}
where $mc_j$ denotes the marginal cost of product $j$, $s_j(\mathbf{p})$ is the product’s market share, specified earlier by equation \eqref{eqn: market share}, and $M$ is the market size. The first-order conditions for profit maximization are given in vector notation by
\begin{equation}\label{eqn: foc}
    \mathbf{p}^{G} = \mathbf{mc} - \left(\Omega \circ \nabla_p^s(\mathbf{p})\right)^{-1} \mathbf{s}(\mathbf{p}),
\end{equation}
where $\mathbf{p}^{G} = \mathbf{p} + \bm{\tau}$, $\Omega$ is the $J \times J$ (block-diagonal) ownership matrix with elements equal to 1 if product pairs belong to the same firm and zero otherwise, the operator $\circ$ denotes element-by-element multiplication, and $\nabla_p^s$ is the $J \times J$ matrix of partial derivatives of market shares with respect to prices.\footnote{Each firm is defined at the parent company level. That is, we have four domestic firms (Hyundai Motor Group, KG Mobility, GM Korea, and Renault Korea) and five foreign firms (Mercedes-Benz Group, BMW, Volkswagen Group, Toyota Group, and Tesla).}

We first use the system \eqref{eqn: foc} to recover the (constant) marginal costs at the observed price vector. We subsequently use it to compute equilibrium prices under counterfactual subsidy scenarios. We obtain the equilibrium price vector using a root solver (or fixed-point algorithm), which is then used to compute counterfactual market shares, sales, and the implied total emissions.

We consider two counterfactual subsidy scenarios, relative to the current status quo scenario of BEV subsidies: (i) no subsidy to either BEVs or HEVs and (ii) a reallocation of the total BEV subsidy budget toward uniform subsidies to HEVs. We focus on a reallocation of the direct BEV purchase subsidy budget, excluding additional BEV incentives through infrastructure development and tax credits. Given that existing BEV subsidies already incorporate some degree of optimization, our estimates may serve as a lower bound on the relative effectiveness of HEV subsidies.

\paragraph{Total GHG emissions}

We compute total GHG emissions under the status quo and the two counterfactual policy scenarios. We account for both fuel-cycle emissions (relating to usage) and vehicle-cycle emissions (relating to production, disposal and recycling). We first calculate consumer $i$'s emissions and then aggregate over consumers to obtain total emissions.

Consumer $i$'s expected emissions under policy scenario $\bm{\tau}$, $E_i^{\tau}$, are given by the weighted average of emissions across all products and the outside option, using $\mathbf{s}^{\tau}_{i}$, consumer $i$'s choice-probability vector under policy $\bm{\tau}$, as weights: 
\begin{equation}\label{eqn: consumer/product emissions}
E_i^{\tau} = \sum_{j \in \mathscr{J}} \left( \underbrace{ T \cdot m_i \cdot e^{FC}_j + e^{VC}_j }_{e_{ij}} \right) s^{\tau}_{ij} + e_{i0} s^{\tau}_{i0},
\end{equation}
where $m_i$ is consumer $i$'s annual driving distance (in km) and $T$ is the expected vehicle operating period.

The variable $e^{FC}_j$ denotes vehicle $j$'s fuel-cycle emissions associated with the energy use required for driving 1~km (in g CO$_2$e/km). The term $T \cdot m_i \cdot e^{FC}_j$ represents vehicle $j$'s cumulative fuel-cycle emissions for consumer $i$ over the ownership period. The variable $e^{VC}_j$ denotes $j$'s vehicle-cycle emissions generated during production, disposal, and recycling (in g CO$_2$e). The sum of these two terms constitutes product $j$’s GHG emissions for consumer $i$, denoted as $e_{ij}$.

The variable $e_{i0}$ represents emissions from the outside option, i.e., secondhand vehicle emissions. Given the lack of information on each consumer’s specific outside option, we set the fuel-cycle emissions per km of the outside option, denoted by $e^{FC}_{0}$, equal to the sales-weighted average fuel-cycle GHG emissions per km in the respective province over the preceding eight years, excluding the most recent three years (reflecting the assumption that consumers do not replace their vehicles within the first three years of ownership). Furthermore, we set the vehicle-cycle emissions of the outside option equal to zero. Assuming a remaining life span of $T^O$ years, emissions from the outside option are given by $e_{i0}=T^O\cdot m_i \cdot e^{FC}_{0}$.

We measure fuel-cycle and vehicle-cycle emissions for each product-year combination following the approach outlined in Section \ref{sec: supplementary data}. In our baseline scenario, we set the expected vehicle operating period for new and secondhand cars equal to 15 years. This is based on the average vehicle life span for passenger vehicles in South Korea and in line with standard assumptions in the literature \citep[e.g.,][]{allcott2024effects}. We also conduct a sensitivity analysis by considering lower expected operating periods.

Expected total emissions in market $m$ under policy $\bm{\tau}$ are obtained by integrating individual consumers' emissions $E^{\tau}_i$ over $\mathcal{D}$ and $v$:
\begin{equation}\label{eqn: mkt emissions}
E^{\tau}_m = M_m \int_{\mathcal{D}} \int_{v}E^{\tau}_i \, dF(v) \, dF(\mathcal{D}),
\end{equation}
where $M_m$ is the market size. We approximate the expected market-level emissions \eqref{eqn: mkt emissions} based on our $N=1\,000$ draws from the income and mileage distribution ($\mathcal{D}$), and Halton draws to approximate the normal distribution ($v$). 

We then calculate changes in emissions, $\Delta E^{\tau}_m$, following two distinct policy transitions: (i) from no subsidies to the current BEV subsidies, and (ii) from no subsidies to HEV subsidies. We also report the changes in their underlying components: fuel-cycle and vehicle-cycle emissions associated with the inside goods $j \in \mathscr{J}$, and outside emissions (i.e., fuel-cycle emissions from the outside good). See Appendix \ref{sec: Calculation of GHG emissions changes} for details.

\subsection{Budget-neutral subsidy reallocation}\label{sec: counterfactual 1}

\paragraph{Marginal abatement returns}\label{sec: Empirical test of Proposition 2}

Before quantifying the emission effects of the policy transitions in our equilibrium model, it is instructive to evaluate condition \eqref{eqn: multiproduct} in Proposition \ref{p: per-unit subsidy generalized}, under which subsidizing HEVs at the expense of taxing BEVs reduces emissions. Recall that this condition compares the marginal abatement returns across fuel types under zero government spending, which depend on demand responses and net diverted emissions, weighted by equilibrium subsidy pass-through rates and incorporating heterogeneous usage emissions.

Table~\ref{tab: prop 2} presents the marginal abatement returns for subsidies to each fuel type, evaluated at counterfactual prices under the no-subsidy scenario. We first derive the subsidy pass-through matrix under our Bertrand--Nash setup (see Appendix \ref{app: pass thru derive} for details) and compute the marginal abatement returns separately for each market. The estimates are then averaged across the 62 markets over the 2020--2023 period.\footnote{Of the 64 markets, 2 are excluded (2021 Chungnam and 2022 Gwangju) because the computation of subsidy pass-through matrices involves ill-conditioned, non-invertible matrices.} The marginal abatement return for HEV subsidies (1.671 tonnes of CO$_2$e per million KRW) is approximately 20\% higher than that for BEV subsidies (1.402 tonnes).\footnote{The inverse of the marginal abatement return is the marginal abatement cost, a measure commonly used in the literature and expressed as the cost per tonne of CO$_2$ (or CO$_2$e) abated. Under our no-subsidy baseline, the marginal abatement cost is 0.713 million KRW per tonne of CO$_2$e abated for the BEV subsidy policy and 0.598 million KRW for the HEV subsidy policy.} This finding suggests that subsidizing HEVs (financed by taxing BEVs) would reduce emissions relative to the no-subsidy baseline. Additionally, the marginal abatement returns for subsidies to the other fuel types are all negative, implying taxes would be appropriate, especially for diesel and LPG fuel types. The distributions of the marginal abatement returns across markets are consistent with these conclusions, as shown in Figure~\ref{fig: abatement_return_bubble_plot} in Appendix~\ref{sec: Additional figures and tables}.

We also decompose the abatement return into two components: one associated with products eligible for the subsidy (the own abatement return) and the other associated with products of other fuel types (the cross abatement return). Table~\ref{tab: prop 2} shows that the second term is negligible relative to the first across all fuel types, reflecting the low cross-fuel-type pass-through rates.

\begin{table}[!t]
  \centering
  \caption{Marginal abatement return by fuel type}
    \begin{tabular}{lrrrrr}
    \toprule
          & \multicolumn{1}{c}{BEV} & \multicolumn{1}{c}{HEV} & \multicolumn{1}{c}{Gasoline} & \multicolumn{1}{c}{Diesel} & \multicolumn{1}{c}{LPG} \\
    \midrule
    Marginal abatement return  & 1.402 & 1.671 & -0.355 & -1.006 & -0.939 \bigstrut[b]\\
    \hspace{0.1in}Own abatement return & 1.395 & 1.553 & -0.338 & -1.045 & -0.970 \bigstrut[t]\\
    \hspace{0.1in}Cross abatement return & 0.007 & 0.118 & -0.017 & 0.039 & 0.031 \\
    \bottomrule
    \end{tabular}%
  \label{tab: prop 2}%
  \tablenotes For each of the five fuel types, the first row presents the market-size-weighted average of the marginal abatement return (in tonnes of CO$_2$e per mil. KRW) across 62 markets during the 2020--2023 period. The bottom two rows separately report two components of the abatement return: one associated with products eligible for the subsidy (the own abatement return: $-\frac{1}{S_{g_1}} \sum_{j\in\mathscr{J}_g} \Theta_{j,g_1} \frac{\partial s_j}{\partial p_j}(e_j^D-e_j)$) and the other associated with products of other fuel types (the cross abatement return: $-\frac{1}{S_{g_1}} \sum_{g \neq g_1} \sum_{j\in\mathscr{J}_g} \Theta_{j,g_1} \frac{\partial s_j}{\partial p_j}(e_j^D-e_j)$). 
\end{table}%

\paragraph{Total abatement returns}\label{sec: Total abatement returns}

The previous results suggest that, under zero government spending, a marginal increase in HEV subsidies at the expense of taxing BEVs would reduce emissions. We now compare the total effects of BEV and HEV subsidies while accounting for the substantial existing BEV subsidies in South Korea. Specifically, we compare the current BEV subsidies with a uniform per-vehicle HEV subsidy calibrated so that total HEV subsidy expenditures approximately match the BEV subsidy budget over the four-year period 2020--2023. 

Table~\ref{tab: subsidy effects} shows the main findings. According to Panel A, the budget-neutral HEV subsidy amounts to an average of 1.9 million KRW per vehicle, substantially lower than the actual BEV subsidy of 8.6 million KRW per vehicle. This reflects the fact that HEV sales are significantly higher than BEV sales. At the same time, the HEV subsidy induces manufacturers to reduce gross prices of HEVs across the board (implying pass-through rates above 100\% for all models), whereas the BEV subsidy results in both increases and decreases in BEV prices (implying pass-through rates that vary across models between 71\% and 115\%); see Table~\ref{tab: counterfactual prices} in Appendix \ref{sec: Additional figures and tables}. 
In sum, the budget-neutral HEV subsidy is substantially below the BEV subsidy, but this is accompanied by slightly higher pass-through. Panels B and C of Table~\ref{tab: subsidy effects} show how the subsidies and implied price adjustments affect sales and emissions.

According to Panel B, the HEV subsidies would increase HEV sales by almost 178 thousand units, whereas BEV subsidies would increase BEV sales by only 78 thousand units. Hence, a budget-neutral HEV subsidy has a much larger sales effect despite the smaller magnitude of the subsidy. HEV subsidies would encourage more consumers to substitute away from ICEVs (gasoline, diesel, and LPG) than BEV subsidies do. Most notably, HEV subsidies would reduce gasoline-vehicle sales by 92 thousand, compared with only 28 thousand under BEV subsidies. Furthermore, HEV subsidies divert comparatively more sales from the outside good (consisting mainly of older ICEVs and therefore entailing the highest emissions). In sum, the HEV subsidy policy is more effective at inducing substitution away from ICEVs than the BEV subsidy policy. Consequently, the former may yield greater reductions in aggregate emissions than the latter.

Panel C of Table~\ref{tab: subsidy effects} confirms that this is indeed the case. HEV subsidies would lead to a reduction in total emissions of 2.1 million tonnes of CO$_2$e over the four-year period, approximately 47\% larger than the amount of the emission savings attributable to BEV subsidies. We separately calculate the various parts of equation \eqref{eqn: mkt emissions} to decompose the change in total emissions into three components: fuel-cycle, vehicle-cycle, and outside emissions. BEV subsidies yield larger reductions in fuel-cycle emissions than HEV subsidies (942 versus 572 thousand tonnes of CO$_2$e). However, they also imply larger vehicle-cycle emissions. Furthermore, they imply much lower savings in outside emissions (871 versus 1\,821 thousand tonnes of CO$_2$e), because they are less successful in diverting sales from the outside good.

\begin{table}[!t]
  \centering
  \caption{Comparison of BEV and HEV subsidy effects}
    \begin{tabular}{lrr}
    \toprule
          & \multicolumn{1}{l}{BEV subsidies} & \multicolumn{1}{l}{HEV subsidies} \\
    \midrule
    \multicolumn{3}{l}{\textit{Panel A: Subsidy spending}} \\
          &       &  \\
    \multicolumn{1}{l}{\quad Per-unit subsidy (million KRW)} & 8.646 & 1.943 \\
    \multicolumn{1}{l}{\quad Total subsidy expenditure (trillion KRW)} & 1.563 & 1.565 \\
          &       &  \\
    \midrule
    \multicolumn{3}{l}{\textit{Panel B: Sales effects (number of vehicles)}} \\
          &       &  \\
    \multicolumn{1}{l}{Total sales change} &       &  \bigstrut[b] \\
    \multicolumn{1}{l}{\quad BEV} & 77\,655  & -3\,206  \\
    \multicolumn{1}{l}{\quad HEV} & -15\,026  & 177\,697  \\
    \multicolumn{1}{l}{\quad Gasoline} & -28\,244  & -92\,213  \\
    \multicolumn{1}{l}{\quad Diesel} & -13\,793  & -31\,708  \\
    \multicolumn{1}{l}{\quad LPG} & -2\,114  & -8\,666  \\
    \multicolumn{1}{l}{\quad Outside option} & -18\,478  & -41\,903  \\
          &       &  \\
    \midrule
    \multicolumn{3}{l}{\textit{Panel C: Emissions effects (thousand tonnes CO$_2$e)}} \\
          &       &  \\
    \multicolumn{1}{l}{Total emissions change} & -1\,433  & -2\,104  \bigstrut[b]\\
    \multicolumn{1}{l}{\quad Fuel-cycle} & -942  & -572  \bigstrut[b]\\
    \quad \quad - ICEV & -1\,930  & -5\,208  \\
    \quad \quad - BEV & 1\,473  & -69  \\
    \quad \quad - HEV & -485  & 4\,705  \\
    \multicolumn{1}{l}{\quad Vehicle-cycle} & 381   & 289  \bigstrut[t]\\
    \multicolumn{1}{l}{\quad Outside option} & -871  & -1\,821  \\
          &       &  \\
    \bottomrule
    \end{tabular}%
  \label{tab: subsidy effects}%
\begin{flushleft}
\footnotesize Notes: The top panel compares the observed BEV subsidy levels with the counterfactual HEV subsidy design. The per-vehicle subsidy is calculated as the sales-weighted average across all nameplates. The bottom two panels report changes in sales and emissions over the 2020--2023 period following the implementation of BEV and HEV subsidies.
\end{flushleft}
\end{table}

Note that HEV subsidies lead to a smaller reduction in fuel-cycle emissions than in outside emissions (0.6 versus 1.8 million tonnes of CO$_2$e). This result, however, does not imply that the HEV subsidy policy reduces emissions primarily by inducing consumers to substitute away from the outside option rather than from ICEVs. Instead, it arises because the reduction in inside fuel-cycle emissions is a net figure that accounts for the simultaneous increase in emissions from newly adopted HEVs. To illustrate this point, Table~\ref{tab: subsidy effects} also decomposes the change in fuel-cycle emissions by fuel type: ICEVs (combining gasoline, diesel, and LPG), BEVs, and HEVs; see Appendix \ref{sec: Calculation of GHG emissions changes} for details of the decomposition. The decomposition reveals that HEV subsidies lead to a substantial reduction in fuel-cycle emissions from ICEVs (5.2 million tonnes of CO$_2$e in total), nearly three times as large as outside emissions reductions. Furthermore, BEV subsidies lead to a much smaller reduction in fuel-cycle emissions from ICEVs (only 1.9 million tonnes of CO$_2$e), reflecting lower substitution between BEVs and ICEVs. 

\paragraph{Sensitivity analysis}

To better understand the magnitude of the estimated emission effects, we compare the baseline emission effects with those obtained under a range of alternative specifications and assumptions. Table~\ref{tab: alternative specifications/assumptions} shows the results from this exercise.

First, the baseline emission effects are based on the flexible random coefficients logit model reported in Table \ref{tab: rc demand estimation results}, RC logit 2, which incorporates mileage heterogeneity and other sources of unobserved consumer heterogeneity. To illustrate the role of this flexibility, we quantify the impact of the HEV and BEV subsidies under our more restrictive specifications in Panel A of Table~\ref{tab: alternative specifications/assumptions}. RC logit 1, which assumes homogeneous mileage but accounts for unobserved heterogeneity, also predicts considerably larger emission savings under the HEV-focused subsidy, although it underestimates total emission savings of both policies, compared to the most flexible RC logit 2 specification.\footnote{Appendix \ref{sec: Mileage heterogeneity and policy effects} provides a detailed interpretation of why the assumption of homogeneous mileage underestimates emission savings under both policies.} In contrast, IV logit without consumer heterogeneity yields \textit{larger} emission savings under the BEV-focused subsidy. This is consistent with our theoretical analysis in Section \ref{sec: theory}, which implies that under symmetric logit demand a BEV subsidy is more effective whenever BEV emissions are lower than HEV emissions.


\begin{table}[!t]
  \centering
  \caption{Emission effects under alternative specifications and assumptions}
    \begin{tabular}{lrrr}
    \toprule
 &       & \multicolumn{1}{c}{BEV sub.} & \multicolumn{1}{c}{HEV sub.} \bigstrut[t] \\
    \midrule
    Baseline &       & -1\,433  & -2\,104  \\
          &       &       &  \\
    \midrule
    \multicolumn{4}{l}{\textit{Panel A: restrictive demand models}} \\
          &       &       &  \\
    IV Logit &       & -1\,203  & -1\,105  \\
    RC logit 1 &       & -1\,025  & -1\,624  \\
          &       &       &  \\
    \midrule
    \multicolumn{4}{l}{\textit{Panel B: alternative vehicle operating periods}} \\
          &       &       &  \\
    15- \& 9-year lifespan &       & -1\,084  & -1\,376  \\
    6- \& 6-year lifespan &       & -344  & -668  \\
          &       &       &  \\
    \midrule
    \multicolumn{4}{l}{\textit{Panel C: alternative ownership structures}} \\
          &       &       &  \\
    Brand &       & -1\,458  & -1\,942  \\
    Product &       & -1\,453  & -1\,877  \\
    \bottomrule
    \end{tabular}%
  \label{tab: alternative specifications/assumptions}%
  \tablenotes Panel A presents emission changes (in thousand tonnes of CO$_2$e) under restrictive demand models: IV logit and RC logit 1. Panel B presents emission changes under alternative vehicle operating periods: the ``15- \& 9-year lifespan'' row evaluates a truncated 9-year ownership horizon strictly for the outside option. The ``6- \& 6-year lifespan'' row applies a 6-year holding period to both inside and outside goods. Panel C presents emission changes under alternative ownership structures: brand-level ownership and product-level ownership.
\end{table}%

Second, the relative magnitude of the emission effects under HEV and BEV subsidies hinges on assumptions regarding the vehicles' expected operating period for new and secondhand vehicles.  Our baseline analysis assumes a 15-year operating period, in line with both the average passenger-vehicle lifespan in South Korea and the existing literature, and uses this horizon to compute changes in cumulative emissions. Since the lifespan of the outside option (secondhand vehicles) exceeds 15 years under this baseline specification, we consider two alternatives here. First, we continue to assume an expected operating period of 15 years for new vehicles, but set the expected operating period of secondhand vehicles to 9 years, which equals the expected vehicle life span of 15 years minus the average holding period of new vehicles of 6 years. Second, we conservatively assume an expected operating horizon of 6 years for all options, matching the average initial holding period in South Korea, to measure policy-induced environmental effects over the interval before consumers make their subsequent vehicle choices.

Panel B of Table~\ref{tab: alternative specifications/assumptions} compares the results from these alternative scenarios with our baseline case. In both alternative scenarios, HEV subsidies still reduce emissions more than BEV subsidies. In the first alternative scenario, with a shorter expected operating period of 9 years for secondhand vehicles, HEV subsidies reduce emissions by 1.4 million tonnes of CO$_2$e, compared with only 1.1 million tonnes for BEV subsidies. The difference is smaller because emission savings from diverting sales away from secondhand vehicles (the outside option) are especially reduced under HEV subsidies, owing to the shorter assumed 9-year horizon. In the second scenario, with a shorter expected operating period of 6 years for both new and secondhand vehicles, estimated emission savings are naturally smaller, but HEV subsidies remain more effective, even more so in percentage terms.

Third, the relative magnitude of the emission effects may also depend on the degree of subsidy pass-through, which in turn is influenced by firms' market power under multi-product Bertrand competition. As discussed in Section~\ref{sec: south korean car market}, the South Korean passenger vehicle market is highly concentrated, with Hyundai Motor Group accounting for over 70\% of total sales. To assess the role of market power, we compare the predicted emission effects under our baseline parent-company ownership assumption with those under two alternative ownership structures: brand-level ownership (treating Hyundai, Kia, and Genesis as separate firms) and product-level ownership.

We find that the emission effects of the subsidies remain very similar across these alternative ownership structures (see Panel C of Table~\ref{tab: alternative specifications/assumptions}). This suggests that competitive constraints from rival firms, such as Toyota for HEVs and Tesla for BEVs, limit Hyundai's incentives to reduce subsidy pass-through, thereby preserving diversion away from ICEVs.

In sum, across these alternative scenarios, the emission savings are consistently lower under BEV subsidies than under budget-neutral HEV subsidies. These findings confirm the comparative strength of HEVs in diverting sales away from conventional ICEVs.

\subsection{The role of the carbon intensity of electricity generation}\label{sec: counterfactual 3}

In the previous subsection, our comparison of BEV and HEV subsidies has been based on fuel-cycle emissions for BEVs evaluated under the current electricity generation mix in South Korea. The carbon intensity of electricity generation depends on the relative importance of fossil fuels and renewable energy sources. This mix evolves over time and varies across countries, partly reflecting past policy choices regarding the promotion of renewable energy. If electricity generation is highly carbon-intensive, HEV subsidies may be even more effective relative to BEV subsidies. Conversely, as electricity generation becomes less carbon-intensive, BEV subsidies may become more effective.

To assess how the energy mix affects the relative effectiveness of HEV and BEV subsidies, we compare their emission savings at different phases of the energy transition, using other countries as benchmarks. Specifically, we adjust the electricity-sector emission factor (g CO$_2$e/kWh of electricity generation) to the levels in other countries. This adjustment alters the fuel-cycle emissions, $e_j^{\text{FC}}$, of BEVs. As discussed in Section~\ref{sec: greet}, these emissions are obtained by dividing the electricity emission factor by vehicle $j$'s fuel economy (in km/kWh). Formally,
\begin{equation*}
	e_j^{\text{FC}} = \frac{\text{emission factor (g CO$_2$e/kWh)}}{\text{fuel economy}_j \text{ (km/kWh)}}.
\end{equation*}

In 2022, the emission factor in South Korea was 453 g CO$_2$e/kWh, comparable to that of the United States and Germany and lower than the world average of 526 g CO$_2$e/kWh. Most European countries have lower emission factors, whereas Japan, China, and many developing countries have higher ones. Table \ref{tab: ghg emission intensity} in Appendix \ref{sec: Data description} provides the emission-factor levels for 148 countries.

We calculate the emission savings under BEV and HEV subsidies using the emission factors of 18 other countries, assuming that these emission factors affect only the fuel-cycle emissions of BEVs, while fuel-cycle emissions for other fuel types and vehicle-cycle emissions remain constant. Figure \ref{fig: comparison energy mix} summarizes the results; the corresponding numerical results are reported in Table \ref{tab: policy counterfactual 3} of Appendix \ref{sec: Additional figures and tables}.

\begin{figure}[!t]
\centering
\caption{Policy comparison with different GHG emission intensity of electricity generation}
    \includegraphics[width=.7\textwidth]{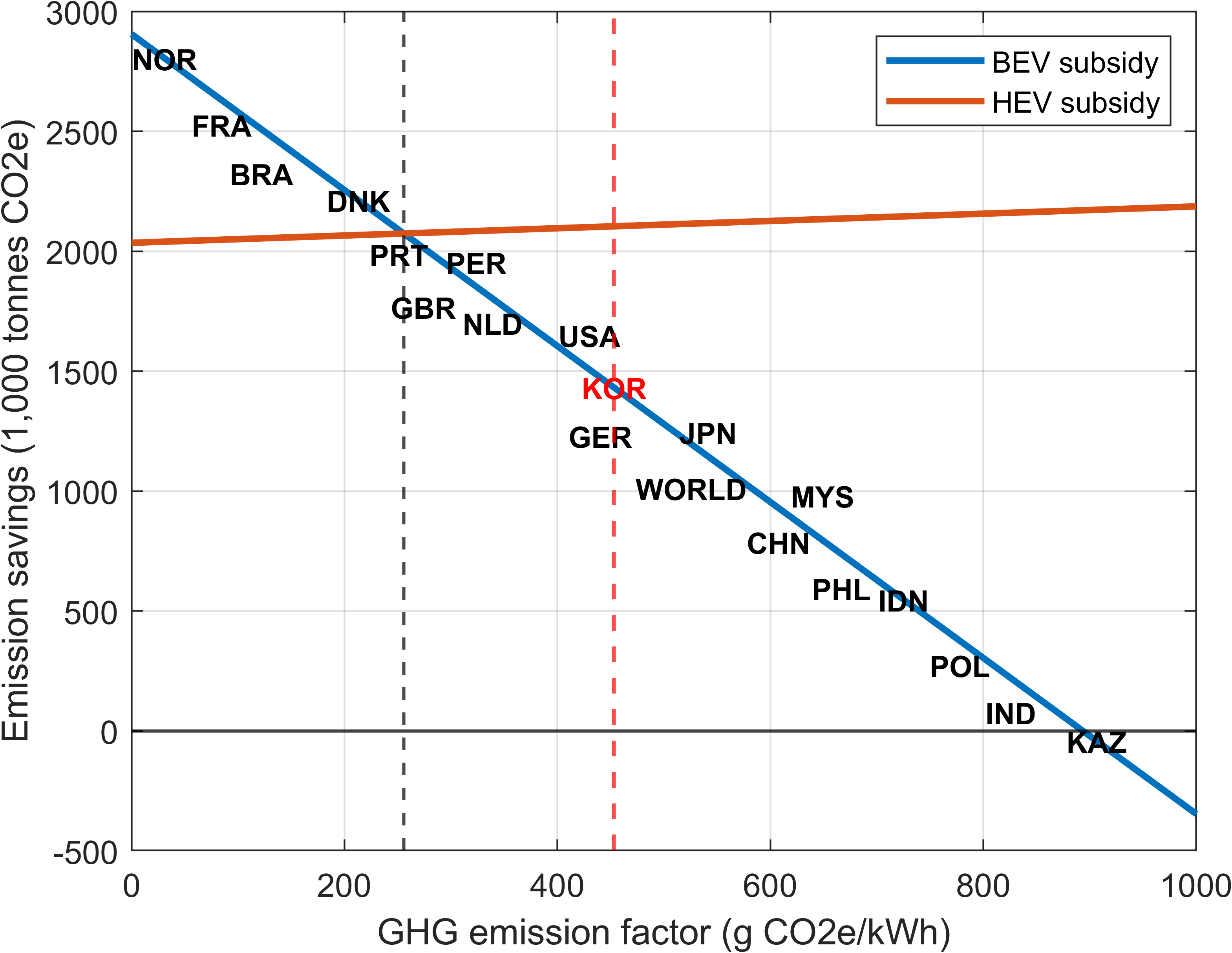}
\label{fig: comparison energy mix}
\tablenotes This figure illustrates the emission savings under BEV and HEV subsidy policies across varying GHG emission factors.
\end{figure}

Emission savings from the BEV subsidy decrease linearly with the emission factor (represented by the blue line), amounting to 2.8 million tonnes of CO$_2$e under Norway's low-carbon electricity generation mix and only 0.2 million tonnes under India's much more carbon-intensive mix. Conversely, emission savings from the HEV subsidy increase slightly with the emission factor (represented by the red line). As the grid becomes more carbon-intensive, the emissions profile of a BEV worsens relative to that of an HEV, which mechanically increases the emission savings of the HEV subsidy policy.

At South Korea's current electricity generation mix, emission savings amount to 1.4 million tonnes of CO$_2$e for the BEV subsidy and 2.1 million tonnes for the HEV subsidy. As the emission factor declines with a cleaner electricity generation mix, the emission savings of a BEV subsidy grow rapidly. The emission savings from the two policies approximately equalize at Portugal's emission factor level (251 g CO$_2$e/kWh), which is 55\% of South Korea's current level. In other words, BEV subsidies become as effective as HEV subsidies only if South Korea's electricity generation mix becomes 45\% cleaner. At even lower emission-factor levels, a BEV subsidy becomes more effective than an HEV subsidy. 

In contrast, BEV subsidies would become less effective at curbing GHG emissions if South Korea's emission factor were higher than its current level. For example, at China's emission factor level, emission savings from a BEV subsidy would be less than half of the savings from an HEV subsidy (0.9 million tonnes versus 2.1 million tonnes of CO$_2$e). If the electricity generation mix were as carbon-intensive as in India and Kazakhstan, BEV subsidies would hardly reduce emissions or could even be environmentally harmful.
 
These results emphasize that the environmental efficacy of a BEV subsidy policy depends critically on the carbon intensity of electricity generation. The greater the reliance on cleaner energy sources such as wind, solar, and hydropower, the greater the environmental benefits of BEV adoption. In fact, many countries with large automobile markets, including China, the United States, Japan, and India, may not yet have transitioned sufficiently toward low-carbon electricity to fully realize these potential benefits.

\section{Conclusion and discussion}\label{sec: conclusion}

We develop a framework for comparing the relative effectiveness of subsidies for alternative emission-reducing technologies. We show that an intermediate technology may reduce emissions more effectively than the cleanest technology if it 
induces sufficiently greater substitution away from the prevailing high-emission technology to offset its emission disadvantage relative to the cleanest technology. Our empirical analysis shows that these conditions are satisfied in the South Korean passenger vehicle market. Specifically, HEVs exhibit a sufficiently small emission gap with BEVs and induce considerably more substitution away from ICEVs. As a result, reallocating the subsidy budget from BEVs to HEVs would reduce total emissions by an additional 47\%. For BEV subsidies to outperform HEV subsidies in reducing emissions, the country's electricity generation mix would need to become 45\% cleaner.

Our findings highlight two main challenges for current BEV-focused policies in fully realizing their environmental benefits. First, technological progress that lowers BEV prices and improves vehicle attributes is essential for inducing consumers of conventional ICEVs to switch to BEVs in response to subsidies. Second, a cleaner electricity generation mix is needed to ensure that increased BEV adoption translates into larger emissions reductions.

If HEV subsidies are currently more effective at reducing emissions than BEV subsidies under South Korea's existing electricity generation mix, why has South Korea nevertheless chosen to focus policy support on BEVs? One possible explanation is that BEV subsidies perform better in terms of other welfare components. To examine this possibility, we use our model to go beyond emissions and compare the effects of the two subsidy policies on firm profits and consumer surplus, as shown in Table \ref{tab: welfare}. First, the industry as a whole benefits much more from HEV than from BEV subsidies. Moreover, Hyundai Motor Group, the dominant domestic player, benefits disproportionately more from HEV than from BEV subsidies, with profits increasing by nearly four times as much. These findings therefore suggest no short-term industrial policy rationale for BEV subsidies. 

\begin{table}[!t]
  \centering
  \caption{Changes in firm profits and consumer surplus}
    \begin{tabular}{lrr}
    \toprule
    Term  & BEV subsidies & HEV subsidies \\
    \midrule
    \\
    Producer surplus change & 292.5 & 1\,024.3 \bigstrut[b] \\
    \hspace{.1in}- Hyundai Motor Group & 291.7 & 1\,111.4 \\
    \hspace{.1in}- Other domestic firms & -12.4 & -66.0 \\
    \hspace{.1in}- German firms & -48.4 & -66.7 \\
    \hspace{.1in}- Toyota Group & -5.8  & 52.7 \\
    \hspace{.1in}- Tesla & 67.5  & -7.1   \\
    \\
    Consumer surplus change & 971.5 & 1\,253.4  \\
    \hspace{.1in}- Bottom 20\% & 0.4   & -0.7 \\
    \hspace{.1in}- 20-40\% & 23.6  & 21.2 \\
    \hspace{.1in}- 40-60\% & 109.4 & 194.5 \\
    \hspace{.1in}- 60-80\% & 230.0 & 378.2 \\
    \hspace{.1in}- Top 20\% & 608.1 & 660.2 \\
    \\
    \bottomrule
    \end{tabular}%
  \label{tab: welfare}%
\tablenotes This table reports changes (in billion KRW) in firm profits and consumer surplus (by income quantile) under vehicle subsidy policies (for either BEVs or HEVs). Other domestic firms include GM Korea, KG Mobility, and Renault Korea, while German firms include BMW, Mercedes-Benz, and Volkswagen Group.
\end{table}%

Second, consumers also benefit more from HEV than BEV subsidies, reflecting the larger demand responses and higher pass-through rates associated with HEV subsidies. Moreover, no income group benefits more from BEV than from HEV subsidies, apart from negligible differences for the two lowest income groups. Hence, BEV subsidies also cannot be rationalized as a means of protecting consumers generally or targeting particular income groups.

The absence of an immediate welfare rationale within our framework suggests that the motivation for BEV subsidies may lie elsewhere. One possibility is that political economy considerations play a role if governments place disproportionate emphasis on promoting frontier technologies that are particularly salient to voters. Politicians often invest in public goods based on the visibility of their outcomes rather than their contribution to voters' welfare (see, for example, \citet{mani2007democracy}). In line with this reasoning, BEVs' widespread reputation as the cleanest available technology may lead politicians to promote their adoption without closely examining the welfare consequences relative to other technologies.

Another possible rationale for BEV subsidies is that governments pursue long-term industrial policy objectives. Indeed, since January 2026, the South Korean government has explicitly identified ``support for the domestic BEV industry'' as a policy objective (alongside carbon reduction). Similar considerations may also apply in other countries that actively promote BEV adoption, such as China, whose electricity generation mix is more carbon-intensive than South Korea's. However, if the objective is to stimulate innovation and learning-by-doing in the domestic BEV industry, it remains unclear whether purchase subsidies are the most appropriate policy instrument, or whether alternative instruments targeted more directly at innovation and production would be preferable.

Finally, governments may use BEV subsidies as part of a broader international trade strategy in response to similar industrial policies adopted abroad.\footnote{\citet{barwick2026trade} show that a revenue-neutral combination of a moderate tariff on imported EVs and a subsidy for domestic EV production improves national welfare compared to both outright protectionism and laissez-faire.} The United States and the European Union, for example, have recently combined restrictions on Chinese EVs with domestic incentives for EV and battery production \citep{USTR2024Section301, USTR2024NoticeModification, EuropeanCommission2024CVD,IRA2022, EuropeanCommission2021IPCEI}. South Korea has recently moved in a similar direction by allocating subsidies to firms based on their contribution to the domestic BEV supply chain \citep{joongang2026ev}, effectively reducing subsidies to Chinese automakers such as BYD. Exploring the interaction between environmental policy and these broader policy objectives remains an important topic for future research.

\newpage
\singlespace
\bibliographystyle{apalike}
\bibliography{reference}

\newpage
\appendix
\renewcommand{\thesection}{\Alph{section}}
\renewcommand{\thefigure}{\thesection\arabic{figure}}
\setcounter{section}{0}
\setcounter{figure}{0}

\onehalfspacing

\section{Supplement to Section \ref{sec: theory}}\label{app: theory_supplement}

\subsection{Proof of Proposition \ref{p: per-unit subsidy}}\label{app: prop 1 derive}

Given our focus on BEV and HEV subsidies, $\bm{\tau}=(0,\tau_H,\tau_B)$, total emissions and government expenditure can be expressed as
\begin{align*}
&E(\tau_H, \tau_B) = M\sum_{j \in \{I, H, B\}} e_j s_j(\mathbf{p}(\tau_H, \tau_B)), \\
&G(\tau_H, \tau_B) = M\sum_{j \in \{H, B\}} \tau_js_j(\mathbf{p}(\tau_H, \tau_B)).
\end{align*}

Under complete subsidy pass-through and no price response for non-subsidized products, the effect of a subsidy for product $j \in \{H, B\}$ on total emissions is given by Equation \eqref{eqn: partial E partial tau} in the main text. The impact on total government expenditure is given by
\begin{equation}\label{eqn: partial G partial tau}
    \frac{\partial G}{\partial\tau_j} = M \left( s_j - \sum_{\ell \in \{H, B\}} \tau_{\ell} \frac{\partial s_{\ell}}{\partial p_j}\right),    
\end{equation}
which we assume is positive. This assumption holds provided that (i) the current subsidy $\tau_k$ ($k\neq j$) is sufficiently small, or (ii) the subsidy increase primarily diverts sales away from ICEVs (the highest emission product) rather than between HEVs and BEVs.

To calculate the impact of a budget-neutral subsidy increase for product $j\in\{H,B\}$, financed by a subsidy decrease for product $k\neq j$ to maintain a constant government budget $\bar{G}$, let the function $\tau_k(\tau_j)$ be implicitly defined by the government budget constraint $G(\tau_j, \tau_k(\tau_j)) = \bar{G}$. Applying the implicit function theorem yields
\begin{equation}\label{eqn: d tau B over d tau H}
    \frac{\partial G}{\partial \tau_j} + \frac{\partial G}{\partial \tau_k}\frac{d \tau_k}{d \tau_j} = 0,
\end{equation}
for $j\in\{H,B\}$, $k\neq j$, so that
\begin{equation}\label{eqn: partial tau B over partial tau H}
    \frac{d \tau_k}{d \tau_j} = - \frac{\partial G / \partial \tau_j}{\partial G / \partial \tau_k}.
\end{equation}

The marginal impact of a budget-neutral increase in the per-unit subsidy $\tau_j$ on total emissions is then given by
\begin{equation}\label{eqn: ift}
    \frac{d E(\tau_j, \tau_k(\tau_j))}{d \tau_j} = \frac{\partial E}{\partial \tau_j} +  \frac{\partial E}{\partial \tau_k}\frac{d \tau_k}{d \tau_j}.
\end{equation}
Substituting \eqref{eqn: partial tau B over partial tau H} and rearranging, it follows that $\frac{d E}{d \tau_j}$ is negative if and only if
\begin{equation}\label{eq: abatement return condition}
    \frac{-\partial E/\partial \tau_j}{\partial G/\partial\tau_j}>\frac{-\partial E/\partial \tau_k}{\partial G/\partial\tau_k},
\end{equation}
i.e., if and only if the marginal abatement return of a per-unit subsidy is higher for product $j$ than for product $k$.

Substituting equations \eqref{eqn: partial E partial tau} and \eqref{eqn: partial G partial tau}, we can express the condition \eqref{eq: abatement return condition} under which a budget-neutral subsidy increase for product $j$ decreases total emissions (i.e., $dE/d\tau_j < 0$) as
\begin{equation}\label{eqn: general proposition}
    \frac{- \frac{\partial s_j}{\partial p_j} (e_j^D-e_j)}{s_j - \tau_j \frac{\partial s_j}{\partial p_j} - \tau_k \frac{\partial s_k}{\partial p_j}} >  \frac{- \frac{\partial s_k}{\partial p_k} (e_k^D-e_k)}{s_k - \tau_k \frac{\partial s_k}{\partial p_k} - \tau_j \frac{\partial s_j}{\partial p_k}}.
\end{equation}
Evaluating condition \eqref{eqn: general proposition} at zero subsidies ($\tau_H = \tau_B = 0$) yields condition \eqref{eqn: main} in Proposition \ref{p: per-unit subsidy}. 

\subsection{Comparison of net diverted emissions}\label{app: diversion emission derive}

If demand responses are identical across technologies, $-\frac{\partial s_j}{\partial p_j}\frac{1}{s_j} = -\frac{\partial s_k}{\partial p_k}\frac{1}{s_k}$, then the condition for an HEV subsidy to be more effective than a BEV subsidy reduces to $e_H^D-e_H > e_B^D-e_B$. This condition can also be expressed directly in terms of diversion ratios, $D_{H\to I}$ and $D_{B\to I}$, instead of diverted emissions, $e_H^D$ and $e_B^D$:
\begin{equation}\label{eqn: condition diversion ratio}
    e_H^D - e_H > e_B^D - e_B \quad \text{if and only if} \quad \frac{2-D_{H\to I}}{2-D_{B\to I}} < \frac{e_I-e_H}{e_I-e_B}.
\end{equation}
To show this, we use the definition of $e_j^D$ and the fact that $\sum_{k\neq j}D_{j\to k}=1$ to obtain
\begin{equation*}
\begin{aligned}
    e_H^D - e_H & = D_{H\to I}e_I + D_{H\to B}e_B - e_H \\
    & = (1-D_{H\to B}) e_I + D_{H\to B}e_B - e_H \\
    & = D_{H\to B}(e_B-e_I) + e_I -e_H.
\end{aligned}
\end{equation*}
Similarly,
\begin{equation*}
    e_B^D - e_B = D_{B\to H}(e_H-e_I) + e_I - e_B.
\end{equation*}
Substituting these expressions and rearranging yield the equivalent condition above.

Intuitively, condition \eqref{eqn: condition diversion ratio} highlights that when $e_H=e_B$, an HEV subsidy is more effective whenever $D_{H\to I}>D_{B\to I}$, i.e. when HEVs divert more sales from ICEVs than BEVs do.

\subsection{Proof of Proposition \ref{p: per-unit subsidy generalized}}\label{app: prop 2 derive}

Here, we allow for multiple products within each fuel type, additional fuel types, and non-trivial subsidy pass-through under imperfect competition between multi-product price-setting firms. In this setup, let $\mathscr{J}$ and $\mathscr{J}_g$ denote the set of all inside products and the set of products with fuel type $g$, respectively. Inside products are partitioned into $G$ groups according to their fuel types, such that $\mathscr{J} = \mathscr{J}_1 \cup \dots \cup \mathscr{J}_g \cup \dots \cup \mathscr{J}_G$. Fix a subsidy policy $\bm{\tau}=(\tau_1, \dots, \tau_g, \dots, \tau_G)$, where $\tau_g$ denotes the per-unit subsidy for all products with fuel type $g$. Let $\Theta_{jg} = -\frac{\partial p_j}{\partial \tau_g}$ denote the subsidy pass-through.

We also allow usage emissions to vary across consumers through heterogeneous annual mileage. Let $e_{ik}$ denote the \emph{individual} emissions of product $k$, that is, the emissions from the product when used by consumer $i$. Let $D_{j\to k}^i = -\frac{\partial s_{ik}/\partial p_j}{\partial s_{ij}/\partial p_j}$ denote the \emph{individual} diversion ratio from product $j$ to $k$. The \emph{individual} diverted emissions, defined as 
\begin{equation*}
    e_{ij}^D = \sum_{k\neq j}D_{j\to k}^ie_{ik},
\end{equation*}
vary across consumers due to heterogeneity in both emissions and demand responses.

Following a similar logic in Section \ref{app: prop 1 derive}, the condition under which a budget-neutral subsidy increase for fuel type $g_1$, financed by a subsidy decrease for fuel type $g_2 \neq g_1$, decreases total emissions is given by
\begin{equation}\label{eqn: equiv}
    \frac{-\partial E/\partial \tau_{g_1}}{\partial G/\partial \tau_{g_1}} > \frac{-\partial E/\partial \tau_{g_2}}{\partial G/\partial \tau_{g_2}}.
\end{equation}

Consumer $i$'s expected emissions under the subsidy policy $\bm{\tau}$, $E_i(\bm{\tau})$, are given by the weighted average of emissions across all products and the outside option, using the consumer's choice probabilities as weights:
\begin{equation*}
    E_i(\bm{\tau}) = \sum_{k\in\mathscr{J}\cup\{0\}} e_{ik} s_{ik}(\mathbf{p}(\bm{\tau})).
\end{equation*}
Total market emissions are obtained by integrating these individual emissions:
\begin{equation*}
    E(\bm{\tau}) = M \int E_i(\bm{\tau}) \, dF(i) = M \int \sum_{k\in\mathscr{J}\cup\{0\}} e_{ik} s_{ik}(\mathbf{p}(\bm{\tau})) \, dF(i),
\end{equation*}
where $M$ represents the market size, and $F(i)$ denotes the consumer distribution, defined over observed demographics (income and mileage) as well as unobserved heterogeneity.

Differentiating $-E(\bm{\tau})$ with respect to $\tau_g$ gives
\begin{equation}\label{eqn: partial E partial tau generalized}
\begin{aligned}
    -\frac{\partial E}{\partial\tau_g} 
    & = -M \cdot \int \sum_{k\in\mathscr{J}\cup\{0\}}e_{ik} \frac{\partial s_{ik}}{\partial\tau_g} d F(i) \\
    & = -M \cdot \int \sum_{k\in\mathscr{J}\cup\{0\}} e_{ik} \left( \sum_{g=1}^G\sum_{j\in\mathscr{J}_g} \frac{\partial s_{ik}}{\partial p_j}(-\Theta_{jg}) \right) d F(i) \\
    & = M \cdot \int \sum_{g=1}^G\sum_{j\in\mathscr{J}_g} \Theta_{jg} \left( e_{ij}\frac{\partial s_{ij}}{\partial p_j} + \sum_{k\neq j}e_{ik}\frac{\partial s_{ik}}{\partial p_j} \right) d F(i) \\
    & = M \sum_{g=1}^G\sum_{j\in\mathscr{J}_g} \Theta_{jg} \cdot \int - \frac{\partial s_{ij}}{\partial p_j} \left( \sum_{k\neq j} D_{j\to k}^i e_{ik} - e_{ij} \right) d F(i) \\
    & = - M \sum_{g=1}^G\sum_{j\in\mathscr{J}_g}  \Theta_{jg} \frac{\partial s_j}{\partial p_j} \left( \int \omega_{ij} e_{ij}^D d F(i) - \int \omega_{ij} e_{ij} d F(i) \right) \\
    & = - M \sum_{g=1}^G\sum_{j\in\mathscr{J}_g}  \Theta_{jg} \frac{\partial s_j}{\partial p_j} ( e_{j}^D - e_{j} ),  
\end{aligned}
\end{equation}
where $\omega_{ij}$ is defined as
\begin{equation*}
    \omega_{ij} = \frac{\partial s_{ij}/\partial p_j}{\int \frac{\partial s_{ij}}{\partial p_j} \, dF(i)} = \frac{\partial s_{ij}/\partial p_j}{\partial s_j/\partial p_j}. 
\end{equation*}
We use $\omega_{ij}$ as the weights for the individual-level emissions metrics ($e_{ij}$ and $e_{ij}^D$) to define the product-level emissions $e_j$ and product-level diverted emissions $e_j^D$ as weighted averages of their individual counterparts:
\begin{equation*}
    e_j = \int \omega_{ij} e_{ij} \, dF(i) 
    \quad \text{and} \quad
    e_j^D = \int \omega_{ij} e_{ij}^D \, dF(i). 
\end{equation*}
The final equality in equation \eqref{eqn: partial E partial tau generalized} follows from these definitions.

Next, observe that
\begin{equation*}
\begin{aligned}
    G(\bm{\tau}) & = M\sum_{j\in\mathscr{J}} \tau_js_j(\mathbf{p}(\bm{\tau})) \\
    & = M\left( \sum_{j\in\mathscr{J}_g} \tau_gs_j(\mathbf{p}(\bm{\tau})) + \sum_{r\neq g, r\neq 0}\sum_{k\in\mathscr{J}_r}\tau_rs_k(\mathbf{p}(\bm{\tau})) \right) \\
    & = M \left( \tau_gS_g + \sum_{r\neq g,r\neq 0}\tau_rS_r\right).
\end{aligned}
\end{equation*}
Differentiating $G(\bm{\tau})$ with respect to $\tau_g$ gives
\begin{equation}\label{eqn: partial G partial tau generalized}
    \frac{\partial G}{\partial\tau_g} = M\left( S_g + \tau_g\frac{\partial S_g}{\partial\tau_g} + \sum_{r\neq g, r\neq 0}\tau_r\frac{\partial S_r}{\partial \tau_g} \right).
\end{equation}

Finally, substituting equations \eqref{eqn: partial E partial tau generalized} and \eqref{eqn: partial G partial tau generalized}, we can express condition \eqref{eqn: equiv} as
\begin{equation}\label{eqn: generalized proposition} 
    - \frac{\sum_{g=1}^G\sum_{j\in\mathscr{J}_g} \Theta_{j,g_1} \frac{\partial s_j}{\partial p_j}  (e_j^D-e_j)}{S_{g_1} + \tau_{g_1}\frac{\partial S_{g_1}}{\partial\tau_{g_1}} + \sum_{r\neq g_1, r\neq 0}\tau_r\frac{\partial S_r}{\partial\tau_{g_1}}}
    >  
    - \frac{\sum_{g=1}^G\sum_{j\in\mathscr{J}_g}  \Theta_{j,g_2} \frac{\partial s_j}{\partial p_j} (e_j^D-e_j)}{S_{g_2} + \tau_{g_2}\frac{\partial S_{g_2}}{\partial\tau_{g_2}} + \sum_{r\neq g_2, r\neq 0}\tau_r\frac{\partial S_r}{\partial\tau_{g_2}}}.
\end{equation}
Evaluating condition \eqref{eqn: generalized proposition} at zero subsidies (i.e., $\tau_g = 0, ~ \forall g$) yields condition \eqref{eqn: multiproduct} in Proposition \ref{p: per-unit subsidy generalized}.

\subsection{Total effects of budget-neutral reallocation}\label{app: total condition}

Proposition \ref{p: per-unit subsidy} and \ref{p: per-unit subsidy generalized} in the main text and the proofs in Appendix \ref{app: prop 1 derive} and \ref{app: prop 2 derive} have focused on the \emph{marginal} effects of changing per-unit subsidies across fuel types. We now examine the \emph{total} effects of a budget-neutral reallocation. Specifically, for a fixed government budget, we provide a condition under which one fuel-specific policy yields a greater \emph{total} emissions reduction than the other.

Let $G(\bm{\tau})$ denote total government expenditure under subsidy vector $\bm{\tau}$. For a subsidy policy targeting fuel type $g$, let $\tau_g(\bar{G})$ denote the per-unit subsidy level that yields total government expenditure $\bar{G}>0$, with the corresponding subsidy vector
\begin{equation*}
\bm{\tau}^g(\bar{G}) = (0,\dots,\tau_g(\bar{G}),\dots,0).
\end{equation*}
Let $E_g(\bar{G}) \equiv E(\bm{\tau}^g(\bar{G}))$ denote the aggregate emissions under this policy. Further, let $G_g(\tau)$ represent the total government expenditure required to sustain a per-unit subsidy level $\tau$ for fuel type $g$. 

The total emissions reduction under this policy is given by
\begin{equation}\label{eqn: total emission}
\begin{aligned}
\Delta E_g(\bar{G}) &= E(\bm{\tau}^g(0)) - E(\bm{\tau}^g(\bar{G})) \\
&= \int_0^{\tau_g(\bar{G})} \left. -\frac{\partial E(\bm{\tau})}{\partial \tau_g} \right\rvert_{\bm{\tau} = (0,\dots, \tau,\dots,0)} \, d\tau \\
&= \int_0^{\tau_g(\bar{G})} \left. -\frac{{\partial E(\bm{\tau})}/{\partial \tau_g}}{{\partial G(\bm{\tau})}/{\partial \tau_g}} \right\rvert_{\bm{\tau} = (0,\dots, \tau,\dots,0)} \frac{dG_g(\tau)}{d\tau} \, d\tau \\
&= \int_0^{\bar{G}} \left. -\frac{{\partial E(\bm{\tau})}/{\partial \tau_g}}{{\partial G(\bm{\tau})}/{\partial \tau_g}} \right\rvert_{\bm{\tau} = \bm{\tau}^g(G)} \, dG,
\end{aligned}
\end{equation}
where $\bm{\tau}^g(G) \equiv (0,\dots, \tau_g(G),\dots,0)$. The integrand represents the marginal emissions reduction per unit of government expenditure (marginal abatement return) at expenditure level $G$.

Consequently, the subsidy policy targeting fuel type $g_1$, given by $\bm{\tau}^{g_1}(\bar{G})$, reduces emissions more than the policy targeting fuel type $g_2$, given by $\bm{\tau}^{g_2}(\bar{G})$, if and only if
\begin{equation}\label{eqn: total emission comparison}
    \int_0^{\bar{G}} \underbrace{ \left. -\frac{{\partial E(\bm{\tau})}/{\partial \tau_{g_1}}}{{\partial G(\bm{\tau})}/{\partial \tau_{g_1}}} \right\rvert_{\bm{\tau} = \bm{\tau}^{g_1}(G)} }_{MAR_{g_1}(G)}\, dG 
    ~ > ~ 
    \int_0^{\bar{G}} \underbrace{ \left. -\frac{{\partial E(\bm{\tau})}/{\partial \tau_{g_2}}}{{\partial G(\bm{\tau})}/{\partial \tau_{g_2}}} \right\rvert_{\bm{\tau} = \bm{\tau}^{g_2}(G)} }_{MAR_{g_2}(G)} \, dG.
\end{equation}

Figure \ref{fig: total effect} illustrates condition \eqref{eqn: total emission comparison} in a setting where the marginal abatement return declines with total government expenditure. Observe that comparing the marginal abatement returns at the zero-subsidy baseline, as in Propositions \ref{p: per-unit subsidy} and \ref{p: per-unit subsidy generalized}, is equivalent to comparing the $y$-intercepts of $MAR_{g_1}(G)$ and $MAR_{g_2}(G)$. Similarly, comparing the marginal abatement returns at the current subsidy level would be equivalent to comparing the vertical lines at $\bar{G}$. Our empirical analysis of HEV- vs. BEV-focused subsidy policies yields results consistent with this illustration when $g_1$ and $g_2$ denote the HEV and BEV fuel types, respectively. Specifically, the first part of Section \ref{sec: counterfactual 1} shows that an HEV-focused subsidy policy yields a higher marginal abatement return at the zero-subsidy baseline than a BEV-focused policy (by the blue intercept). The second part of Section \ref{sec: counterfactual 1} shows that the HEV-focused policy also achieves a greater total emission reduction (by the blue-shaded area). 

\begin{figure}[t!]
    \centering
    \caption{Comparison of total effects}
    \includegraphics[width=0.7\linewidth]{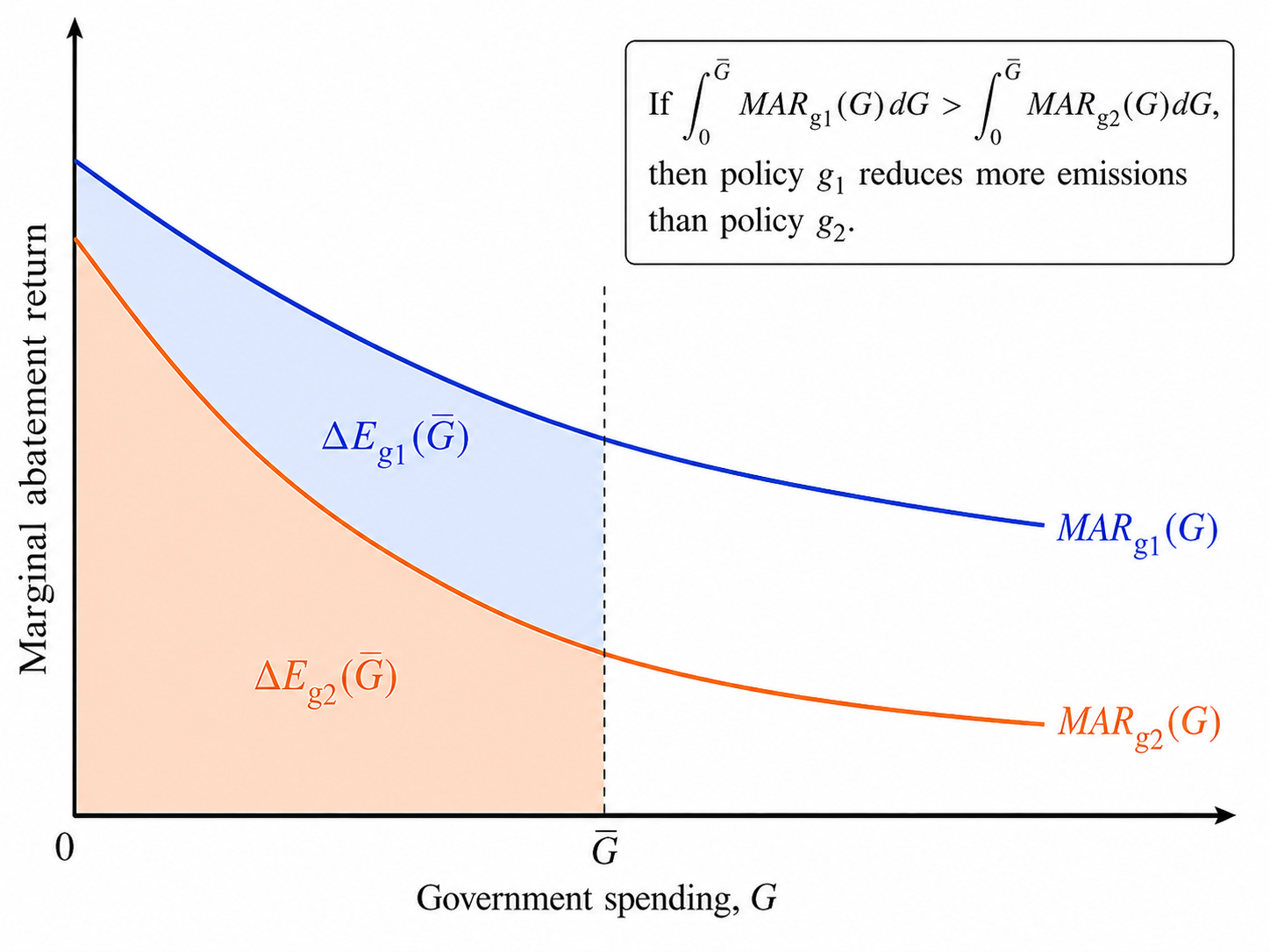}
    \label{fig: total effect}
\end{figure}

\subsection{Derivation of the subsidy pass-through matrix under Bertrand-Nash setup}\label{app: pass thru derive}

Fix a subsidy policy $\bm{\tau}=(\tau_1,\cdots, \tau_g, \cdots, \tau_G)$. Recall that $\mathbf{p}$ is a $J\times 1$ net consumer price vector whose element is given by $p_j=p_j^G- \tau_g$, where $p_j^G$ is a gross price set by firms and $\tau_g$ is the per-unit subsidy given to $j\in\mathscr{J}_g$. Let $B$ be a $J\times G$ subsidy assignment matrix whose $(j,g)$ entry is given by $B_{jg}=\bm{1}\{j\in\mathscr{J}_g\}$. Then, $\mathbf{p}=\mathbf{p}^G-B\bm{\tau}$.

From the Bertrand-Nash first-order condition \eqref{eqn: foc} in the main text, define a $J\times 1$ function $F(\mathbf{p},\bm{\tau})$ by 
\begin{equation*}
    F(\mathbf{p},\bm{\tau}) \equiv s(\mathbf{p}) + \nabla(\mathbf{p})(\mathbf{p} + B\bm{\tau} - \mathbf{mc}) = 0,
\end{equation*}
where $\nabla(\mathbf{p}) = \Omega\circ \nabla_p^s(\mathbf{p})$. Assume that $F$ is continuously differentiable in a neighborhood of $(\mathbf{p},\bm{\tau})$ and that $\frac{\partial F}{\partial\mathbf{p}'}$ is invertible. Then, by the Implicit Function Theorem, the $J\times G$ subsidy pass-through matrix $\Theta$ is given by

\begin{equation}\label{eqn: ivt}
    \Theta= 
    -\frac{\partial \mathbf{p}}{\partial \bm{\tau}} = \left( \frac{\partial F}{\partial \mathbf{p}'} \right)^{-1}\left( \frac{\partial F}{\partial\bm{\tau}} \right).
\end{equation}

First, observe that
\begin{equation}\label{eqn: dF dtau}
    \frac{\partial F}{\partial \bm{\tau}} = \nabla(\mathbf{p})B.
\end{equation}

We next examine $\frac{\partial F}{\partial \mathbf{p}'}$. The $j$-th row of $F(\mathbf{p},\bm{\tau})$ is given by
\begin{equation*}
    F_j(\mathbf{p},\bm{\tau}) = s_j(\mathbf{p}) + \sum_{l=1}^J\nabla(\mathbf{p})_{jl}\left(p_l + (B\bm{\tau})_{l}-mc_l\right).
\end{equation*}
Then, the $(j,k)$ entry of $\frac{\partial F}{\partial \mathbf{p}'}$ is given by
\begin{equation}\label{eqn: dF dp}
\begin{aligned}
    \left( \frac{\partial F}{\partial \mathbf{p}'} \right)_{jk} & = \frac{\partial F_j}{\partial p_k} \\
    & = \frac{\partial s_j}{\partial p_k} + \nabla(\mathbf{p})_{jk} + \sum_{l=1}^J\frac{\partial \nabla(\mathbf{p})_{jl}}{\partial p_k}(p_l + (B\bm{\tau})_{l} - mc_l) \\
    & = \frac{\partial s_j}{\partial p_k} + \nabla(\mathbf{p})_{jk} + \sum_{l=1}^J\Omega_{jl}\frac{\partial^2 s_l}{\partial p_k \partial p_j}(p_l + (B\bm{\tau})_{l} - mc_l).
\end{aligned}
\end{equation}

Writing equation \eqref{eqn: dF dp} in vector notation and substituting it together with \eqref{eqn: dF dtau} in \eqref{eqn: ivt} gives 
\begin{equation}\label{eqn: pass thru matrix}
    \Theta =
    \left(
    \underbrace{\big(\nabla_p^s(\mathbf{p})\big)'}_{\text{demand slope}}
    +
    \underbrace{\nabla(\mathbf{p})}_{\text{demand slope / ownership}}
    +
    \underbrace{K(\mathbf{p},\bm{\tau})}_{\text{demand curvature}} \right)^{-1}
    \left(\underbrace{\nabla(\mathbf{p})}_{\text{demand slope / ownership}}B\right), 
\end{equation}
where
\begin{equation*}
    (K(\mathbf{p},\bm{\tau}))_{jk} = \sum_l(p_l+ (B\bm{\tau})_{l} -mc_l)\Omega_{jl}\frac{\partial^2 s_l}{\partial p_k\partial p_j}.
\end{equation*}

\newpage
\section{HEVs and PHEVs}\label{sec: HEV and PHEV}

 An HEV, which utilizes both an internal combustion engine and an electric motor, is more fuel-efficient than conventional ICEVs (gasoline, diesel, and LPG). Its superior fuel economy is achieved through the following mechanisms:

\begin{enumerate}\addtolength{\itemsep}{-.2\baselineskip}
    \item \textbf{Regenerative braking:} When an HEV decelerates, its motor switches to generator mode, capturing the vehicle's kinetic energy and converting it into electricity. The HEV stores this energy, which would have been wasted in conventional ICEVs, to propel the vehicle later.
    \item \textbf{Engine optimization:} When the engine operates inefficiently (e.g., at low speeds or in stop-and-go traffic), an HEV runs exclusively on the electric motor. Conversely, when the engine generates more power than required while running in its high-efficiency range, the excess power is used to charge the battery.
    \item \textbf{Stop-start system:} An HEV automatically shuts off the engine whenever the vehicle comes to a stop and restarts it using the electric motor, thereby preventing fuel waste during idling.\footnote{Many standard ICEVs have also been equipped with stop-start systems in recent years.}
\end{enumerate}

 Higher fuel economy results in lower fuel consumption per km, which concurrently reduces both tank-to-wheel emissions (from combustion) and well-to-tank emissions (associated with fuel production and distribution). Consequently, HEVs generate lower fuel-cycle GHG emissions than conventional ICEVs.

 Similar to an HEV, a PHEV can charge its battery through the internal mechanisms described above. However, unlike HEVs, a PHEV can also be plugged into an external power source (the grid) to charge its battery, as Figure \ref{fig: HEV vs PHEV} shows.  
 
\begin{figure}[htbp]
	\centering
	\caption{HEV vs PHEV}
	\includegraphics[width=.8\textwidth]{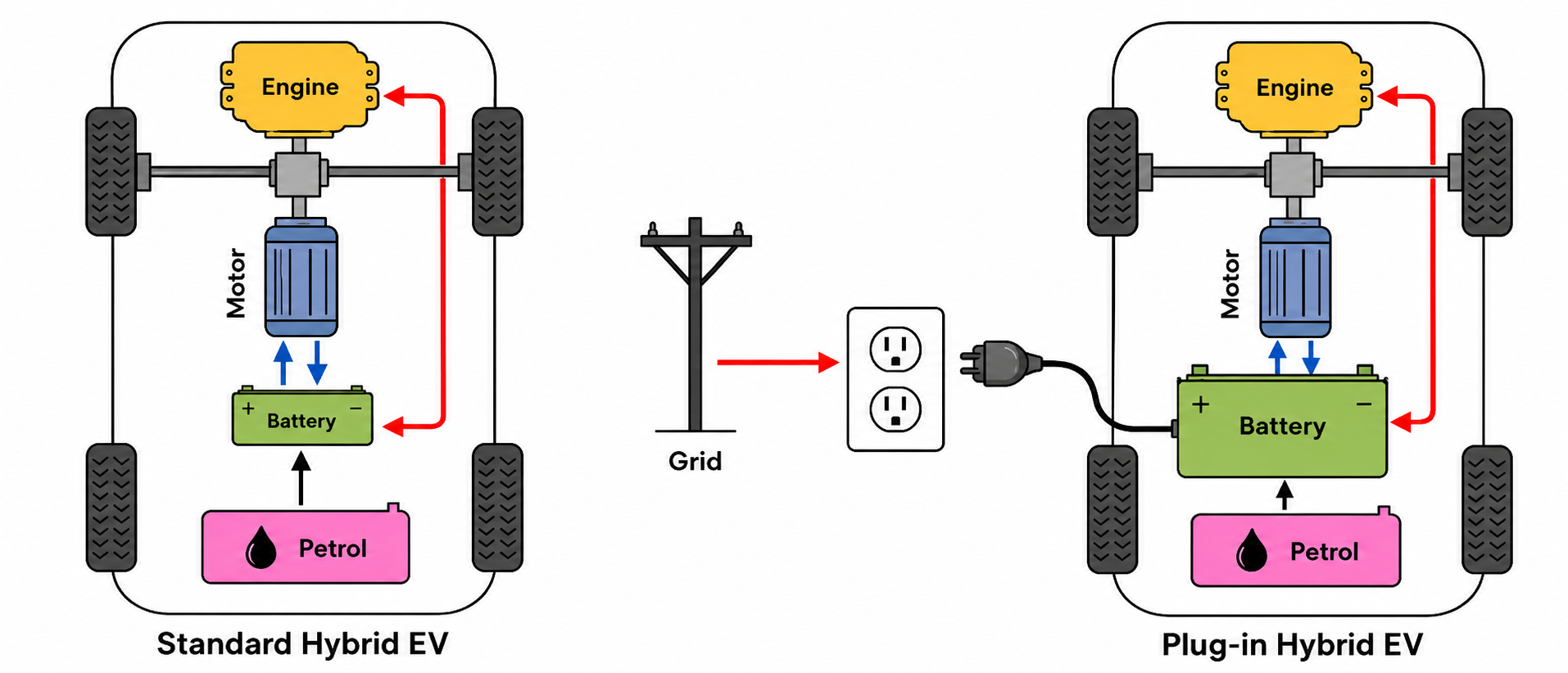}
\label{fig: HEV vs PHEV}
\end{figure}

\newpage
\section{Additional data description}\label{sec: Data description}

\renewcommand{\thefigure}{\thesection\arabic{figure}}
\renewcommand{\thetable}{\thesection\arabic{table}}
\setcounter{table}{0}
\setcounter{figure}{0}%

\subsection*{Charging network} 

We obtain charger data managed by the Korean Environment Corporation.\footnote{\url{https://www.keco.or.kr/en/}} These data contain charger-level information on the installation year and street address. We use four separate datasets, each recording chargers available at the end of each year from 2020 to 2023.\footnote{Each of these four datasets was retrieved from the Public Data Portal: \url{https://www.data.go.kr/en/index.do}} 

We calculate the number of chargers in each market (year-by-province level) as follows. First, for the years 2020 through 2023, we directly obtain the market-level charger counts from the respective datasets. Second, for the years prior to 2020, we reconstruct the number of available chargers in each market based on the installation years reported in the 2020 baseline dataset. For example, the total number of chargers in Seoul in 2015 is constructed by aggregating all chargers located in Seoul with an installation year less than or equal to 2015. 

Figure \ref{fig: charging station distribution} illustrates the evolution of the charging network over time in South Korea. Each circle represents the number of chargers at a given street-level address. The number of chargers began to increase rapidly in 2017, driven by government incentives and growing electric vehicle demand. Chargers were primarily installed around apartment complexes, public parking lots, public institutions, and commercial areas.

\subsection*{Country-year-level GHG emissions intensity of electricity generation} 

 The country-year-level GHG emissions intensity of electricity generation (in g~CO$_2$e/kWh) is obtained from \textit{Our World in Data}. We utilize the 2022 intensity levels for our counterfactual analysis in Subsection~\ref{sec: counterfactual 3}. Since these measures are not adjusted for grid losses (i.e., transmission losses are not accounted for), we merge this dataset with the country-year-level transmission loss rate data obtained from the \textit{World Bank}.\footnote{\textit{Our World in Data}: \url{https://ourworldindata.org/grapher/carbon-intensity-electricity}; \textit{World Bank} Data: \url{https://data.worldbank.org/indicator/EG.ELC.LOSS.ZS}}
 
 Accounting for transmission losses, the adjusted GHG emission factor is calculated by scaling the unadjusted intensity data:
\begin{equation*}
    \text{GHG emission factor} = \text{Unadjusted GHG emission factor} \cdot \left( \frac{100}{100 - \text{grid loss (\%)}} \right)
\end{equation*}
 Table \ref{tab: ghg emission intensity} provides the adjusted GHG emission factors for 148 countries and the world average, alongside the transmission loss rates (as a percentage) for 2022.

\newpage
\begin{figure}[htbp]
	\centering
	\caption{Charging network: 2013, 2017, 2020,
and 2023}
	\begin{subfigure}[b]{0.4\textwidth}	
    	\caption{2013}
    	\includegraphics[width=\textwidth]{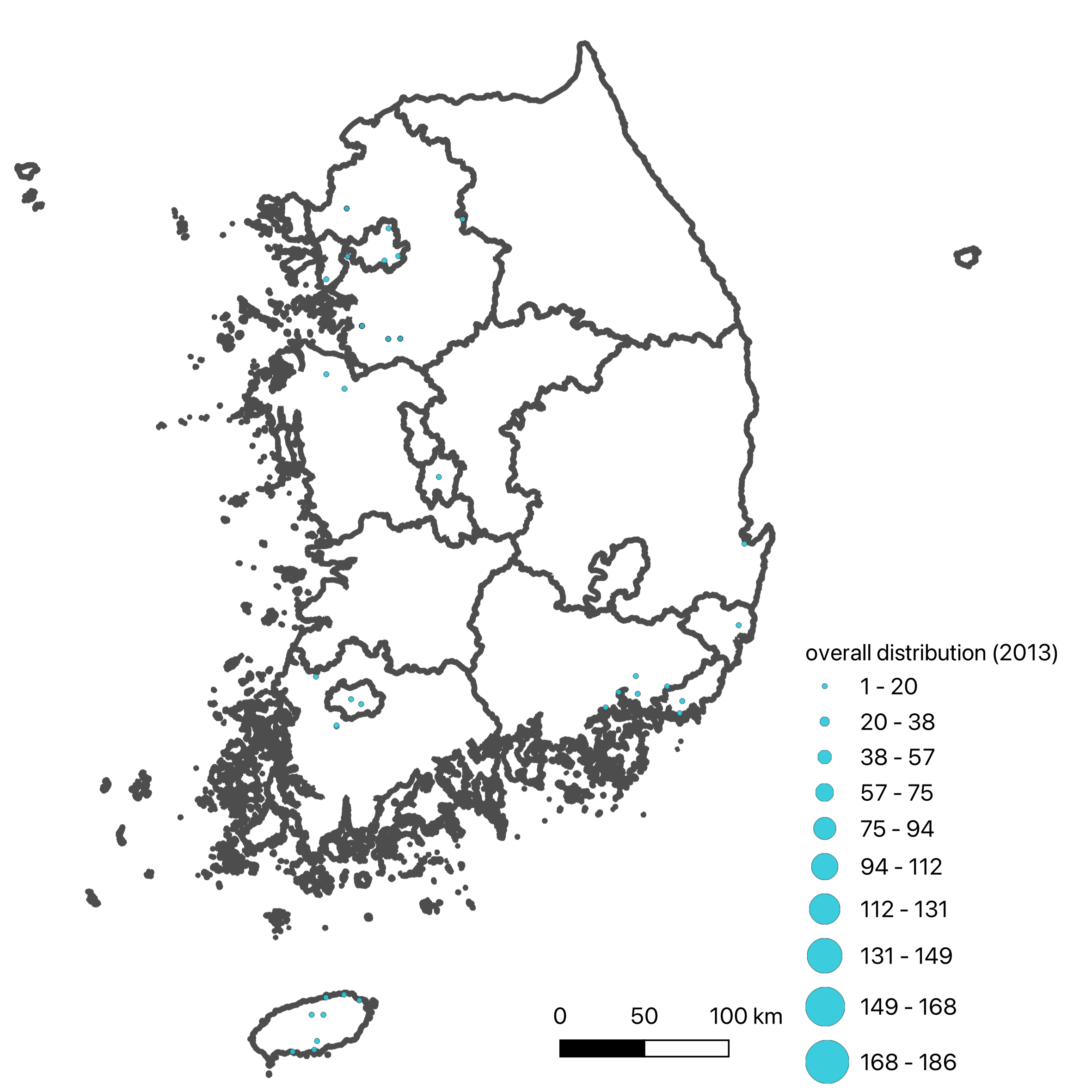}
    	\end{subfigure}
    \qquad
	\begin{subfigure}[b]{0.4\textwidth}	
    	\caption{2017}
    	\includegraphics[width=\textwidth]{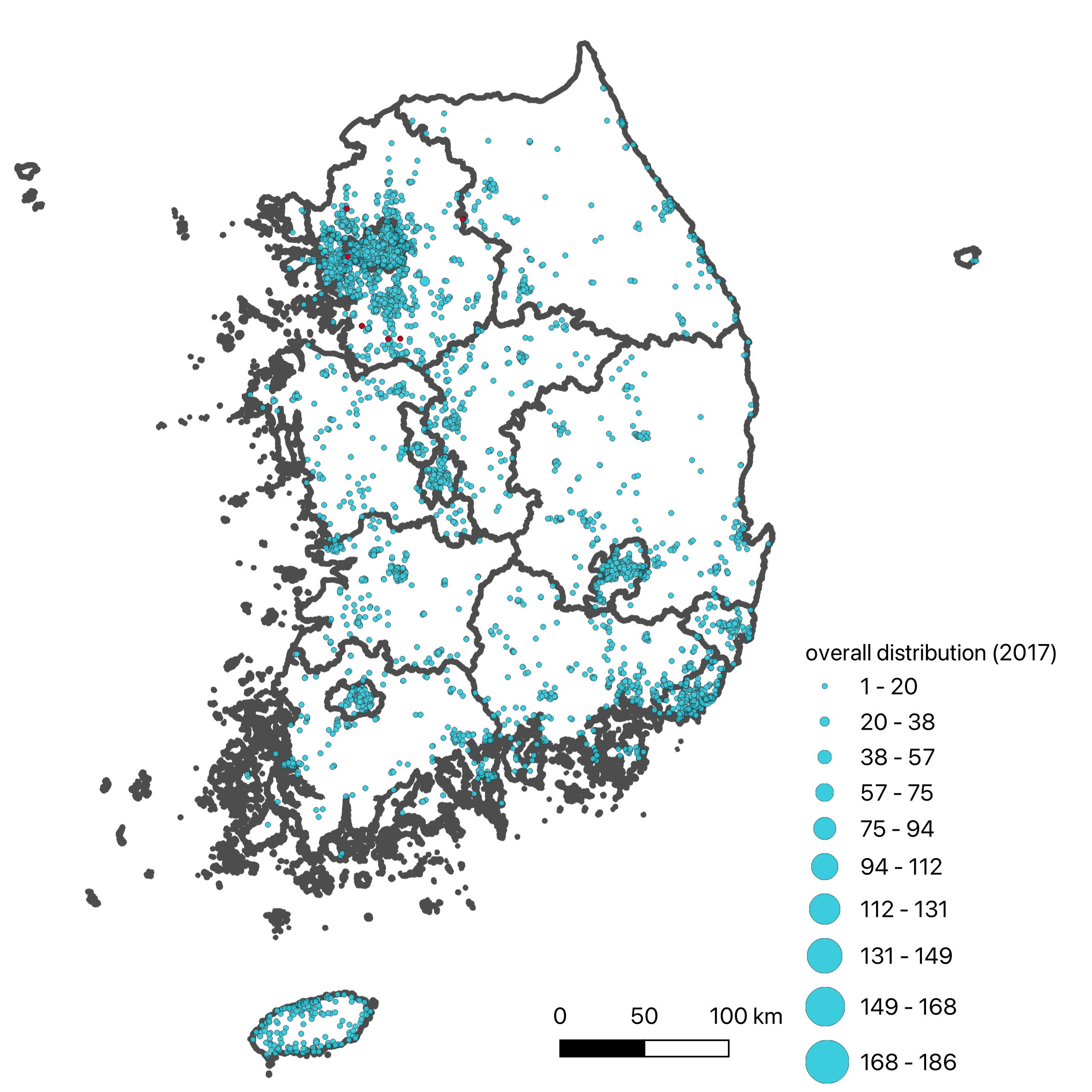}
    	\end{subfigure}
    	\begin{subfigure}[b]{0.4\textwidth}
    \vspace{0.2in}
    	\caption{2020}
    	\includegraphics[width=\textwidth]{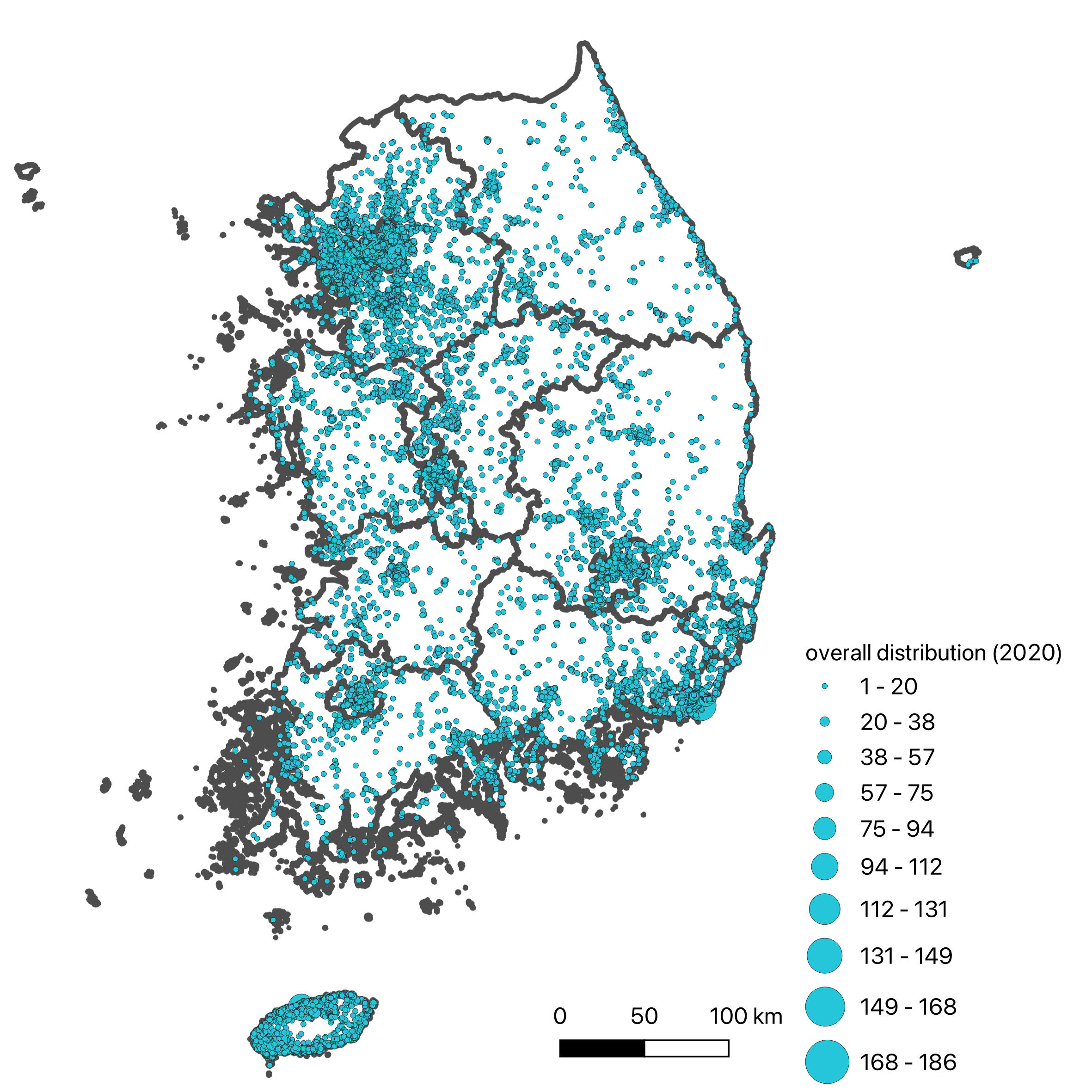}
    	\end{subfigure}
        \qquad
    	\begin{subfigure}[b]{0.4\textwidth}
    	\caption{2023}
    	\includegraphics[width=\textwidth]{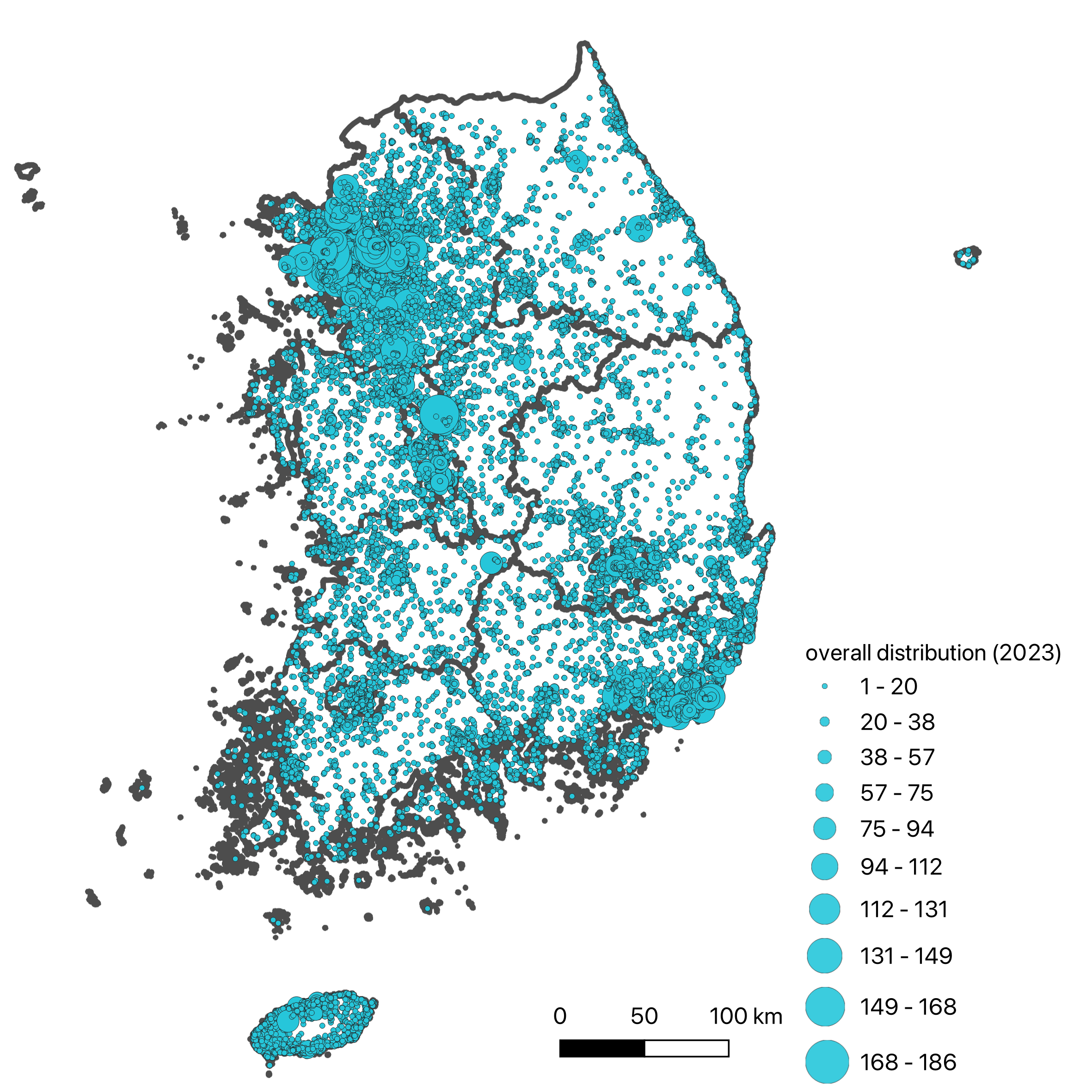}
    	\end{subfigure}
\label{fig: charging station distribution}
\tablenotes This figure illustrates the evolution of the charging network over time in South Korea. Each circle represents the number of chargers at a given street-level address.
\end{figure}
\clearpage

\newpage
\begin{table}[htbp]
\footnotesize
  \centering
  \caption{GHG emissions intensity of electricity generation and grid loss across countries in 2022}
    \begin{tabular}{lrrrlrrrrrr}
    \toprule
    \multicolumn{1}{c}{Code} & \multicolumn{1}{c}{Level} & \multicolumn{1}{c}{Loss (\%)} &       & \multicolumn{1}{c}{Code} & \multicolumn{1}{c}{Level} & \multicolumn{1}{c}{Loss (\%)} &       & \multicolumn{1}{c}{Code} & \multicolumn{1}{c}{Level} & \multicolumn{1}{c}{Loss (\%)} \\
\cmidrule{1-3}\cmidrule{5-7}\cmidrule{9-11}    WORLD & 526   & 6.87  &       & PER   & 324   & 11.10 &       & \multicolumn{1}{l}{MMR} & 613   & 8.02 \\
    CRI   & 26    & 8.39  &       & ROU   & 330   & 9.75  &       & \multicolumn{1}{l}{ISR} & 625   & 4.67 \\
    NPL   & 28    & 15.32 &       & NLD   & 339   & 4.00  &       & \multicolumn{1}{l}{HTI} & 628   & 14.92 \\
    PRY   & 28    & 11.77 &       & RWA   & 350   & 17.62 &       & \multicolumn{1}{l}{TGO} & 634   & 24.63 \\
    ALB   & 28    & 13.68 &       & KHM   & 353   & 13.65 &       & \multicolumn{1}{l}{SSD} & 636   & 3.90 \\
    ISL   & 28    & 2.68  &       & NIC   & 358   & 19.18 &       & \multicolumn{1}{l}{QAT} & 639   & 5.71 \\
    COD   & 30    & 9.04  &       & IRL   & 359   & 7.23  &       & \multicolumn{1}{l}{DOM} & 644   & 9.98 \\
    ETH   & 31    & 23.06 &       & CMR   & 360   & 21.97 &       & \multicolumn{1}{l}{MYS} & 649   & 6.90 \\
    NOR   & 31    & 4.54  &       & ZWE   & 384   & 22.48 &       & \multicolumn{1}{l}{BFA} & 649   & 12.46 \\
    CHE   & 40    & 7.44  &       & CHL   & 392   & 5.00  &       & \multicolumn{1}{l}{MUS} & 651   & 5.51 \\
    SWE   & 43    & 5.61  &       & ITA   & 406   & 6.74  &       & \multicolumn{1}{l}{PHL} & 667   & 9.59 \\
    UGA   & 71    & 19.32 &       & PRK   & 407   & 15.50 &       & \multicolumn{1}{l}{ERI} & 676   & 12.53 \\
    NAM   & 75    & 39.22 &       & LBN   & 421   & 9.76  &       & \multicolumn{1}{l}{GNQ} & 676   & 10.88 \\
    FRA   & 85    & 7.44  &       & GRC   & 423   & 10.88 &       & \multicolumn{1}{l}{KWT} & 700   & 8.97 \\
    TJK   & 106   & 18.24 &       & SUR   & 424   & 11.61 &       & \multicolumn{1}{l}{DZA} & 701   & 9.62 \\
    URY   & 121   & 8.82  &       & TZA   & 429   & 18.57 &       & \multicolumn{1}{l}{HKG} & 714   & 4.48 \\
    NZL   & 122   & 7.32  &       & USA   & 430   & 4.56  &       & \multicolumn{1}{l}{EGY} & 717   & 19.87 \\
    BRA   & 122   & 15.35 &       & HND   & 430   & 31.85 &       & \multicolumn{1}{l}{TTO} & 719   & 5.17 \\
    SLV   & 126   & 10.61 &       & GER   & 440   & 4.55  &       & \multicolumn{1}{l}{YEM} & 720   & 18.91 \\
    ZMB   & 126   & 11.46 &       & VNM   & 443   & 6.20  &       & \multicolumn{1}{l}{TUN} & 723   & 21.74 \\
    FIN   & 135   & 3.90  &       & KOR   & 453   & 3.20  &       & \multicolumn{1}{l}{IDN} & 725   & 6.81 \\
    KEN   & 142   & 22.85 &       & BLR   & 471   & 6.81  &       & \multicolumn{1}{l}{AZE} & 726   & 7.52 \\
    BEL   & 144   & 3.50  &       & RUS   & 479   & 8.56  &       & \multicolumn{1}{l}{MDA} & 726   & 10.88 \\
    LVA   & 146   & 7.18  &       & CIV   & 486   & 14.96 &       & \multicolumn{1}{l}{IRN} & 728   & 10.04 \\
    AUT   & 148   & 4.69  &       & ARG   & 494   & 20.51 &       & \multicolumn{1}{l}{BGD} & 734   & 7.79 \\
    SVK   & 149   & 4.92  &       & BGR   & 501   & 4.97  &       & \multicolumn{1}{l}{BIH} & 740   & 7.71 \\
    LUX   & 152   & 6.60  &       & MDG   & 504   & 5.28  &       & \multicolumn{1}{l}{MKD} & 747   & 16.36 \\
    MOZ   & 155   & 17.28 &       & SGP   & 506   & 0.28  &       & \multicolumn{1}{l}{JAM} & 765   & 26.61 \\
    CAN   & 163   & 4.01  &       & CZE   & 509   & 4.02  &       & \multicolumn{1}{l}{SAU} & 770   & 9.07 \\
    SWZ   & 173   & 26.43 &       & PAK   & 513   & 14.83 &       & \multicolumn{1}{l}{POL} & 778   & 5.80 \\
    KGZ   & 177   & 17.20 &       & GAB   & 529   & 18.89 &       & \multicolumn{1}{l}{SRB} & 806   & 11.77 \\
    GEO   & 181   & 7.97  &       & GHA   & 531   & 11.54 &       & \multicolumn{1}{l}{MAR} & 806   & 17.79 \\
    PAN   & 182   & 7.18  &       & BOL   & 534   & 9.03  &       & \multicolumn{1}{l}{ZAF} & 811   & 10.04 \\
    AGO   & 194   & 11.27 &       & TUR   & 535   & 8.85  &       & \multicolumn{1}{l}{CUB} & 820   & 21.95 \\
    DNK   & 213   & 5.61  &       & MEX   & 536   & 11.87 &       & \multicolumn{1}{l}{IND} & 826   & 14.91 \\
    ECU   & 213   & 15.01 &       & JPN   & 542   & 4.10  &       & \multicolumn{1}{l}{SYR} & 864   & 21.00 \\
    COL   & 219   & 7.71  &       & MLT   & 544   & 8.02  &       & \multicolumn{1}{l}{TCD} & 867   & 28.99 \\
    LTU   & 219   & 17.67 &       & LKA   & 557   & 8.46  &       & \multicolumn{1}{l}{KAZ} & 907   & 8.27 \\
    GTM   & 220   & 12.91 &       & MNE   & 564   & 14.36 &       & \multicolumn{1}{l}{MNG} & 910   & 14.08 \\
    ESP   & 238   & 8.59  &       & EST   & 565   & 13.74 &       & \multicolumn{1}{l}{BHR} & 931   & 3.02 \\
    HUN   & 248   & 7.77  &       & CYP   & 571   & 2.72  &       & \multicolumn{1}{l}{BEN} & 968   & 38.46 \\
    LAO   & 249   & 4.68  &       & NGA   & 579   & 14.08 &       & \multicolumn{1}{l}{BRN} & 981   & 8.98 \\
    PRT   & 251   & 10.24 &       & ARE   & 584   & 4.73  &       & \multicolumn{1}{l}{LBY} & 1\,040  & 20.17 \\
    ARM   & 258   & 6.58  &       & GIB   & 588   & 2.86  &       & \multicolumn{1}{l}{BWA} & 1\,119  & 24.22 \\
    VEN   & 262   & 31.14 &       & SEN   & 606   & 13.02 &       & \multicolumn{1}{l}{UZB} & 1\,175  & 4.55 \\
    SDN   & 262   & 22.08 &       & JOR   & 606   & 10.96 &       & \multicolumn{1}{l}{COG} & 1\,280  & 44.58 \\
    GBR   & 274   & 8.53  &       & THA   & 606   & 7.47  &       & \multicolumn{1}{l}{NER} & 1\,306  & 46.68 \\
    SVN   & 275   & 6.21  &       & OMN   & 607   & 9.78  &       & \multicolumn{1}{l}{TKM} & 1\,446  & 9.64 \\
    HRV   & 277   & 11.67 &       & CHN   & 607   & 3.42  &       & \multicolumn{1}{l}{IRQ} & 1\,690  & 59.25 \\
    UKR   & 285   & 9.99  &       & AUS   & 608   & 4.91  &       &       &       &  \\
    \bottomrule
    \end{tabular}%
  \label{tab: ghg emission intensity}%
\tablenotes This table provides the adjusted emissions intensity of electricity generation (g~CO$_2$e/kWh) for 148 countries and the world average, as well as the transmission loss rate (as a percentage) for 2022.
\end{table}%
\clearpage

\newpage
\section{Micro moments}\label{sec: micro moment description}
\setcounter{table}{0}
\setcounter{figure}{0}%

\subsection*{Second choice survey}

We utilize second choice responses of consumers to construct micro statistics matched to simulated model counterparts. As illustrated in \cite{berry2004differentiated} and \cite{conlon2025incorporating}, data on second choices help identify a random coefficient parameter that governs rich substitution patterns between products. For example, the correlation between the attributes of the first and second choice cars of a consumer is informative about the degree to which consumers substitute toward similar vehicles.

As described in the main text, two surveys oversample consumers who purchased low-market-share vehicles such as BEVs, HEVs, or some luxury foreign cars, and undersample those who purchased high-market-share vehicles. Consequently, the micro survey data and the aggregate car sales data are not perfectly consistent with one another. As noted in \cite{conlon2025incorporating}, this compatibility between aggregate and micro data is an issue when researchers utilize moment conditions from both data sources in estimation. For example, if the aggregate and micro data are incompatible, there may be no set of parameters that simultaneously satisfies the moment conditions derived from both aggregate and micro data.

To address this concern, we utilize the inverse sampling weights recorded by the data provider to correct underlying oversampling and undersampling issues in survey data. Based on market reports and sales information, the data provider calculates and provides an inverse sampling weight for each individual response. As expected, almost all responses with BEVs as the primary choice, and most responses with HEVs as the primary choice, have inverse sampling weights less than one.

Observed second choice micro statistics are constructed as follows. Let $j$ denote the first-choice product and $k$ the second-choice product of surveyed consumer $i$. Let $x$ be a vehicle attribute, such as a dummy for BEV, HEV, diesel, gasoline, or SUV, or driving cost measured as fuel cost per kilometer. Conditional on $j,k\neq 0$, an observed correlation of an attribute $x$ between the first and second choice cars is given by
\begin{equation}\label{eqn: second moment equation}
	\frac{\sum_{i}^{n}w_i(x_{ij} - \overline{x}_{j})(x_{ik}-\overline{x}_{k})}{\sqrt{\sum_{i}^{n}w_i(x_{ij} - \overline{x}_{j})^2}\sqrt{\sum_{i}^{n}w_i(x_{ik} - \overline{x}_{k})^2}},
\end{equation}
where $w_i$ is an inverse sampling weight for $i$ and the bar notation indicates a simple average.

The second choice survey conducted in 2023 includes responses from consumers who bought cars in 2022 and 2023, and the survey conducted in 2018 includes responses from consumers who bought cars in 2017 and 2018. For each of the vehicle attributes described above, we use equation \eqref{eqn: second moment equation} to compute the observed micro statistics separately for the 2017–2018 and 2022–2023 market samples. We then match these observed micro statistics to their model-simulated counterparts.

A model-simulated counterpart to the correlation between the first and the second choice characteristic can be computed as follows using the notation in \cite{conlon2025incorporating}. Let $\theta$ be a vector of all parameters and $v_p(\theta)$ be a micro part $p$ which is an expected value of micro value $v_{pijkt}$ for individual $i$ whose first choice is $j$ and second choice is $k$ in market $t$. Let $\mathcal{T}$ be a set of markets in which observed micro statistics are computed. In our case, $\mathcal{T}$ is either a set of markets in year 2017-2018 or a set of markets in year 2022-2023. For given $\theta$ and 1\,000 individual demographic and Halton draws in each market $t\in\mathcal{T}$ with equal weights, a micro part $v_p(\theta)$ can be computed as
\begin{equation*}
    v_p(\theta) = \frac{\sum_{t\in\mathcal{T}}\sum_{i\in I_t}\sum_{j\in\mathcal{J}_t\cup\{0\}}\sum_{k\in\mathcal{J}_t\cup\{0\}\setminus\{j\}}w_{it}\cdot s_{ijkt}(\theta)\cdot w_{d_pijkt}\cdot v_{pijkt}}{\sum_{t\in\mathcal{T}}\sum_{i\in I_t}\sum_{j\in\mathcal{J}_t\cup\{0\}}\sum_{k\in\mathcal{J}_t\cup\{0\}\setminus\{j\}}w_{it}\cdot s_{ijkt}(\theta)\cdot w_{d_pijkt}},
\end{equation*}
where $s_{ijkt}(\theta)$ is an individual $i$'s choice probability of first choice $j$ and second choice $k$ and $w_{it} = \frac{1}{1000}$. Since our survey is conditioned on the case in which both first and second choice are inside goods, we let $w_{d_pijkt} = \bm{1}\{j,k\neq 0\}$.\footnote{This survey weight is different across each combination of $i$, $j$, and $k$ since survey is not conducted in a perfectly random manner. In an ideal situation where we have an access to each survey weight across all combinations of $i$, $j$, and $k$, or to analytical survey weighting formula, we may put different weights across each combination to address underlying sampling issues in a survey. Since this situation is hardly a case, we instead use inverse sampling weights when we compute the observed micro statistics as described earlier.} Each micro value $v_{pijkt}$ generates a different micro part $v_p(\theta)$ denoted by below:
\begin{equation*}
\begin{aligned}
    & v_1(\theta) \quad \xleftarrow{} \quad v_{1ijkt} = x_{jt}\cdot x_{kt} \\
    & v_2(\theta) \quad \xleftarrow{} \quad v_{2ijkt} = x_{jt} \\
    & v_3(\theta) \quad \xleftarrow{} \quad v_{3ijkt} = x_{kt} \\
    & v_4(\theta) \quad \xleftarrow{} \quad v_{4ijkt} = x_{jt}^2 \\
    & v_5(\theta) \quad \xleftarrow{} \quad v_{5ijkt} = x_{kt}^2
\end{aligned}
\end{equation*}
Finally, a model-simulated counterpart to the correlation between the first and the second choice characteristics $x$ is given by the standard correlation formula:
\begin{equation*}
    f(v(\theta)) = \frac{v_1(\theta) - v_2(\theta)\cdot v_3(\theta)}{\sqrt{v_4(\theta) - \left(v_2(\theta)\right)^2} \sqrt{v_5(\theta) - \left(v_3(\theta)\right)^2}} 
\end{equation*}

An upper panel of Table \ref{tab: micromoment values} reports these observed and simulated values of micro statistics from our main demand estimation result with mileage heterogeneity presented in the right panel of Table \ref{tab: rc demand estimation results}.

\subsection*{Mileage data}
In addition to the second choice survey, we use KTSA's vehicle inspection data to help identify the mileage interaction parameters. Using samples for the entire period (2012-2023), we compute the average daily mileage (annual mileage divided by 365) for vehicles of each of the five fuel types: BEV, HEV, diesel, gasoline, and LPG. We then take LPG daily mileage as the reference point and compute the percentage difference between the daily mileage of each of the four remaining fuel types and that of LPG.

To compute model counterparts to observed micro statistics described above, we follow the notation used earlier. In this case, a micro part $v_p(\theta)$ is defined by
\begin{equation*}
    v_p(\theta) = \frac{\sum_{t\in\mathcal{T}}\sum_{i\in I_t}\sum_{j\in\mathcal{J}_t\cup\{0\}}w_{it}\cdot s_{ijt}(\theta)\cdot w_{d_pijt}\cdot v_{pijt}}{\sum_{t\in\mathcal{T}}\sum_{i\in I_t}\sum_{j\in\mathcal{J}_t\cup\{0\}}w_{it}\cdot s_{ijt}(\theta)\cdot w_{d_pijt}},
\end{equation*}
where $s_{ijt}(\theta)$ is an individual $i$'s choice probability of $j$ and $\mathcal{T}$ is a set of entire markets (2012-2023). Each combination of $v_{pijt}$ and a survey weight $w_{d_pijt}$ generates a different micro part $v_p(\theta)$:
\begin{equation*}
\begin{aligned}
    & v_1(\theta) \quad \xleftarrow{} \quad v_{1ijt} = \text{dailymileage}_i \text{ and } w_{d_1ijt} = \bm{1}\{ j \in \text{BEV}\} \\
    & v_2(\theta) \quad \xleftarrow{} \quad v_{2ijt} = \bm{1}\{ j \in \text{BEV}\} \text{ and } w_{d_2ijt} = \bm{1}\{ j \neq 0\} \\
    & v_3(\theta) \quad \xleftarrow{} \quad v_{3ijt} = \text{dailymileage}_i \text{ and } w_{d_3ijt} = \bm{1}\{ j \in \text{LPG}\} \\
    & v_4(\theta) \quad \xleftarrow{} \quad v_{4ijt} = \bm{1}\{ j \in \text{LPG}\} \text{ and } w_{d_4ijt} = \bm{1}\{ j \neq 0\}
\end{aligned}
\end{equation*}
Then, the model-simulated percentage difference between the expected daily mileage of BEVs and LPG is computed as
\begin{equation*}
    f(v(\theta)) = \left( \frac{v_1(\theta)\cdot v_4(\theta)}{v_2(\theta)\cdot v_3(\theta)} - 1 \right) \cdot 100
\end{equation*}
The percentage differences between the expected daily mileage of each of the remaining fuel types and that of LPG are computed analogously.

A lower panel of Table \ref{tab: micromoment values} reports these observed and simulated values of micro statistics from our main demand estimation result with mileage heterogeneity presented in the right panel of Table~\ref{tab: rc demand estimation results}.

\begin{table}[htbp]
  \centering
  \caption{Observed and simulated values of micro moments}
    \begin{tabular}{lcrrr}
    \hline
    Moment & Market & Observed & Simulated & Difference \bigstrut\\
    \hline
    \textit{Second choice} &       &       &       &  \bigstrut[t]\\
    $\mathbb{CORR}(\text{Cost per km}_{jt},\text{Cost per km}_{kt} \mid j,k\neq 0)$ & 2022-2023 & 0.551 & 0.472 & 0.079 \\
          & 2017-2018 & 0.519 & 0.364 & 0.155 \\
    $\mathbb{CORR}(\text{BEV}_{jt},\text{BEV}_{kt} \mid j,k\neq 0)$ & 2022-2023 & 0.662 & 0.650 & 0.012 \\
          & 2017-2018 & 0.707 & 0.638 & 0.068 \\
    $\mathbb{CORR}(\text{HEV}_{jt},\text{HEV}_{kt} \mid j,k\neq 0)$ & 2022-2023 & 0.380 & 0.410 & -0.030 \\
          & 2017-2018 & 0.301 & 0.247 & 0.054 \\
    $\mathbb{CORR}(\text{Diesel}_{jt},\text{Diesel}_{kt} \mid j,k\neq 0)$ & 2022-2023 & 0.382 & 0.330 & 0.051 \\
          & 2017-2018 & 0.568 & 0.519 & 0.049 \\
    $\mathbb{CORR}(\text{Gasoline}_{jt},\text{Gasoline}_{kt} \mid j,k\neq 0)$ & 2022-2023 & 0.453 & 0.474 & -0.021 \\
          & 2017-2018 & 0.558 & 0.509 & 0.049 \\
    $\mathbb{CORR}(\text{SUV}_{jt},\text{SUV}_{kt} \mid j,k\neq 0)$ & 2022-2023 & 0.540 & 0.590 & -0.050 \\
          & 2017-2018 & 0.655 & 0.618 & 0.037 \bigstrut[b]\\
    \hline
    \textit{Mileage:} &       &       &       &  \bigstrut[t]\\
    $\frac{\mathbb{E}(\text{mileage}_i \mid j\in\text{BEV}) - \mathbb{E}(\text{mileage}_i \mid j\in\text{LPG})}{\mathbb{E}(\text{mileage}_i \mid j\in\text{LPG})}\times 100$ & 2012-2023 & 13.3838 & 13.3841 & -0.0003 \\
          &       &       &       &  \\
    $\frac{\mathbb{E}(\text{mileage}_i \mid j\in\text{HEV}) - \mathbb{E}(\text{mileage}_i \mid j\in\text{LPG})}{\mathbb{E}(\text{mileage}_i \mid j\in\text{LPG})}\times 100$ & 2012-2023 & 9.3668 & 9.3668 & 0.0000 \\
          &       &       &       &  \\
    $\frac{\mathbb{E}(\text{mileage}_i \mid j\in\text{Diesel}) - \mathbb{E}(\text{mileage}_i \mid j\in\text{LPG})}{\mathbb{E}(\text{mileage}_i \mid j\in\text{LPG})}\times 100$ & 2012-2023 & 5.9953 & 5.9953 & 0.0000 \\
          &       &       &       &  \\
    $\frac{\mathbb{E}(\text{mileage}_i \mid j\in\text{Gasoline}) - \mathbb{E}(\text{mileage}_i \mid j\in\text{LPG})}{\mathbb{E}(\text{mileage}_i \mid j\in\text{LPG})}\times 100$ & 2012-2023 & -19.6870 & -19.6870 & 0.0000 \bigstrut[b]\\
    \hline
    \end{tabular}%
  \label{tab: micromoment values}%
  \tablenotes This table presents observed and simulated values of micro moments used in the estimation of the random coefficients model with mileage heterogeneity (RC logit 2). The size of the micro data is as follows: 4\,145 for second choice data for 2022-2023 market, 4\,035 for 2017-2018 market, and 69\,707\,543 for mileage data.
\end{table}%
\clearpage

\newpage

\section{Calculation of vehicle-level and market-level GHG emissions}

\subsection{Calculation of vehicle-cycle and fuel-cycle GHG emissions}\label{sec: Greenhouse gas emissions of vehicles}

\setcounter{table}{0}
\setcounter{figure}{0}%

 Evaluating the environmental effects of BEV and HEV subsidies requires calculating the vehicle-cycle and fuel-cycle GHG emissions of each product-year combination.

\subsubsection*{Vehicle-cycle GHG emissions}

 Unfortunately, data on vehicle-cycle emissions are not available for South Korea. Instead, we adopt the estimates of \cite{kelly2023cradle} for each fuel type of sedans and SUVs sold in the U.S., shown in Figure \ref{fig: vehicle manufacturing GHG emissions}. In doing so, we assume that vehicle-cycle energy use and the resulting GHG emissions in Korea are comparable to those in the U.S., which is reasonable given the similarity in fossil fuel generation shares between the two countries; according to Table \ref{tab: electricity generation mix}, the share is 59.6\% in the U.S. in 2020 and 58.2\% in Korea in 2023. 
 
 For each BEV body type (sedan and SUV), \cite{kelly2023cradle} reports three emission estimates based on driving ranges of 200, 300, and 400 miles. According to Figure \ref{fig: vehicle manufacturing GHG emissions}, producing a 400-mile-range BEV generates nearly twice the GHG emissions of producing an ICEV, while the corresponding emission gap for a 200-mile-range BEV varies between 25\% and 32\%. We utilize linear interpolation to calculate the vehicle-cycle GHG emissions for each product-year combination based on the driving range information sourced from CARISYOU, a South Korean automotive data provider.

\begin{table}[htbp]
  \centering
  \caption{Electricity generation mix}
    \begin{tabular}{lrr}
    \hline
    \hline
    Power source & U.S. & Korea \bigstrut\\
    \hline
    Natural gas & 36.8 & 26.8 \bigstrut[t]\\
    Coal  & 22.8 & 31.4 \\
    Nuclear & 20.3 & 30.7 \\
    Renewables & 19.4 & 8.4 \\
    Other sources  & 0.7 & 2.6 \bigstrut[b]\\
    \hline
    \hline
    \end{tabular}%
  \label{tab: electricity generation mix}%
\tablenotes This table presents the electricity generation mix of the U.S. in 2020 and Korea in 2023. Other sources include biomass, hydrogen, and residual oil. Sources: \cite{kelly2023cradle}; The 11th Basic Plan for Long-Term Electricity Supply and Demand (2024–2038), Ministry of Trade, Industry and Energy of Korea, 2025.    
\end{table}%

\subsubsection*{Fuel-cycle GHG emissions}

 \begin{enumerate}
 
 \item ICEVs and HEVs: We calculate the fuel-cycle GHG emissions per km (in g CO$_2$e/km) for each ICEV and HEV as follows. First, for each product, we collect the official trim-year-level tailpipe (tank-to-wheels) $\text{CO}_2$ emissions (in g $\text{CO}_2$/km) from CARISYOU and calculate the median for each year. Because official $\text{CO}_2$ emissions tend to understate actual emissions, we adjust these figures using benchmarks provided by the \citet{ICCT2024}. Specifically, given that the gap between official and actual $\text{CO}_2$ emissions averaged approximately 14\% in the European market in 2022, we multiply the median by $1.14$ and adopt this figure as the product-year-level tailpipe $\text{CO}_2$ emissions.
 
 Next, given that CO$_2$ emissions constitute 95\% or more of total GHG emissions from transportation,\footnote{The other two GHG emissions are CH$_4$ and N$_2$O. See ``Emissions of Carbon Dioxide in the Transportation Sector,'' Congressional Budget Office, December 2022; and ``Greenhouse Gas Emissions from a Typical Passenger Vehicle,'' Environmental Protection Agency, March 2018 for more details on tailpipe GHG emissions.} we convert the tailpipe CO$_2$ emissions (in g CO$_2$/km) into tailpipe GHG emissions (in g CO$_2$e/km) by applying a conversion factor of 1.05. For instance, the 2020 Toyota Camry HEV is available in two trims (LE and XLE) with CO$_2$ emissions of 92 and 95~g/km, respectively. We calculate the tank-to-wheels GHG emissions for the 2020 Toyota Camry HEV as 111.92 g CO$_2$e/km ($93.5 \times 1.14 \times 1.05$).
 
 Finally, tank-to-wheels GHG emissions constitute the majority of the fuel-cycle GHG emissions for ICEVs and HEVs (e.g., \citealp{elgowainy2016cradle, prussi2020jec, choi2020greenhouse}). Based on \cite{choi2020greenhouse},\footnote{They modified the GREET model to reflect the specific conditions and energy policies of South Korea, and found that the ratio of tank-to-wheels GHG emissions to the fuel-cycle GHG emissions ranges between 85.6\% to 86.7\% across fuel types.} we assume that tank-to-wheels GHG emissions constitute 86\% of fuel-cycle GHG emissions for an ICEV or HEV in South Korea, and therefore divide the former by 0.86 to obtain the latter. Continuing with the above example, the fuel-cycle GHG emissions per km for the 2020 Toyota Camry HEV are calculated as 130.14 g CO$_2$e/km (111.92 / 0.86).
    
 \item BEVs: The well-to-tank GHG emissions for a BEV are calculated by dividing the GHG emission factor for South Korea (453 g CO$_2$e/kWh in 2022) by the vehicle's fuel economy (in km/kWh). Since a BEV generates no tailpipe emissions, its well-to-tank GHG emissions are equivalent to its fuel-cycle GHG emissions.

 \end{enumerate}
 
In this way, we calculate the fuel-cycle GHG emissions per km for the 2\,304 product-year combinations in our dataset.\footnote{There are 61 product-year combinations in the 2012--2019 period for which we cannot retrieve the official trim-year-level tailpipe $\text{CO}_2$ emissions from CARISYOU. We exclude these combinations when calculating GHG emissions from the outside option.} 

\newpage
\begin{figure}[htbp]
	\centering
	\caption{Vehicle manufacturing cycle GHG emissions }
	\begin{subfigure}[b]{.8\textwidth}	
    	\caption{Sedan}
    	\includegraphics[width=\textwidth]{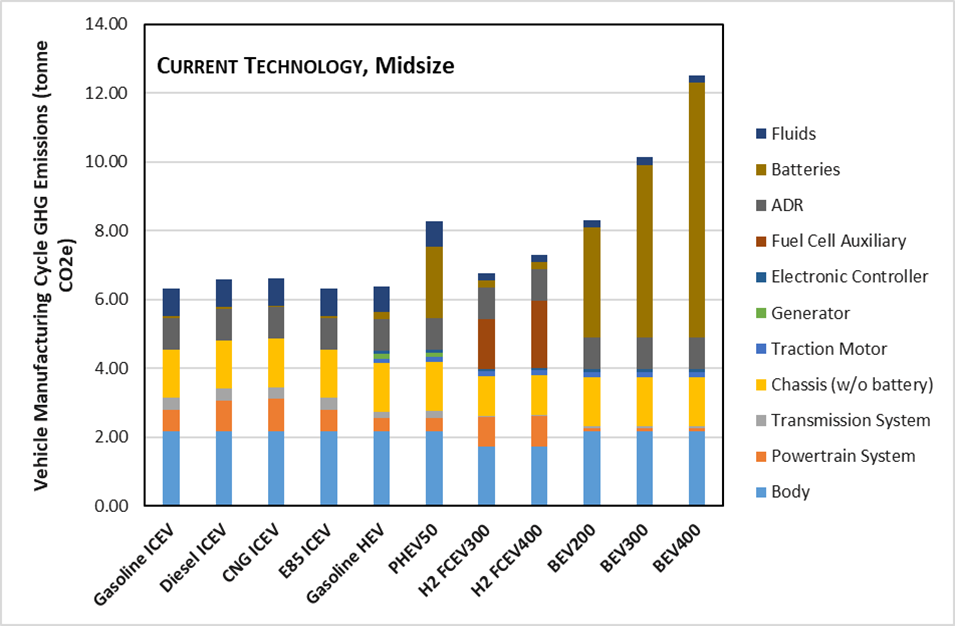}
    	\end{subfigure}
	\begin{subfigure}[b]{.8\textwidth}	
        \vspace{0.2in}
    	\caption{SUV}
    	\includegraphics[width=\textwidth]{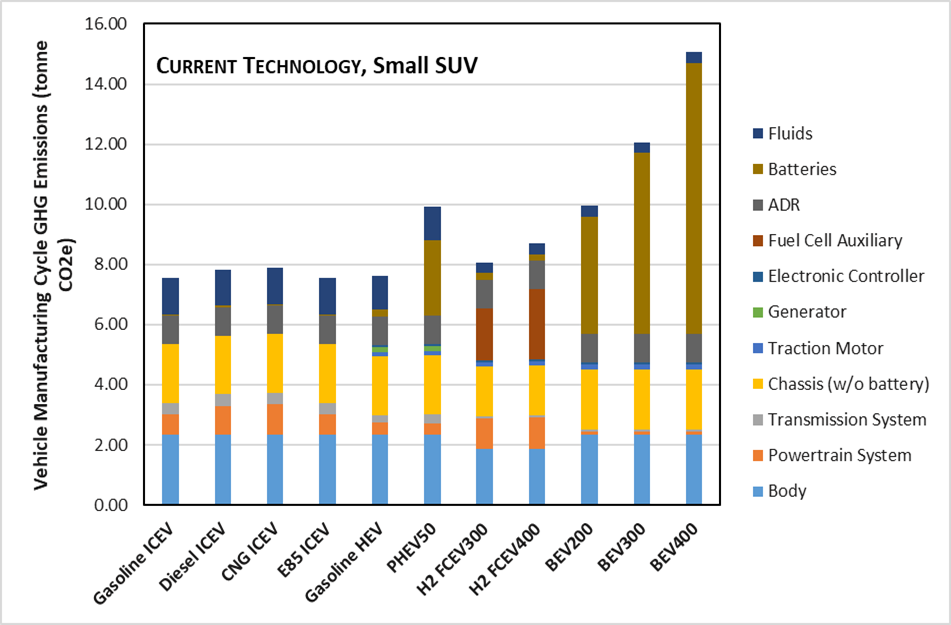}
    	\end{subfigure}
\label{fig: vehicle manufacturing GHG emissions}
\tablenotes The top panel (bottom panel) of this figure presents vehicle-cycle GHG emissions for a sedan (an SUV) in the U.S. Source: \cite{kelly2023cradle}, pp. 137--138.
\end{figure}
\clearpage


\subsection{Estimation of market-level GHG emissions changes}\label{sec: Calculation of GHG emissions changes}

 We estimate the market-level GHG emissions \eqref{eqn: mkt emissions} as
\begin{equation}\label{eqn: GHG emissions}
\widehat{E}^{\tau}_m =M_m \frac{1}{N} \sum_{i=1}^N  E_i^{\tau}.
\end{equation}

Using \eqref{eqn: consumer/product emissions}, the change in cumulative GHG emissions in market $m$ following the policy transition from $\tau$ to $\tau'$, $\Delta E^{(\tau,\tau')}_m$, is estimated as
\begin{equation}\label{eqn: GHG emissions change}
	\widehat{\Delta E}^{(\tau,\tau')}_m = \widehat{\Delta E}_{m,\text{FC}}^{(\tau,\tau')} + \widehat{\Delta E}_{m,\text{VC}}^{(\tau,\tau')} + \widehat{\Delta E}_{m,\text{outside}}^{(\tau,\tau')},
\end{equation}
 where
\begin{equation*}
\begin{aligned}
	& \widehat{\Delta E}_{m,\text{FC}}^{(\tau,\tau')} =  \frac{M_m}{N}\sum_{i=1}^N\sum_{j\in\mathcal{J}} T \cdot m_i\cdot e_j^{FC} \cdot \Delta \hat{s}_{ij} \\
	& \widehat{\Delta E}_{m,\text{VC}}^{(\tau,\tau')} = \frac{M_m}{N}\sum_{i=1}^N\sum_{j\in\mathcal{J}} e_j^{VC} \cdot \Delta \hat{s}_{ij} \\
	& \widehat{\Delta E}_{m,\text{outside}}^{(\tau,\tau')} = \frac{M_m}{N}\sum_{i=1}^N T^O \cdot m_i\cdot e_0^{FC} \cdot \Delta \hat{s}_{i0}
\end{aligned}
\end{equation*}
 where $\Delta \hat{s}_{ij} = \hat{s}^{\tau'}_{ij} - \hat{s}^{\tau}_{ij}$ denotes the change in the consumer’s choice probability after the policy transition from $\tau$ to $\tau'$. Equation \eqref{eqn: GHG emissions change} decomposes the change in emissions within the market into three terms: changes in (i) fuel-cycle emissions ($\widehat{\Delta E}_{m,\text{FC}}^{(\tau,\tau')}$), (ii) vehicle-cycle emissions ($\widehat{\Delta E}_{m,\text{VC}}^{(\tau,\tau')}$), and (iii) outside emissions, which are the emissions from the outside option ($\widehat{\Delta E}_{m,\text{outside}}^{(\tau,\tau')}$).

 We can further decompose the changes in fuel-cycle (as well as vehicle-cycle) emissions by fuel type: ICEVs (combining gasoline, diesel, and LPG), BEVs, and HEVs. Specifically, 
 \begin{equation*}
 \widehat{\Delta E}_{m,\text{FC}}^{(\tau,\tau')} = \frac{M_m}{N}\sum_{i=1}^N T \cdot m_i \left( \sum_{j\in\mathscr{J}_{\text{ICEV}}}  e_j^{FC} \cdot \Delta \hat{s}_{ij} + \sum_{j\in\mathscr{J}_{\text{BEV}}}  e_j^{FC} \cdot \Delta \hat{s}_{ij} + \sum_{j\in\mathscr{J}_{\text{HEV}}}  e_j^{FC} \cdot \Delta \hat{s}_{ij} \right).
 \end{equation*}
 This decomposition allows us to isolate the individual contribution of each fuel type to the overall change in fuel-cycle emissions.

\newpage
\setcounter{table}{0}
\setcounter{figure}{0}%

\section{Mileage heterogeneity and policy effects}\label{sec: Mileage heterogeneity and policy effects}

 In Section~\ref{sec: model and empiric}, we show that a consumer's driving intensity plays a critical role in their fuel-type choice. Furthermore, a subsidy policy induces a disproportionately strong response among high-mileage consumers, causing them to shift away from conventional ICEVs toward subsidized vehicles (either HEVs or BEVs) at a higher rate than low-mileage consumers. 
 
To illustrate this point, we simulate the 2023 shifts in within-fuel-type market shares under two distinct counterfactual scenarios: a BEV-focused subsidy and an HEV-focused subsidy. According to Figure~\ref{fig: mean_ci_2023_change}, the increase in the market share of the subsidized fuel type following subsidy implementation is larger among consumers with higher driving intensity under both policies. For instance, the rise in the HEV market share among consumers in the 40--60~km daily mileage range (2.1 percentage points) is 76\% larger than that among consumers in the 0--20~km range (1.2 percentage points).
  
The RC logit 1 specification (which imposes the homogeneous mileage assumption) fails to capture these shifts in the consumer mileage composition across fuel types following policy implementation. Consequently, this restricted specification understates the environmental effectiveness of the subsidies, as shown in Table~\ref{tab: alternative specifications/assumptions}.


\begin{figure}[htbp]
    \centering
    \caption{Changes in within-group market shares for BEVs and HEVs (2023)}
    \includegraphics[width=.6\textwidth]{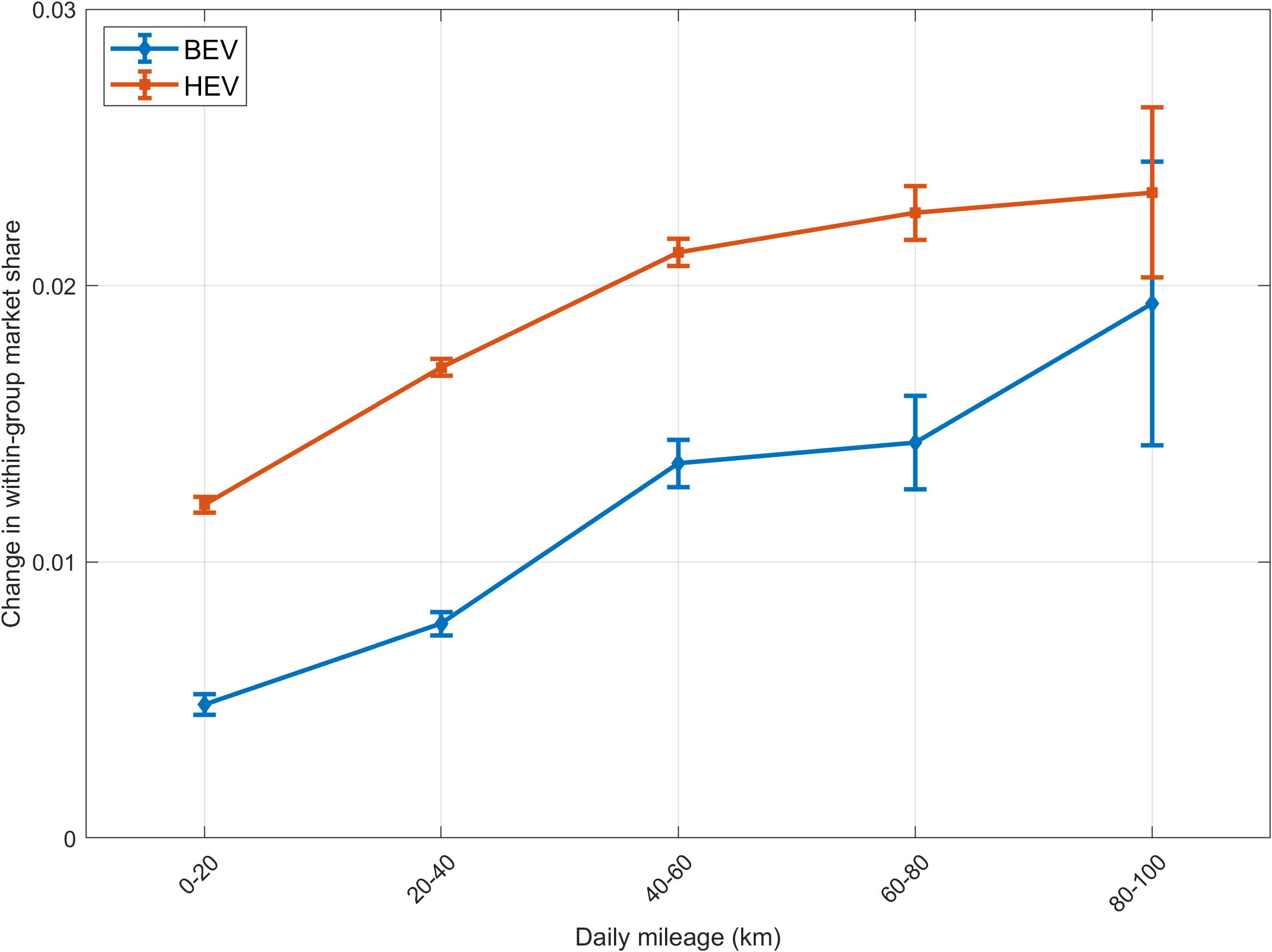}
    \label{fig: mean_ci_2023_change}
\tablenotes This figure presents the nationwide averages and two-standard-deviation intervals of the predicted changes in choice probabilities across five daily mileage groups. The panels compare two counterfactual policy transitions relative to a baseline of no subsidies: the change in BEV choice probability under BEV subsidies, and the corresponding change in HEV choice probability under HEV subsidies. For reference, the nationwide averages are equivalent to the changes in within-group market shares.
\end{figure}
\clearpage

\newpage
\setcounter{table}{0}
\setcounter{figure}{0}%

\section{Additional figures and tables}\label{sec: Additional figures and tables}

\begin{table}[!h]
  \centering
  \caption{BEV nameplates and subsidy}
    \begin{tabular}{rrrrrrr}
    \toprule
          & \multicolumn{3}{c}{\# BEV registered} &       & \multicolumn{2}{c}{\# nameplate} \\
\cmidrule{2-4}\cmidrule{6-7}    \multicolumn{1}{l}{year} & \multicolumn{1}{l}{total} & \multicolumn{1}{l}{with subsidy} & \multicolumn{1}{l}{without subsidy} &       & \multicolumn{1}{l}{with subsidy} & \multicolumn{1}{l}{without subsidy} \\
    \midrule
    2012  & 4     & 4     & 0     &       & 1     & 0 \\
    2013  & 125   & 125   & 0     &       & 3     & 0 \\
    2014  & 507   & 507   & 0     &       & 5     & 0 \\
    2015  & 2\,042  & 2\,042  & 0     &       & 5     & 0 \\
    2016  & 2\,612  & 2\,612  & 0     &       & 6     & 0 \\
    2017  & 7\,149  & 7\,149  & 0     &       & 9     & 0 \\
    2018  & 22\,183 & 22\,104 & 79    &       & 10    & 1 \\
    2019  & 23\,163 & 22\,954 & 209   &       & 11    & 1 \\        
    2020  & 19\,100 & 18\,591 & 509   &       & 12    & 2 \\
    2021  & 35\,684 & 34\,454 & 1\,230  &       & 16    & 8 \\
    2022  & 62\,172 & 60\,979 & 1\,193  &       & 22    & 7 \\
    2023  & 65\,851 & 62\,379 & 3\,472  &       & 25    & 11 \\
    \bottomrule
    \end{tabular}%
  \label{tab: year sub nameplate}%
\tablenotes This table presents the number of BEVs and nameplates, with and without subsidies, for each year between 2012 and 2023. Source: The official BEV policy website of South Korean government.
\end{table}%
\clearpage

\newpage
\begin{table}[htbp]
  \centering
  \caption{Average daily mileage by province and fuel type}
    \begin{tabular}{lrrrrrr}
    \toprule
          & \multicolumn{5}{c}{Fuel-type} \\
\cmidrule{2-6}    Province & \multicolumn{1}{c}{Gasoline} & \multicolumn{1}{c}{Diesel} & \multicolumn{1}{c}{Hybrid} & \multicolumn{1}{c}{BEV} & \multicolumn{1}{c}{LPG}  & \multicolumn{1}{c}{Total} \\
    \midrule
    Seoul & 26.8  & 35.5  & 35.8  & 37.0  & 32.5  & 29.7 \\
    Busan & 28.8  & 37.7  & 39.0  & 39.2  & 34.8  & 31.8 \\
    Daegu & 28.9  & 38.9  & 39.8  & 43.7  & 35.6  & 32.4 \\
    Incheon & 29.0  & 38.0  & 39.8  & 41.2  & 34.5  & 32.3 \\
    Gwangju & 29.4  & 39.0  & 39.5  & 41.3  & 37.3  & 33.4 \\
    Daejeon & 28.5  & 38.7  & 39.7  & 41.0  & 36.0  & 32.2 \\
    Ulsan & 28.7  & 36.6  & 38.7  & 40.1  & 35.2  & 31.5 \\
    Sejong & 32.6  & 43.5  & 45.4  & 48.5  & 41.4  & 36.9 \\
    Gyeonggi & 30.2  & 39.5  & 41.1  & 42.8  & 36.9  & 33.7 \\
    Gangwon & 29.8  & 40.3  & 43.0  & 42.6  & 37.8  & 34.0 \\
    Chungbuk & 30.5  & 40.3  & 42.5  & 43.8  & 38.4  & 34.3 \\
    Chungnam & 31.5  & 41.0  & 43.2  & 43.7  & 39.8  & 35.3 \\
    Jeonbuk & 30.6  & 41.0  & 42.4  & 44.9  & 38.8  & 34.9 \\
    Jeonnam & 31.9  & 41.5  & 43.4  & 43.8  & 40.0  & 36.1 \\
    Gyeongbuk & 30.7  & 40.4  & 43.3  & 44.3  & 39.0  & 34.4 \\
    Gyeongnam & 30.9  & 40.3  & 42.6  & 43.0  & 38.3  & 34.3 \\
    Jeju  & 27.6  & 35.3  & 37.1  & 40.5  & 33.4  & 30.8 \\
    \bottomrule
    \end{tabular}%
  \label{tab: mileage by province and fuel type}%
\tablenotes This table presents the average daily mileage (km) by province and fuel type.
\end{table}%
\clearpage

\newpage
\begin{table}[htbp]
  \centering
  \caption{Estimation results of logit and nested logit demand models}
\begin{tabular}{lccc}
    \hline
    \hline
    Variable & OLS Logit & IV Logit & Nested Logit \bigstrut\\
    \hline
    Price / Avg. Income & -0.406 & -6.013 & -2.076 \bigstrut[t]\\
          & (0.062) & (0.739) & (0.411) \\
    Cost per km & -2.522 & -0.668 & -0.145 \\
          & (0.090) & (0.227) & (0.088) \\
    Acceleration & -0.271 & 18.033 & 6.229 \\
          & (0.461) & (2.659) & (1.290) \\
    Size & 0.324 & 0.859 & 0.296 \\
          & (0.048) & (0.134) & (0.063) \\
    Network & 0.032 & 0.248 & 0.202 \\
          & (0.069) & (0.069) & (0.073) \\
    Entry year & -0.623 & -0.457 & -0.148 \\
          & (0.032) & (0.041) & (0.032) \\
    Exit year & -1.707 & -1.720 & -0.565 \\
          & (0.034) & (0.041) & (0.107) \\
    $\ln s_{j|g}$ &       &       & 0.677 \\
          &       &       & (0.063) \bigstrut[b]\\
    \hline
    Own elasticity &       &       &  \bigstrut[t]\\
    - Median & -0.371 & -5.490 & -5.755 \\
    - Weighted average & -0.291 & -4.302 & -4.370 \bigstrut[b]\\
    \hline
    \hline
    \end{tabular}%
  \label{tab: results of restricted models}%
   \tablenotes This table presents OLS  and IV  estimates of the linear demand model given by $\ln(s_{jm}/s_{0m}) = x_{jm}\beta - \alpha \frac{p_{jm}}{\bar{y}_{m}} + \xi_{jm}$, as well as estimates of the nested logit demand model given by $\ln (s_{jm}/s_{0m}) = x_{jm}\beta - \alpha \frac{p_{jm}}{\bar{y}_{m}} + \rho\ln s_{j\mid g} + \xi_{jm} $. The utility specifications also include fixed effects for nameplate, province, and fuel type interacted with year. Two cost shifters---import tariff rate and raw material price interacted with vehicle weight---are used as excluded instruments for price in the IV logit estimation, while the within-group product count is used as an additional instrument in the nested logit estimation. Robust standard errors, clustered by market, are reported in parentheses.   
\end{table}
\clearpage

\newpage
\begin{table}[htbp]
  \centering
  \caption{First-stage regression results}
    \begin{tabular}{lcccc}
    \hline
    \hline
          & IV Logit &       & \multicolumn{2}{c}{Nested Logit}  \bigstrut\\
\cline{4-5}    Variable &       &       & $p_{j}/\bar{y}$
 & $\ln s_{j|g}$ \bigstrut\\
    \hline
    Raw material prices $\times$ Weight & 0.135 &       & 0.130 & 1.364 \bigstrut[t]\\
          & (0.127) &       & (0.128) & (0.252) \\
    Tariff rates & 0.037 &       & 0.037 & -0.224 \\
          & (0.005) &       & (0.005) & (0.013) \\
    Group product count &       &       & -0.004 & -0.006 \\
          &       &       & (0.001) & (0.002) \\
    Cost per km & 0.308 &       & 0.313 & -2.811 \\
          & (0.026) &       & (0.026) & (0.094) \\
    Acceleration & 3.365 &       & 3.364 & -1.763 \\
          & (0.131) &       & (0.131) & (0.394) \\
    Size  & 0.115 &       & 0.114 & 0.144 \\
          & (0.011) &       & (0.011) & (0.049) \\
    Network & 0.019 &       & 0.011 & -0.109 \\
          & (0.011) &       & (0.011) & (0.046) \\
    Entry year & 0.033 &       & 0.034 & -0.654 \\
          & (0.004) &       & (0.004) & (0.033) \\
    Exit year & -0.020 &       & -0.020 & -1.602 \\
          & (0.004) &       & (0.004) & (0.031) \bigstrut[b]\\
    \hline
    R-squared & 0.955 &       & 0.955 & 0.781 \bigstrut[t]\\
    SW F statistic & 26.29 &       & 16.40 & 22.01 \\
    Observations & 34\,629 &       & 34\,629 & 34\,629 \bigstrut[b]\\
    \hline
    \hline
    \end{tabular}%
  \label{tab: first stage results}%
  \tablenotes This table presents the first-stage regression results from the standard logit and nested logit demand models. The model also includes fixed effects for nameplate, province, and fuel type interacted with year. Robust standard errors, clustered by market, are reported in parentheses. 
\end{table}
\clearpage

\newpage
\begin{table}[htbp]
  \centering
  \caption{Own- and cross-price elasticities by fuel type}
    \begin{tabular}{lcccccc}
    \hline
    \hline
          &       &       &       &       &       &  \bigstrut[t]\\
    \multicolumn{7}{c}{\underline{\textit{Panel A: Own-price elasticities}}} \bigstrut[b] \\
    Fuel Type & Mean  & Min   & P10   & P50   & P90   & Max \bigstrut\\
    \hline
    BEV    & -6.1823 & -16.9671 & -9.0477 & -6.9034 & -4.4898 & -0.4315 \bigstrut[t]\\
    HEV   & -5.9411 & -10.3288 & -7.3308 & -5.9954 & -4.8752 & -3.2270 \\
    Diesel & -6.3986 & -10.1832 & -7.7114 & -6.2005 & -4.9920 & -1.9965 \\
    Gasoline & -5.6100 & -16.1493 & -7.4639 & -6.0916 & -4.8829 & -2.0436 \\
    LPG   & -5.9390 & -9.9424 & -7.2039 & -5.7529 & -4.5009 & -3.1466 \\
          &       &       &       &       &  \\
    \multicolumn{7}{c}{\underline{\textit{Panel B: Cross-price elasticities}}} \bigstrut[b] \\
          & BEV    & HEV   & Diesel & Gasoline & LPG   &  \bigstrut\\
\cline{1-6}    BEV    & 0.0374 & 0.0030 & 0.0012 & 0.0012 & 0.0018 &  \bigstrut[t]\\
    HEV   & 0.0012 & 0.0273 & 0.0034 & 0.0026 & 0.0063 &  \\
    Diesel & 0.0008 & 0.0041 & 0.0086 & 0.0019 & 0.0040 &  \\
    Gasoline & 0.0006 & 0.0030 & 0.0021 & 0.0055 & 0.0040 &  \\
    LPG   & 0.0010 & 0.0064 & 0.0032 & 0.0032 & 0.0243 &  \\
          &       &       &       &       &  \bigstrut[b]\\
    \hline
    \hline
    \end{tabular}%
  \label{tab: elasticities by group}%
  \tablenotes The top panel reports summary statistics (sales-weighted average, minimum, 10th percentile, median, 90th percentile, and maximum) for the own-price elasticities for products within each fuel type. The bottom panel reports the unweighted average cross-price elasticity of demand for products with the fuel type indicated in the row with respect to the prices of products with the fuel type indicated in the column.
\end{table}%
\clearpage

\newpage
\begin{table}[htbp]
\small
  \centering
  \caption{Equilibrium prices of BEVs and HEVs in 2023}
    \begin{tabular}{lrrrrlrrr}
    \hline
    \hline
    Brand and & \multicolumn{1}{c}{No} & \multicolumn{1}{c}{BEV} & \multicolumn{1}{c}{Pass-} &       & Brand and & \multicolumn{1}{c}{No} & \multicolumn{1}{c}{HEV} & \multicolumn{1}{c}{Pass-} \bigstrut[t]\\
    Nameplate & \multicolumn{1}{c}{sub.} & \multicolumn{1}{c}{sub.} & \multicolumn{1}{c}{through} &       & Nameplate & \multicolumn{1}{c}{sub.} & \multicolumn{1}{c}{sub.} & \multicolumn{1}{c}{through}     \bigstrut[b]\\
    \hline
          &       &       &       &       &       &       &       &  \bigstrut[t]\\
    \multicolumn{4}{l}{\textit{Battery Electric Vehicle}} &       & \multicolumn{4}{l}{\textit{Hybrid electric vehicle}} \bigstrut[b] \\
    Audi Q4 e-tron & 60.12 & 60.32 & 0.94  &       & Hyundai Grandeur & 43.46 & 43.28 & 1.11 \\
    Audi e-tron & 82.03 & 82.35 & 0.90  &       & Hyundai Santa Fe & 38.43 & 38.26 & 1.10 \\
    BMW i4 & 65.53 & 65.34 & 1.04  &       & Hyundai Sonata & 32.42 & 31.98 & 1.27 \\
    BMW i5$^{+}$ & 81.42 & 82.27 &   -    &       & Hyundai Elantra & 25.99 & 25.37 & 1.38 \\
    BMW i7$^{+}$ & 159.62 & 160.02 &  -     &       & Hyundai Kona & 30.02 & 29.83 & 1.12 \\
    BMW iX$^{+}$ & 110.59 & 111.30 &   -    &       & Hyundai Tucson & 31.42 & 31.18 & 1.15 \\
    BMW iX1 & 58.31 & 58.24 & 1.02  &       & Kia K5 & 31.40 & 30.92 & 1.30 \\
    BMW iX3 & 63.92 & 63.87 & 1.01  &       & Kia K8 & 39.87 & 39.60 & 1.17 \\
    GM Korea Bolt EUV & 40.62 & 40.55 & 1.01  &       & Kia Niro & 29.61 & 29.42 & 1.12 \\
    GM Korea Bolt EV & 38.57 & 37.85 & 1.07  &       & Kia Sportage & 34.15 & 33.94 & 1.13 \\
    Genesis G80 & 80.04 & 81.46 & 0.71  &       & Kia Sorento & 39.35 & 39.17 & 1.11 \\
    Genesis GV60 & 62.48 & 63.61 & 0.76  &       & Kia Carnival & 45.62 & 45.48 & 1.09 \\
    Genesis GV70 & 69.84 & 70.99 & 0.75  &       & Lexus ES & 60.46 & 60.27 & 1.12 \\
    Hyundai Ioniq 5 & 53.25 & 52.77 & 1.04  &       & Lexus LS & 137.83 & 137.72 & 1.07 \\
    Hyundai Ioniq 5 N & 69.29 & 70.21 & 0.77  &       & Lexus NX & 62.65 & 62.46 & 1.12 \\
    Hyundai Ioniq 6 & 54.60 & 55.02 & 0.96  &       & Lexus RX & 88.86 & 88.74 & 1.07 \\
    Hyundai Kona & 44.67 & 44.25 & 1.04  &       & Lexus UX & 45.81 & 45.58 & 1.15 \\
    KG Mobility Torres & 45.76 & 45.05 & 1.07  &       & Renault Korea XM3 & 31.34 & 31.11 & 1.14 \\
    Kia EV6 & 55.09 & 54.48 & 1.06  &       & Toyota RAV4 & 42.39 & 42.16 & 1.14 \\
    Kia EV9 & 64.20 & 65.14 & 0.80  &       & Toyota Sienna & 61.28 & 61.10 & 1.11 \\
    Kia Niro & 45.96 & 45.36 & 1.05  &       & Toyota Alphard & 88.52 & 88.41 & 1.07 \\
    Kia Ray & 27.08 & 27.36 & 0.96  &       & Toyota Camry & 38.02 & 37.77 & 1.16 \\
    Mercedes EQA & 57.30 & 57.22 & 1.02  &       & Toyota Crown & 51.74 & 51.54 & 1.12 \\
    Mercedes EQB & 64.19 & 64.19 & 1.00  &       & Toyota Prius & 37.92 & 37.65 & 1.16 \\
    Mercedes EQE$^{+}$ & 84.90 & 85.53 &  -     &       & Toyota Highlander & 62.04 & 61.86 & 1.11 \\
    Mercedes EQE SUV$^{+}$ & 98.97 & 99.89 &   -    &       &       &       &       &  \\
    Mercedes EQS$^{+}$ & 134.77 & 135.25 &  -     &       &       &       &       &  \\
    Mercedes EQS SUV$^{+}$ & 137.22 & 137.62 &   -    &       &       &       &       &  \\
    Tesla Model 3 & 58.89 & 58.64 & 1.07  &       &       &       &       &  \\
    Tesla Model S$^{+}$ & 119.59 & 120.39 &   -    &       &       &       &       &  \\
    Tesla Model X$^{+}$ & 129.10 & 129.42 &   -    &       &       &       &       &  \\
    Tesla Model Y & 54.64 & 54.07 & 1.15  &       &       &       &       &  \\
    Volkswagen ID.4 & 54.31 & 53.51 & 1.09  &       &       &       &       &  \\
          &       &       &       &       &       &       &       &  \bigstrut[b]\\
    \hline
    \hline      
    \end{tabular}%
  \label{tab: counterfactual prices}%
  \tablenotes This table compares the equilibrium prices (in million KRW) of 33 BEVs under the no‐subsidy and current BEV subsidy scenarios, as well as those of 25 HEVs under the no-subsidy and HEV subsidy scenarios, for the year 2023. $^{+}$ denotes models that are not subsidy-eligible. 
  \end{table}
\clearpage

\newpage
\begin{table}[htbp]
  \centering
  \caption{Emission savings of BEV and HEV subsidy policies across different emission factors}
    \begin{tabular}{llrrr}
    \toprule
    Country & Code  & GHG factor & BEV subsidy & HEV subsidy \bigstrut[b]\\
    \midrule
    Norway & NOR   & 31    & 2\,805  & 2\,040  \bigstrut[t]\\
    France & FRA   & 85    & 2\,629  & 2\,048  \\
    Brazil & BRA   & 122   & 2\,509  & 2\,054  \\
    Denmark & DNK   & 213   & 2\,213  & 2\,068  \\    
    Portugal & PRT   & 251   & 2\,089  & 2\,073  \\
    United Kingdom & GBR   & 274   & 2\,015  & 2\,077  \\
    Peru  & PER   & 324   & 1\,852  & 2\,084  \\
    Netherlands & NLD   & 339   & 1\,803  & 2\,087  \\
    United States & USA   & 430   & 1\,507  & 2\,101  \\
    Germany & GER   & 440   & 1\,475  & 2\,102  \\
    \textbf{South Korea} & \textbf{KOR} & \textbf{453} & \textbf{1\,433} & \textbf{2\,104} \\
    World Average & WORLD & 526   & 1\,195  & 2\,115  \\
    Japan & JPN   & 542   & 1\,143  & 2\,118  \\
    China & CHN   & 607   & 932   & 2\,127  \\
    Malaysia & MYS   & 649   & 795   & 2\,134  \\
    Philippines & PHL   & 667   & 737   & 2\,137  \\
    Indonesia & IDN   & 725   & 548   & 2\,145  \\
    Poland & POL   & 778   & 376   & 2\,153  \\
    India & IND   & 826   & 220   & 2\,161  \\
    Kazakhstan & KAZ   & 907   & -43   & 2\,173  \bigstrut[b]\\
    \bottomrule
    \end{tabular}%
  \label{tab: policy counterfactual 3}%
\tablenotes This table compares the emission savings (in thousand tonnes of CO$_2$e) under BEV and HEV subsidy policies across 20 distinct GHG emission factor levels, representing 19 countries and the world average. These nations are selected to ensure a representative geographic distribution that spans the entire spectrum of global GHG emission factors.
\end{table}%
\clearpage

\newpage
\begin{figure}[htbp]
    \centering
	\caption{Trends in BEV and HEV sales shares by country}
	\begin{subfigure}[b]{0.3\textwidth}	
    	\includegraphics[width=\textwidth]{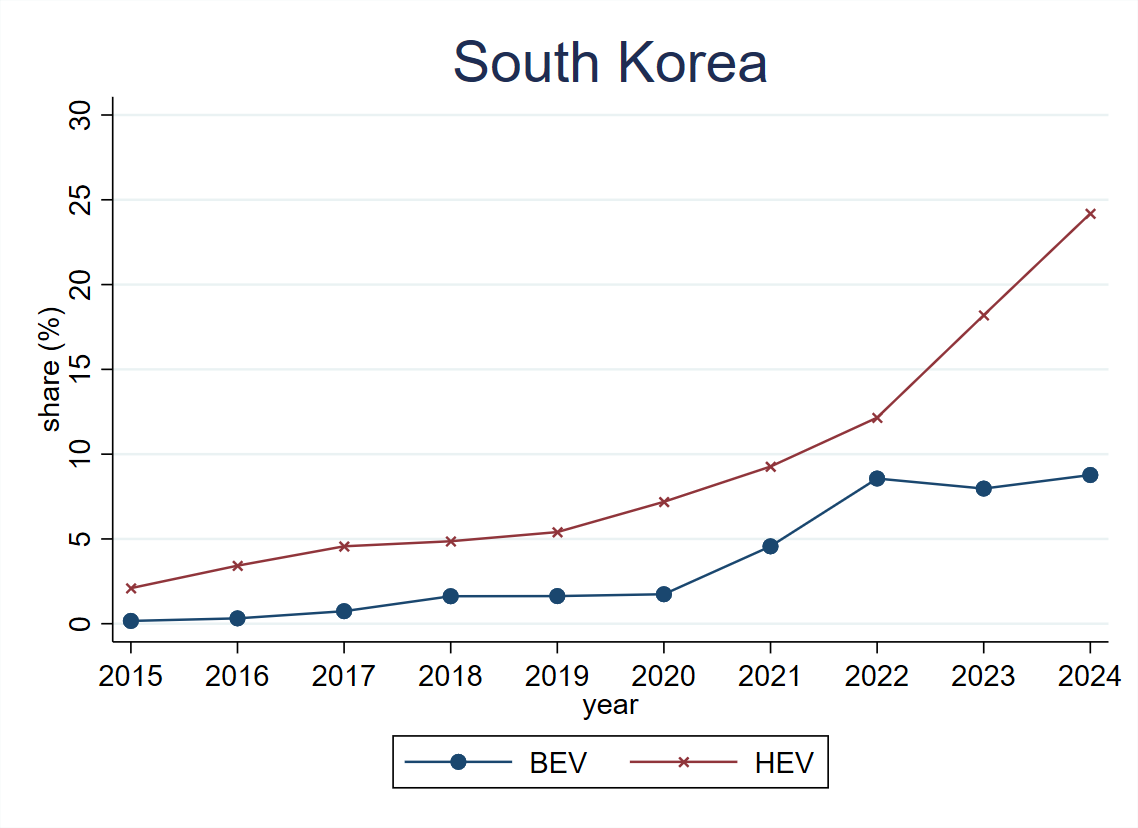}
    	\end{subfigure}
	\begin{subfigure}[b]{0.3\textwidth}	
    	\includegraphics[width=\textwidth]{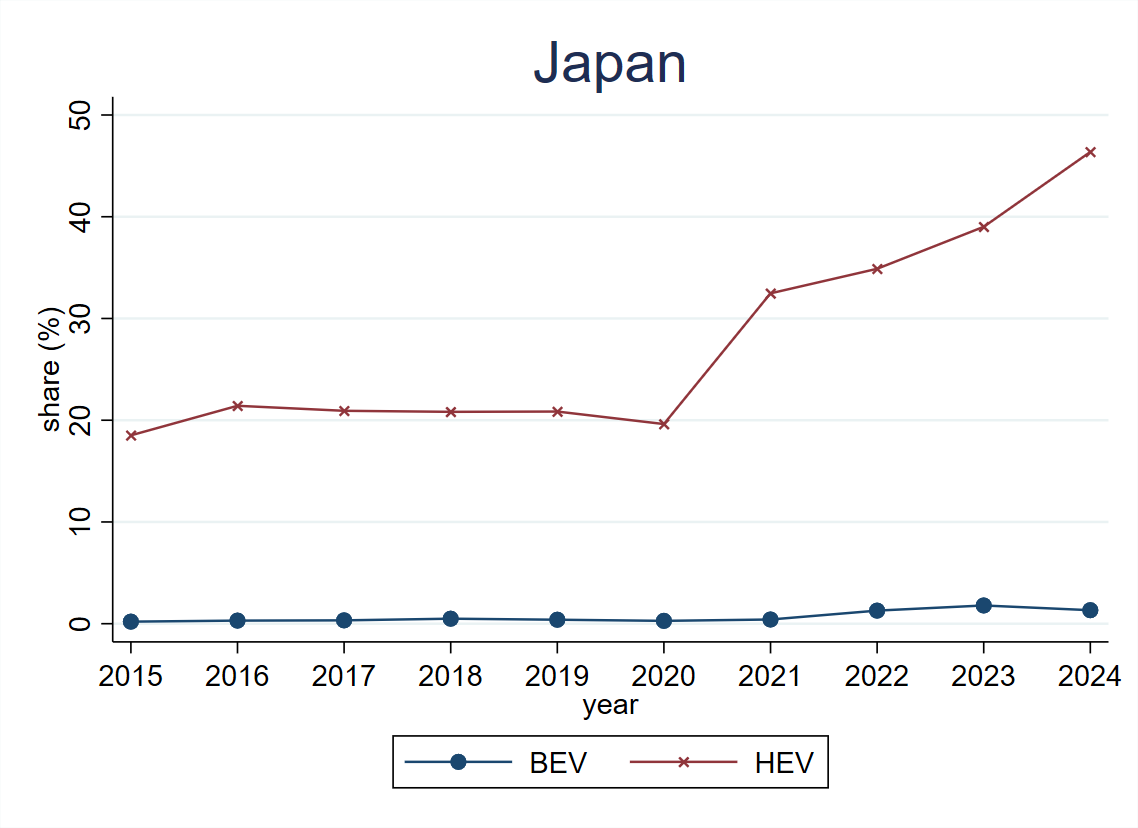}
    	\end{subfigure}
    \begin{subfigure}[b]{0.3\textwidth}	
    	\includegraphics[width=\textwidth]{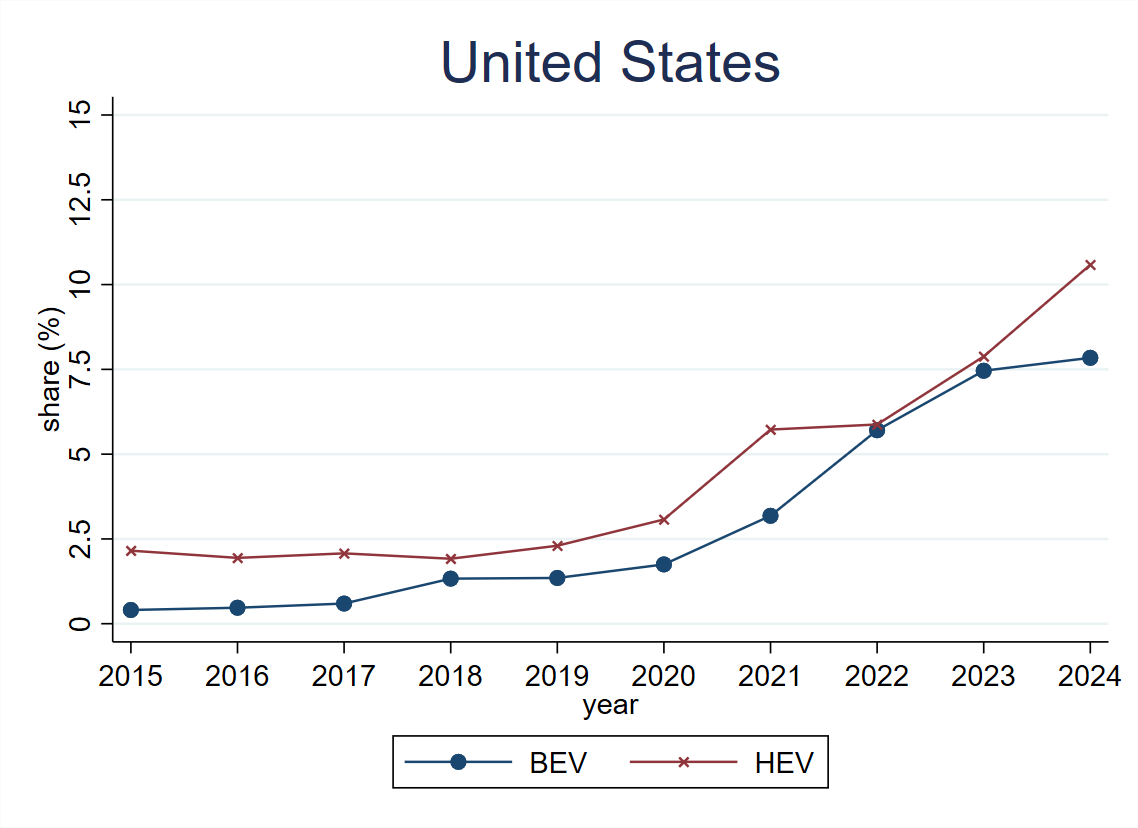}
    	\end{subfigure}

    \begin{subfigure}[b]{0.3\textwidth}	
    	\includegraphics[width=\textwidth]{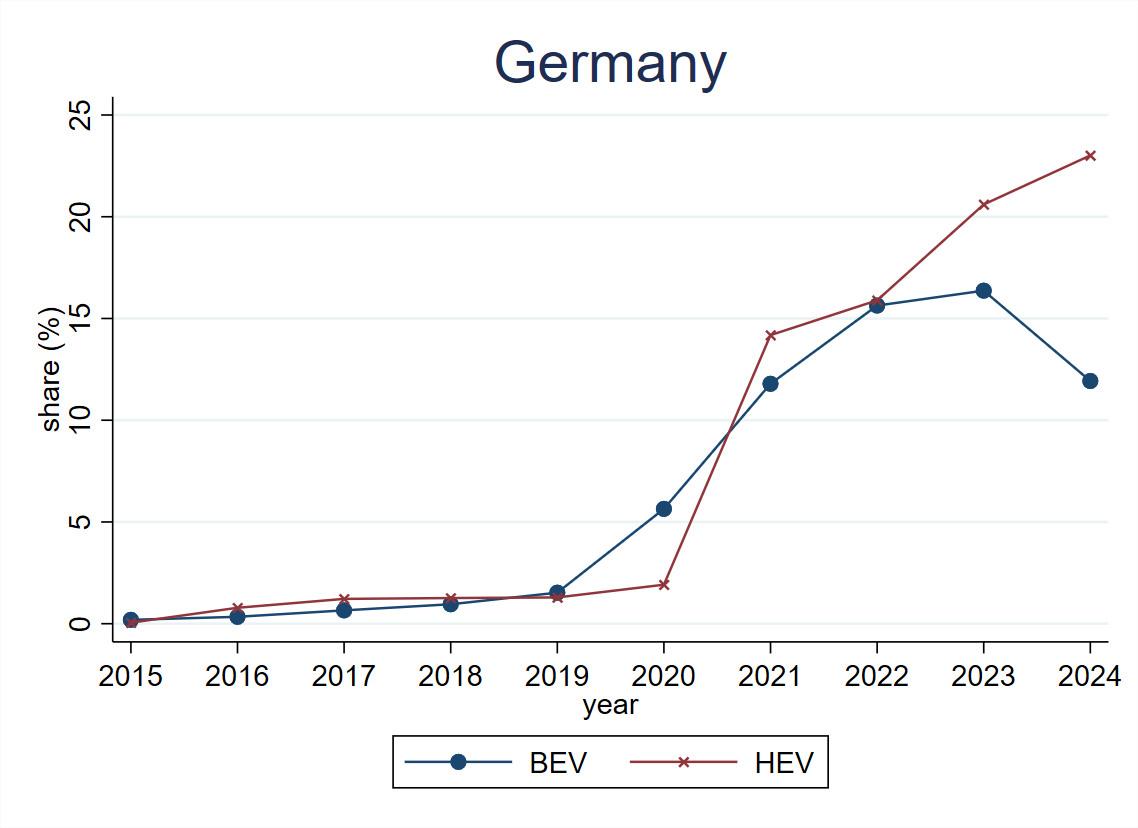}
    	\end{subfigure}
	\begin{subfigure}[b]{0.3\textwidth}	
    	\includegraphics[width=\textwidth]{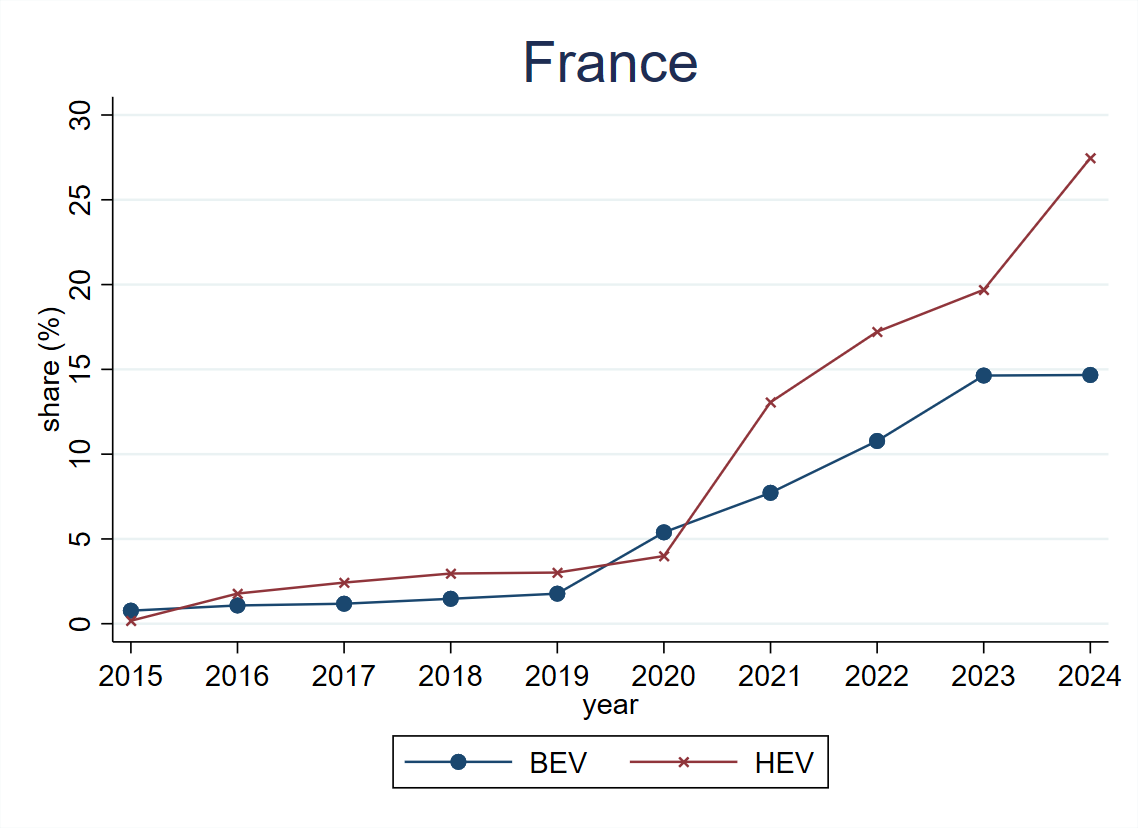}
    	\end{subfigure}
    \begin{subfigure}[b]{0.3\textwidth}	
    	\includegraphics[width=\textwidth]{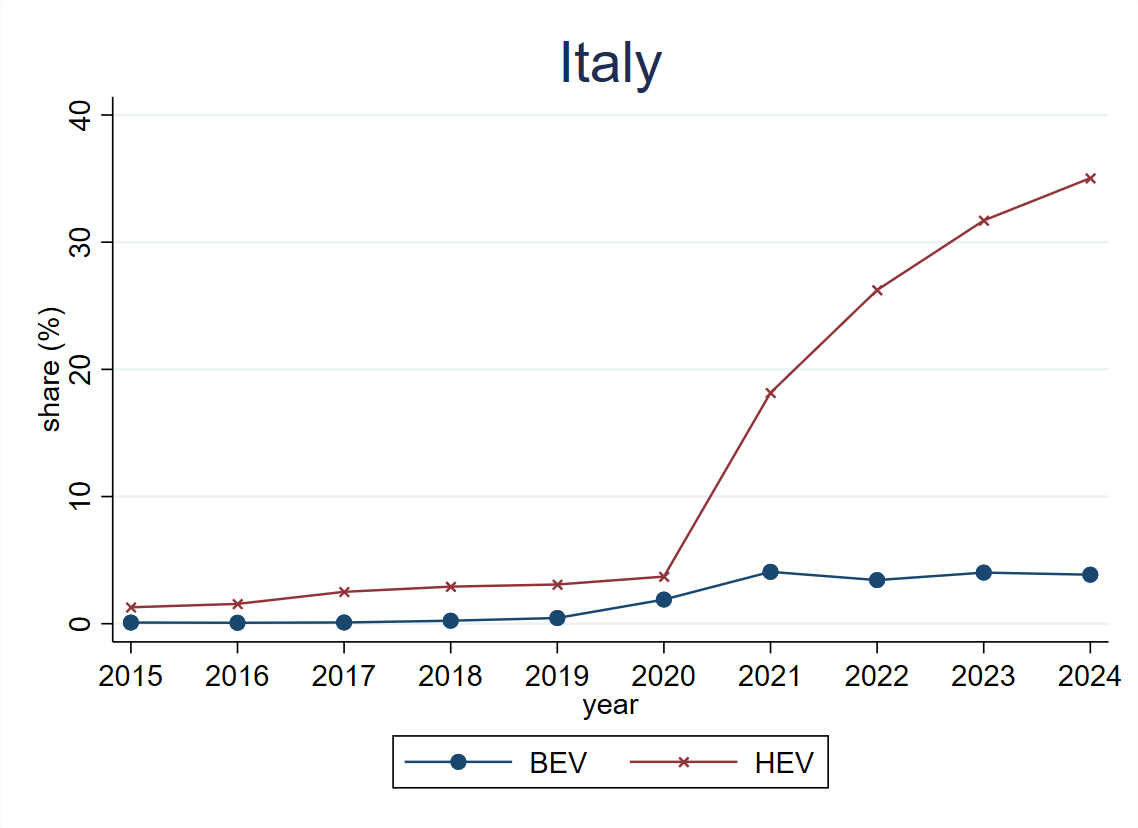}
    	\end{subfigure}

    \begin{subfigure}[b]{0.3\textwidth}	
    	\includegraphics[width=\textwidth]{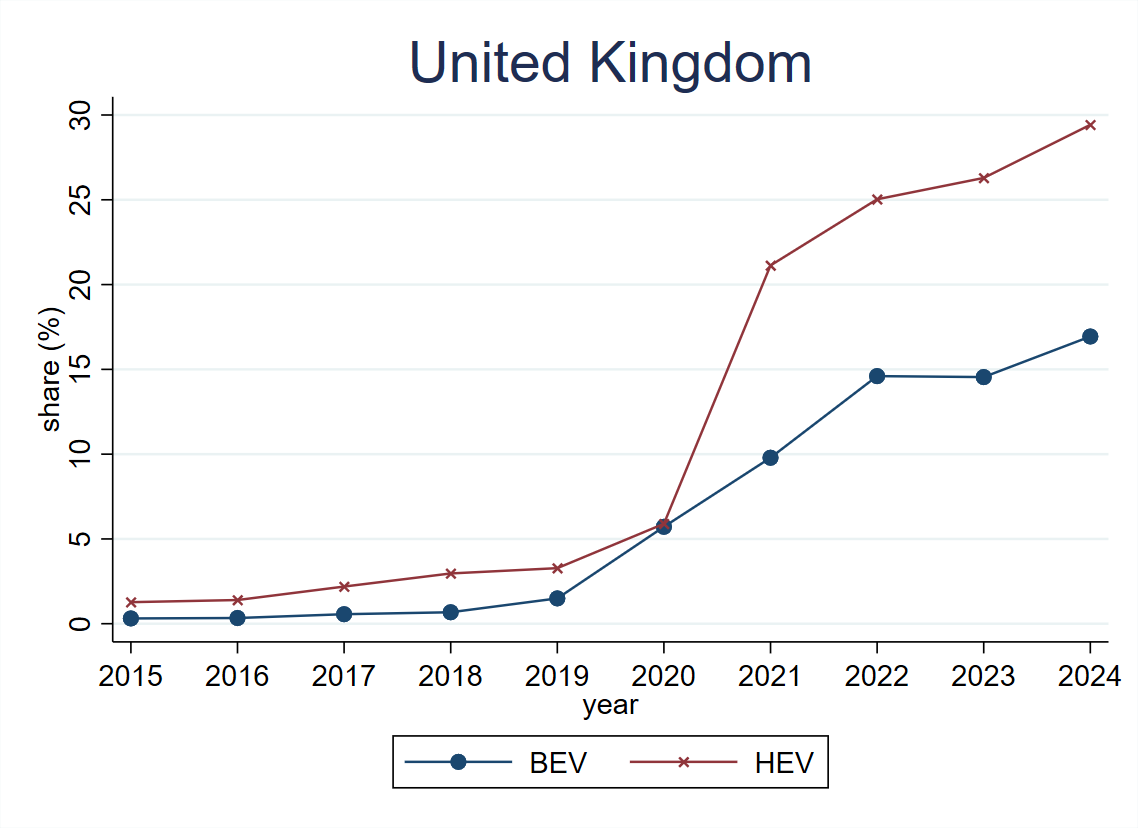}
    	\end{subfigure}
 	\begin{subfigure}[b]{0.3\textwidth}	
    	\includegraphics[width=\textwidth]{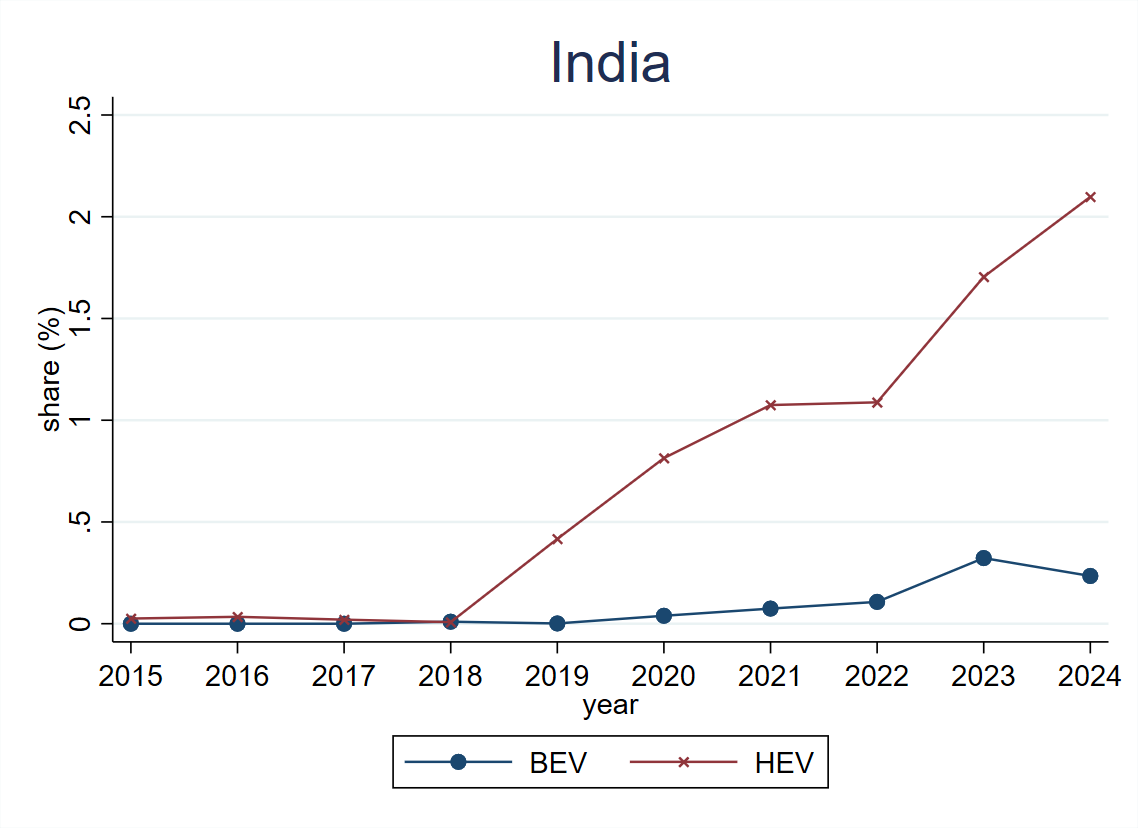}
    	\end{subfigure}
    \begin{subfigure}[b]{0.3\textwidth}	
    	\includegraphics[width=\textwidth]{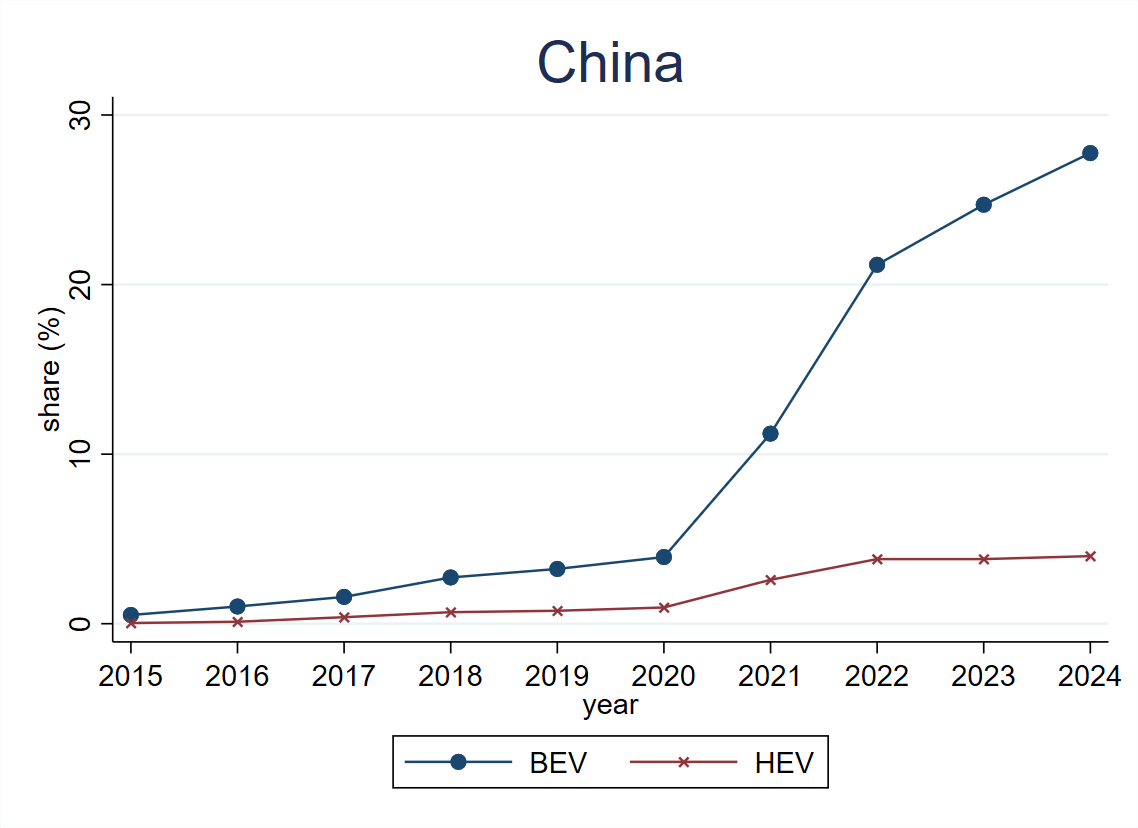}
    	\end{subfigure}
\label{fig: BEV and HEV share}
 \tablenotes This figure shows trends in BEV and HEV shares of new light-vehicle sales by country from 2015 to 2024. PHEV sales are not included in either the BEV or HEV categories. \textit{Source: Korean Automobile \& Mobility Association (KAMA)}.
\end{figure}
\clearpage

\newpage
\begin{figure}[htbp]
	\centering
	\caption{17 provinces in South Korea}
	\includegraphics[width=\textwidth]{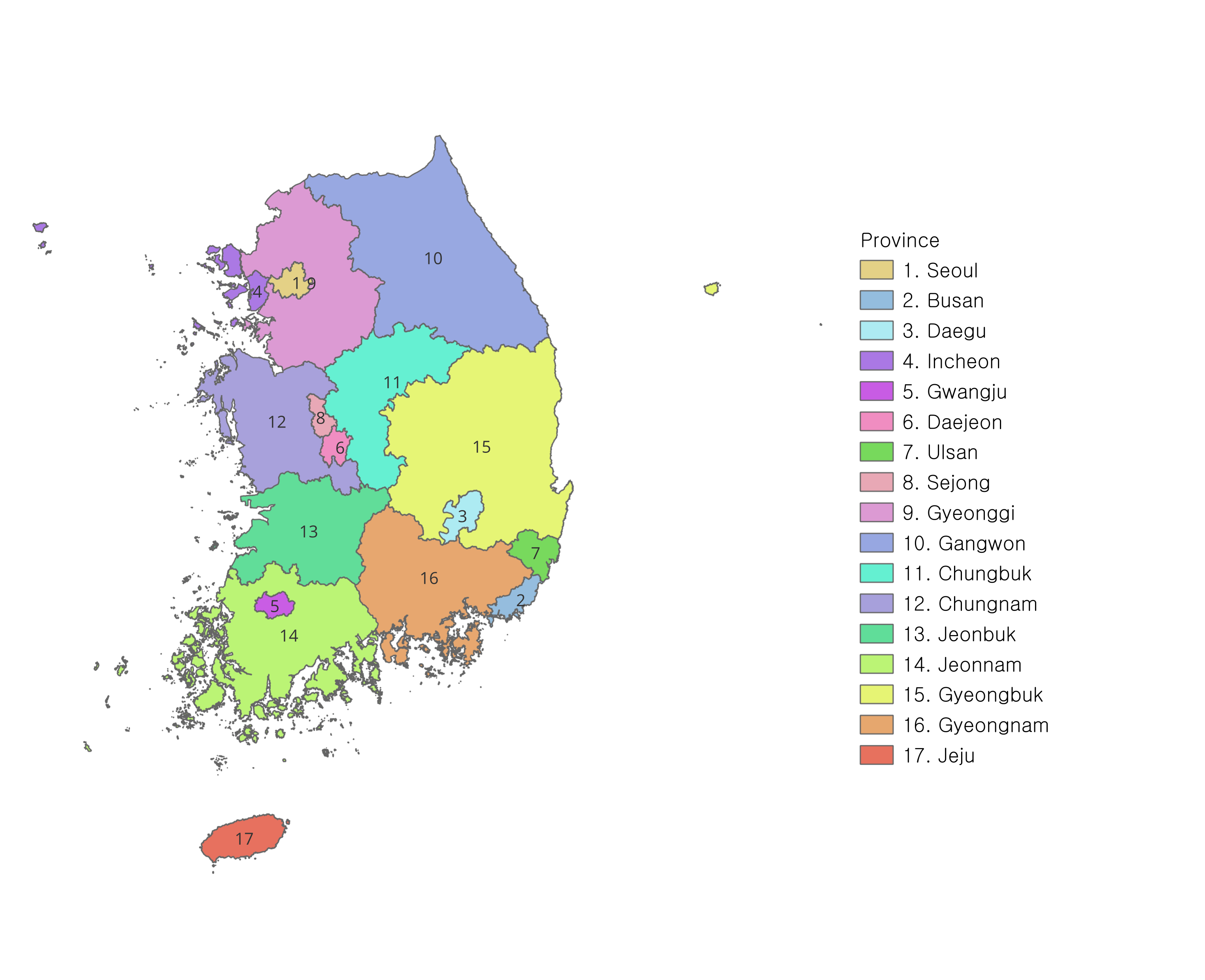}
\label{fig: province illustration}
\tablenotes This figure illustrates how provinces (first-tier administrative divisions) are geographically defined in South Korea.
\end{figure}
\clearpage

\newpage 
\begin{figure}[htbp]
	\centering
	\caption{Sales share trend of each fuel type}
	\includegraphics[width=\textwidth]{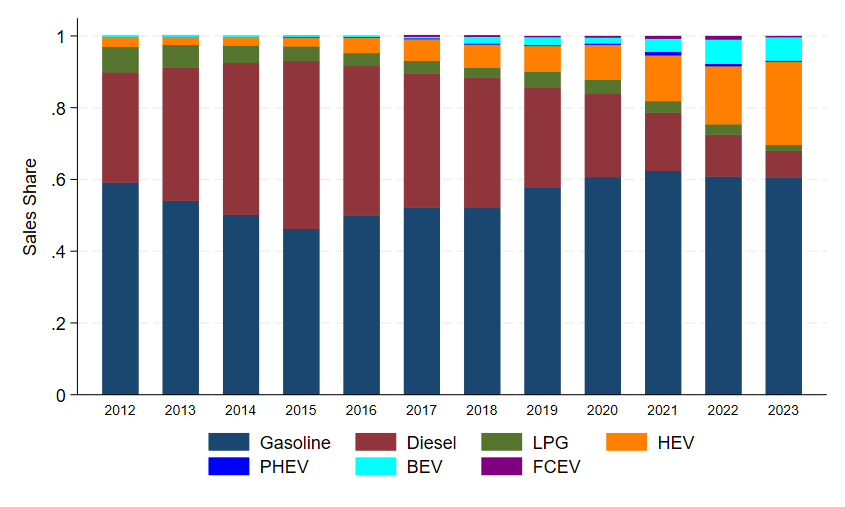}
\label{fig: fuel share trend}
\tablenotes This figure presents sales share of each fuel type during the sample period, 2012--2023.
\end{figure}
\clearpage

\newpage
\begin{figure}[htbp]
	\centering
	\caption{EV subsidy dispersion by year (2016 - 2023)}
	\begin{subfigure}[b]{0.4\textwidth}	
    	\caption{2016}
    	\includegraphics[width=\textwidth]{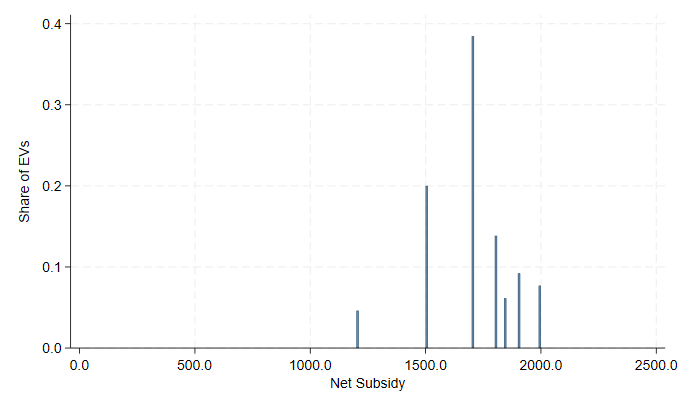}
    	\end{subfigure}
    \qquad
	\begin{subfigure}[b]{0.4\textwidth}	
    	\caption{2017}
    	\includegraphics[width=\textwidth]{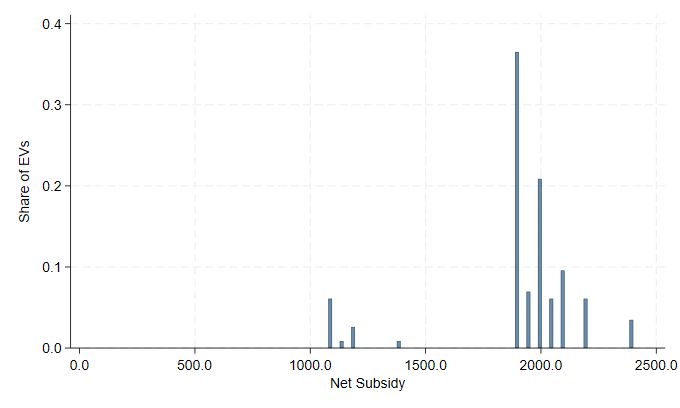}
    	\end{subfigure}
    	\begin{subfigure}[b]{0.4\textwidth}
    \vspace{0.2in}
    	\caption{2018}
    	\includegraphics[width=\textwidth]{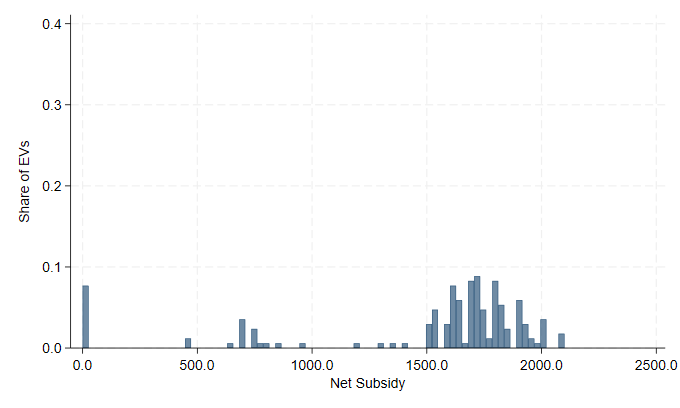}
    	\end{subfigure}
        \qquad
    	\begin{subfigure}[b]{0.4\textwidth}
    	\caption{2019}
    	\includegraphics[width=\textwidth]{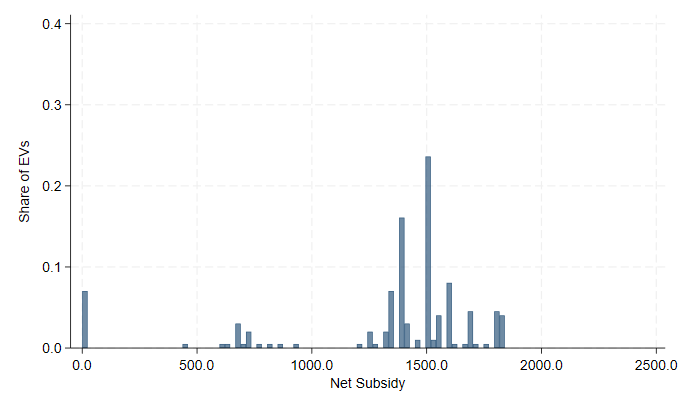}
    	\end{subfigure}
    	\begin{subfigure}[b]{0.4\textwidth}
    \vspace{0.2in}
    	\caption{2020}
    	\includegraphics[width=\textwidth]{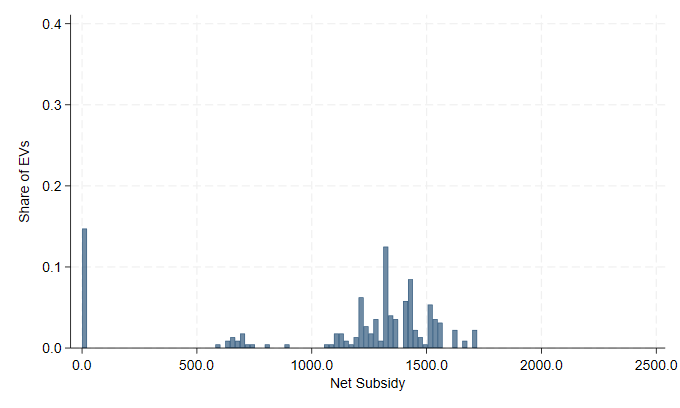}
    	\end{subfigure}
        \qquad
    	\begin{subfigure}[b]{0.4\textwidth}
    	\caption{2021}
    	\includegraphics[width=\textwidth]{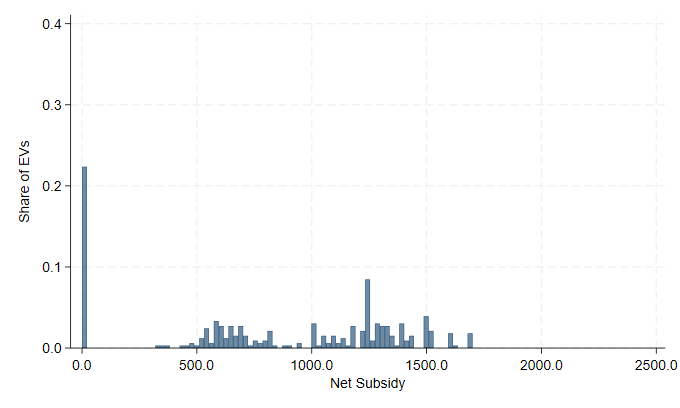}
    	\end{subfigure}
    	\begin{subfigure}[b]{0.4\textwidth}
    \vspace{0.2in}
    	\caption{2022}
    	\includegraphics[width=\textwidth]{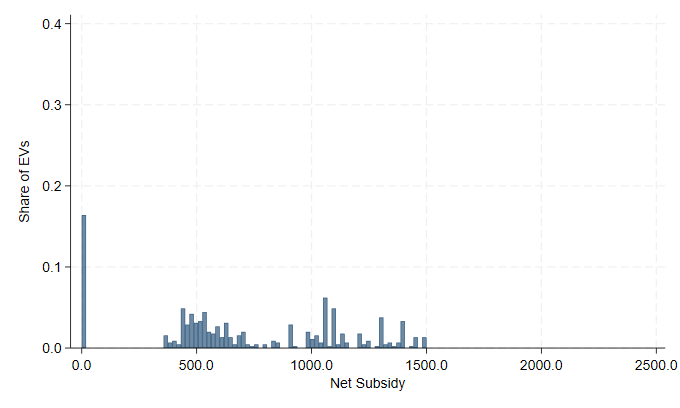}
    	\end{subfigure}
        \qquad
    	\begin{subfigure}[b]{0.4\textwidth}
    	\caption{2023}
    	\includegraphics[width=\textwidth]{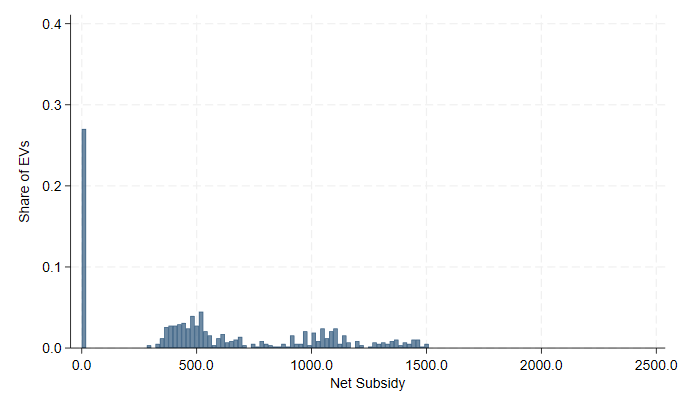}
    	\end{subfigure}
\label{fig: subsidy dispersion}
\tablenotes This figure presents a histogram of the nameplate/province-level total (national + local) BEV subsidy for each year (2016-2023).
\end{figure}
\clearpage

\newpage
\begin{figure}[htbp]
	\centering
	\caption{Nameplate-level local subsidies of selected provinces}
    \includegraphics[width=0.8\textwidth]{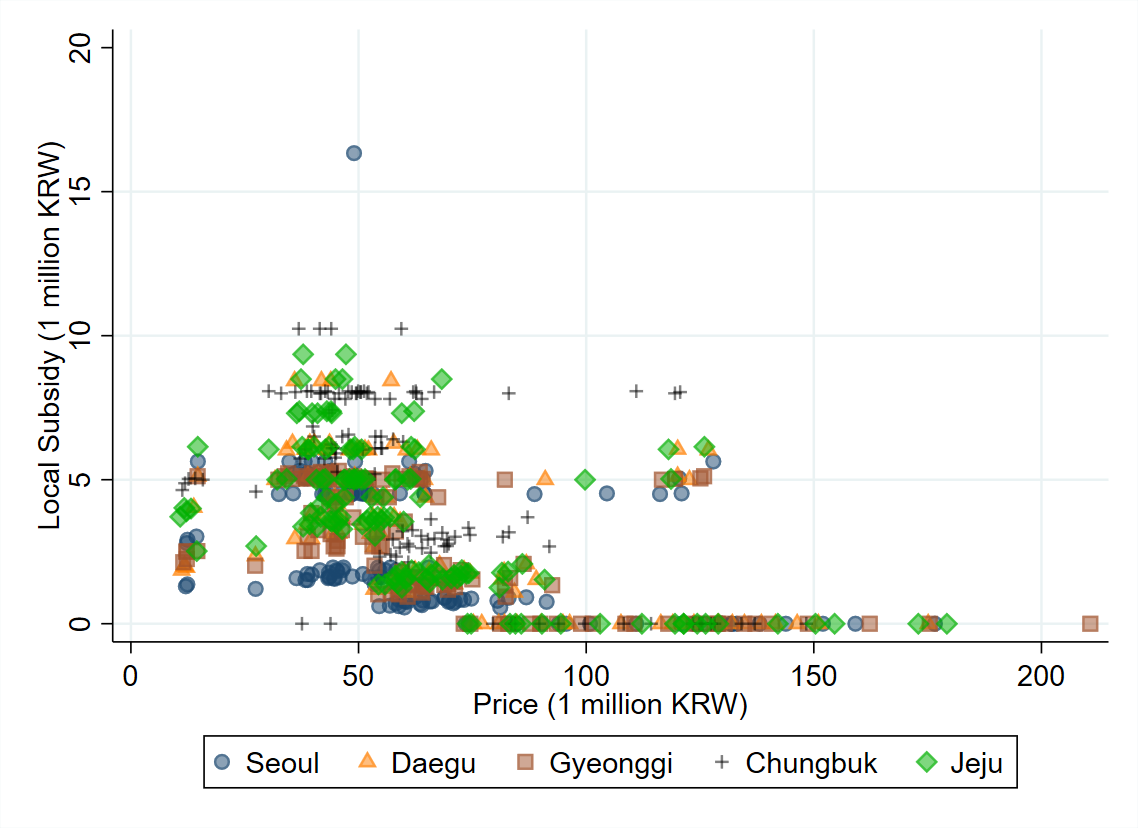}
\label{fig: scatter localsub}
\tablenotes This figure presents a scatter plot of the nameplate-level local BEV subsidy for the following selected provinces (Seoul, Daegu, Gyeonggi, Chungbuk, and Jeju) over the entire sample period (2012-2023). Each point represents a province-nameplate-level observation.
\end{figure}
\clearpage

\newpage
\begin{figure}[htbp]
	\centering
	\caption{Diversion by fuel types}
	\begin{subfigure}[b]{0.48\textwidth}	
    	\caption{From BEV}
    	\includegraphics[width=\textwidth]{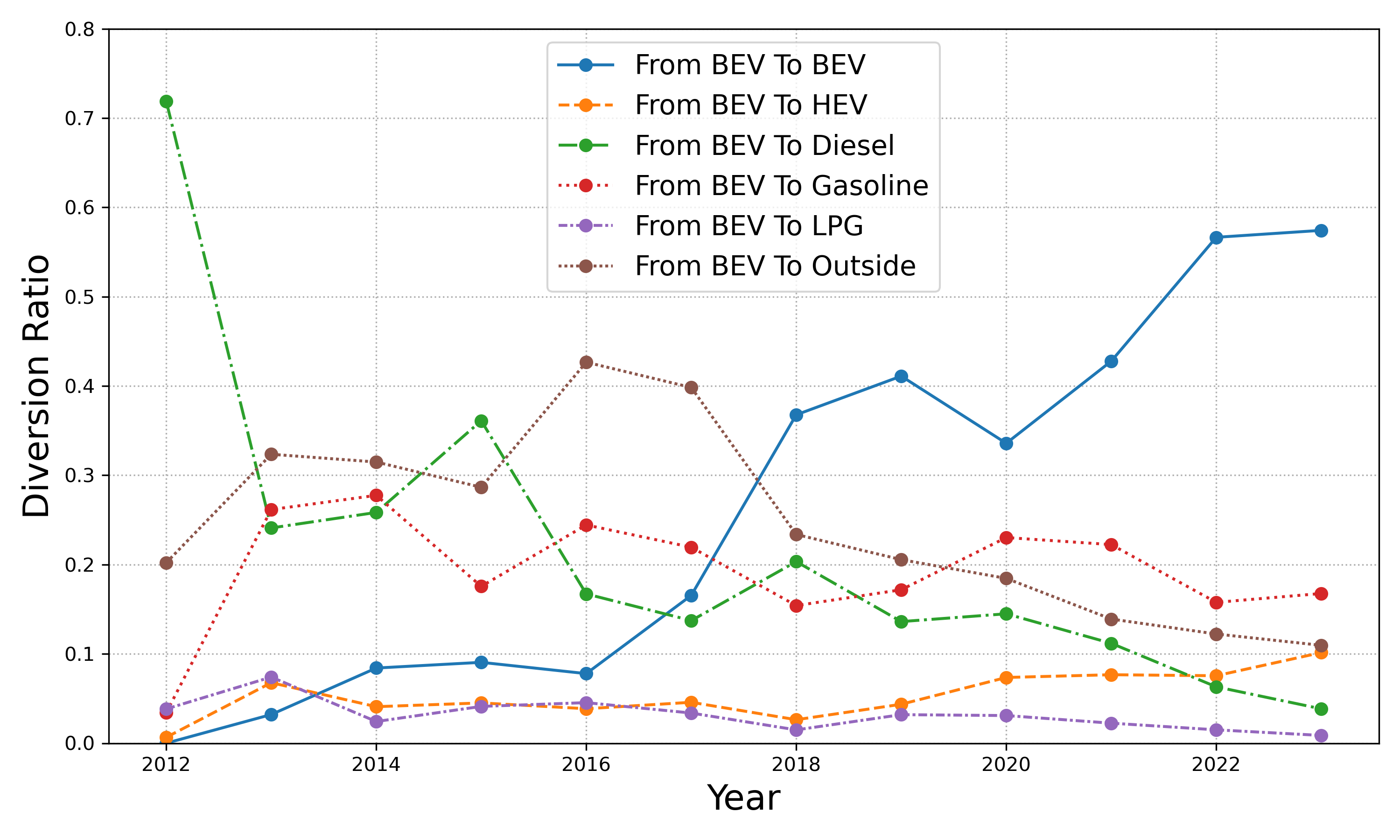}
    	\end{subfigure}
	\begin{subfigure}[b]{0.48\textwidth}	
    	\caption{From HEV}
    	\includegraphics[width=\textwidth]{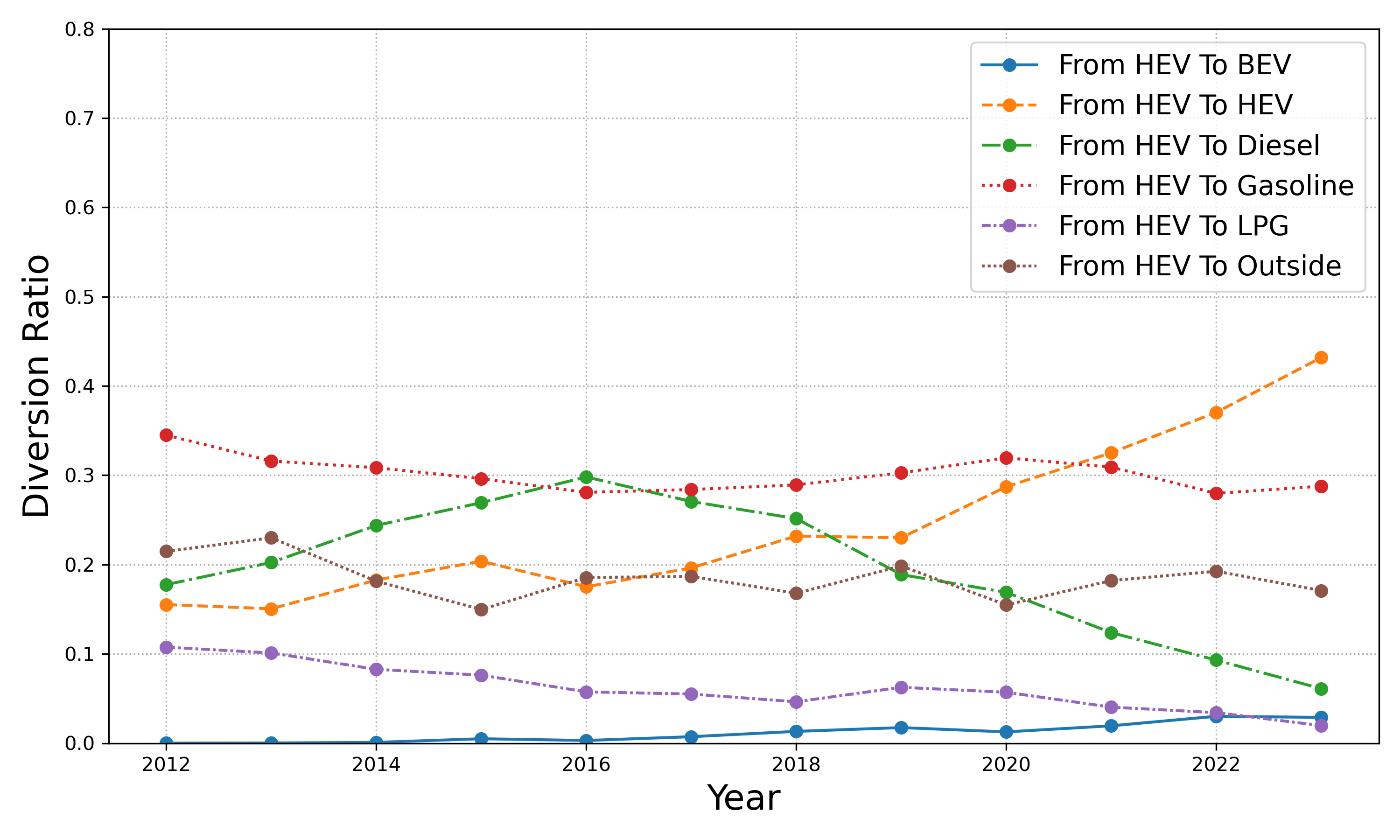}
    	\end{subfigure}
    	\begin{subfigure}[b]{0.48\textwidth}
    \vspace{0.2in}
    	\caption{From Diesel}
    	\includegraphics[width=\textwidth]{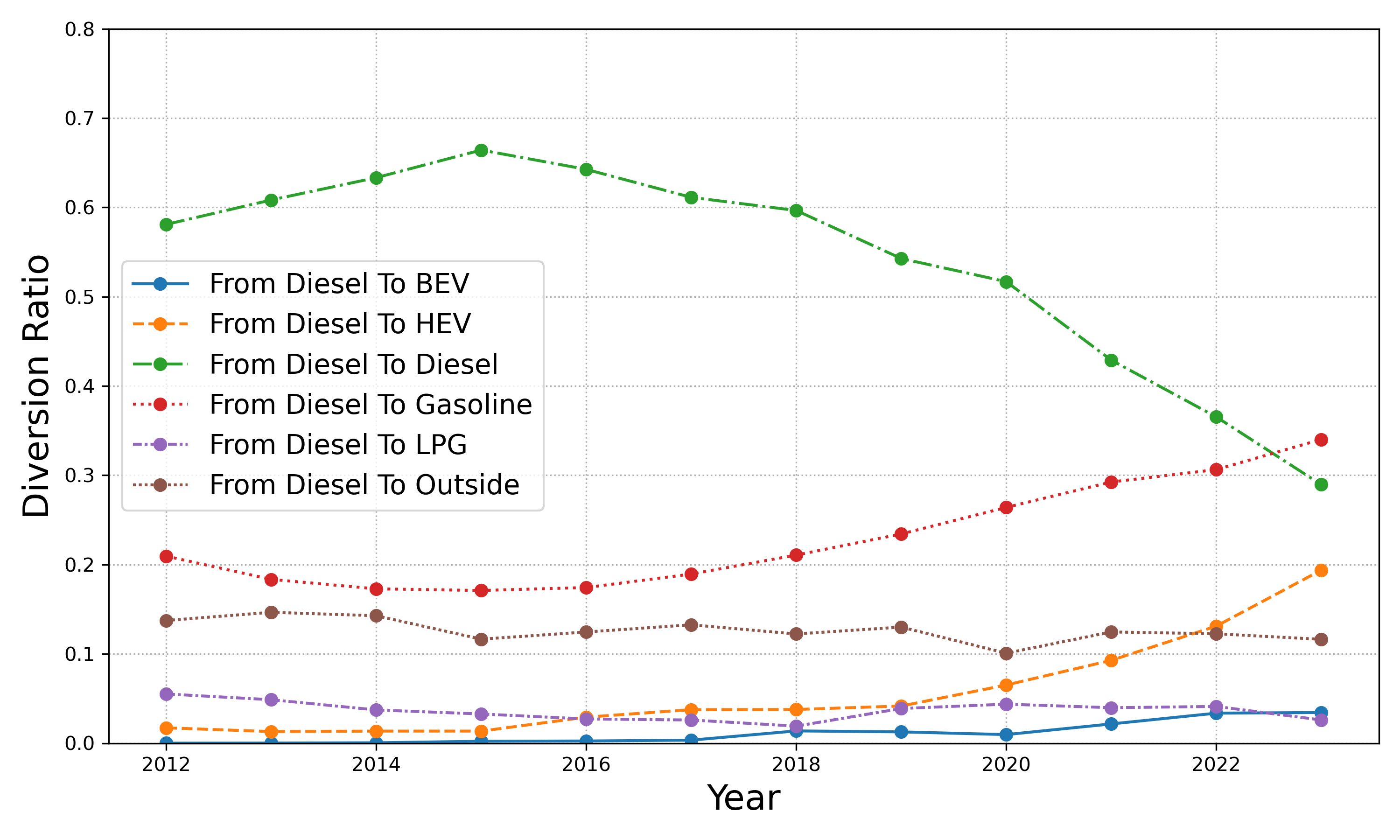}
    	\end{subfigure}
    	\begin{subfigure}[b]{0.48\textwidth}
    	\caption{From Gasoline}
    	\includegraphics[width=\textwidth]{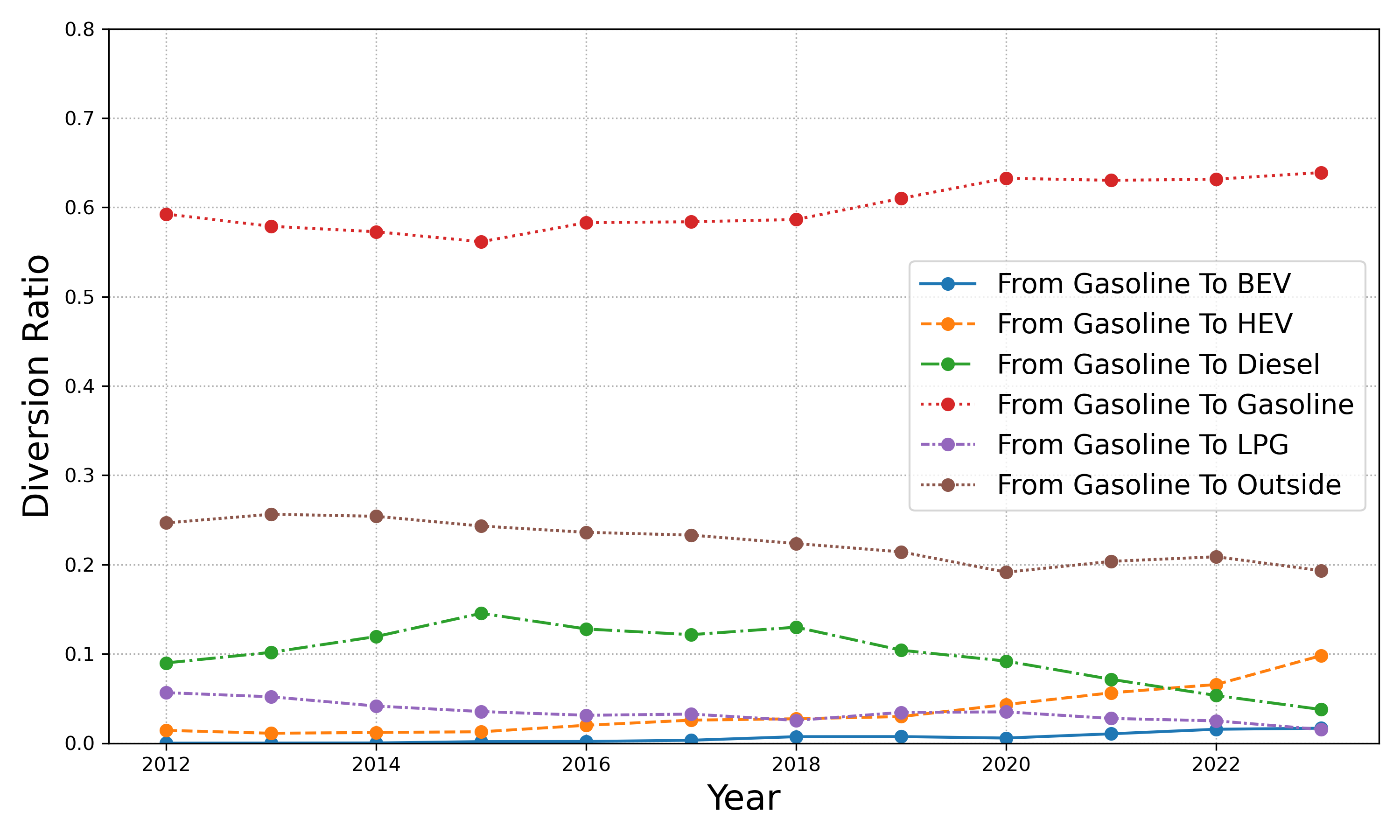}
    	\end{subfigure}
    	\begin{subfigure}[b]{0.48\textwidth}
    \vspace{0.2in}
    	\caption{From LPG}
    	\includegraphics[width=\textwidth]{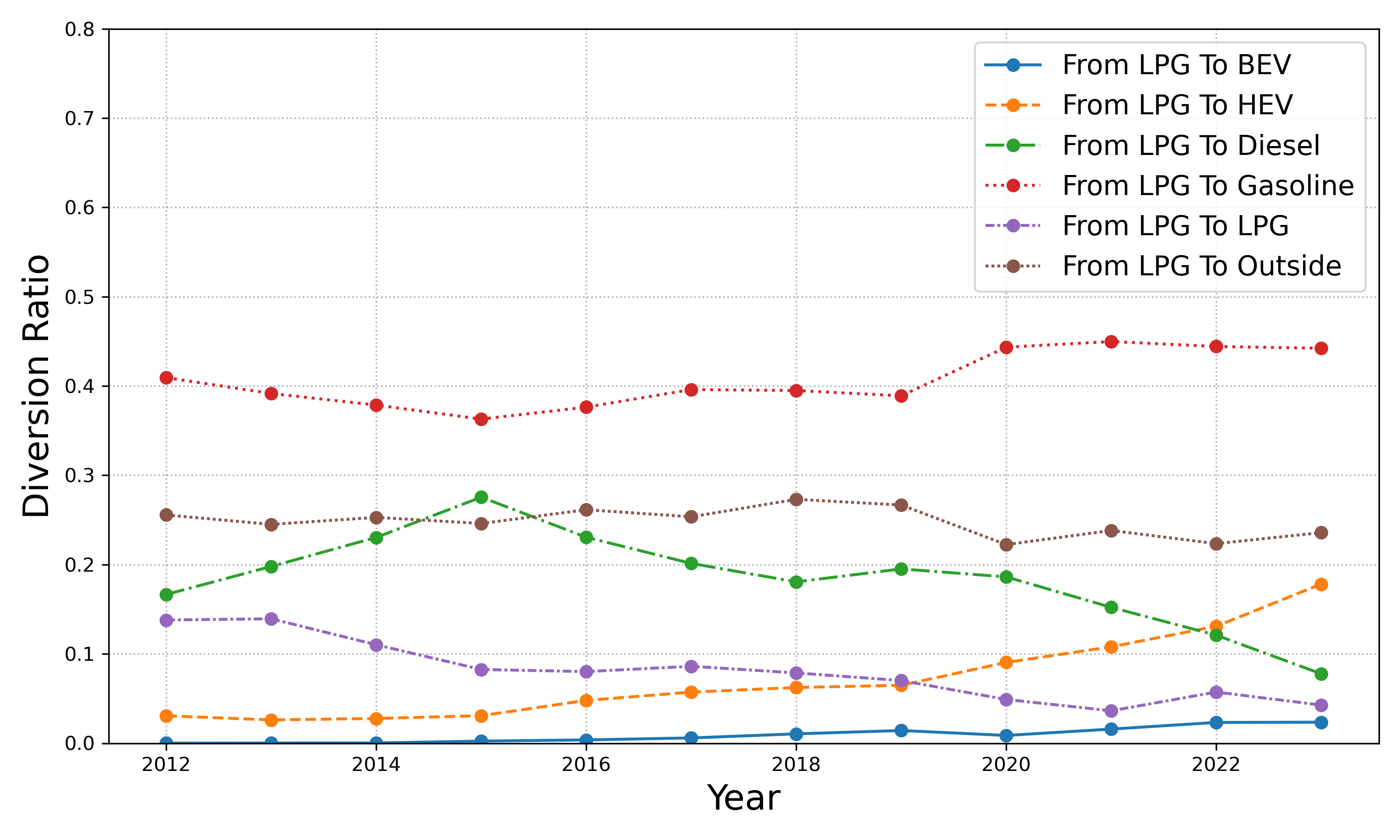}
    	\end{subfigure}
\label{fig: diversion by fuel type}
\tablenotes This figure presents the annual diversion ratios across fuel types between 2018 and 2023, which are computed as follows. First, we calculate the diversion ratio from product $j$ with fuel type $F_1$ to fuel type $F_2$ (or the outside option) in a market: $D_{j,F_2} \equiv \sum_{k \in \mathcal{J}_{F_2}} D_{j,k} \equiv \sum_{k \in \mathcal{J}_{F_2}} \frac{\partial s_k / \partial p_j}{- \partial s_j / \partial p_j},$ where $\mathcal{J}_{F_2}$ denotes the set of products with fuel type $F_2$. Then, we compute the sales-weighted average diversion ratio across fuel types in the market: $D_{F_1,F_2} \equiv \sum_{j \in \mathcal{J}_{F_1}} w_j D_{j,F_2}$, where $w_j$ is product $j$'s share within fuel type $F_1$. Finally, we average the diversion ratios across the 16 markets in each year: $\bar{D}_{F_1,F_2}.$
\end{figure}
\clearpage

\newpage
    \begin{figure}[!t]
    	\centering
    	\caption{Within-group market shares for each fuel type (2023)}
    	\includegraphics[width=.8\textwidth]{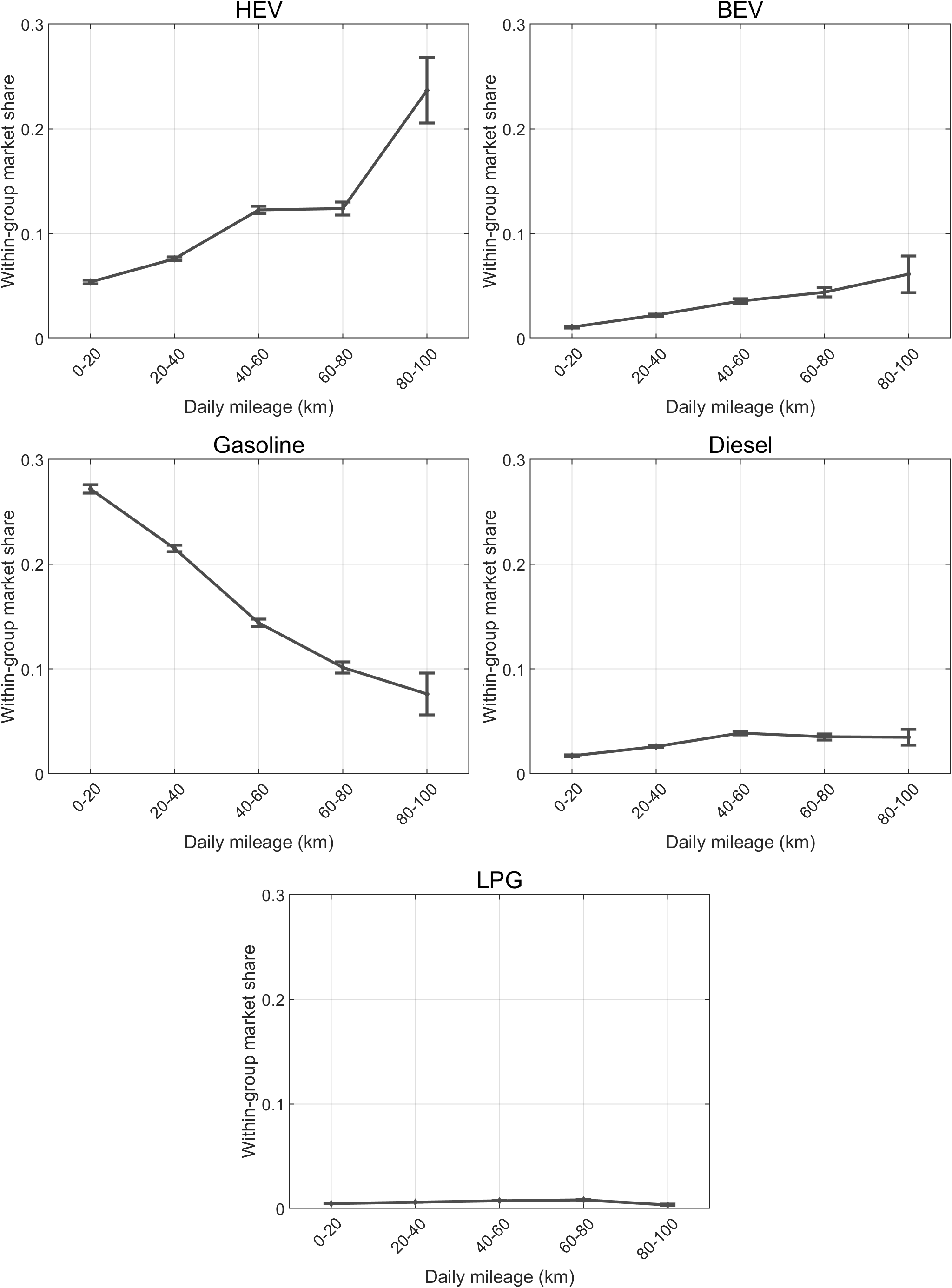}
    \label{fig: mean_ci_2023}
    \tablenotes Each panel in this figure presents nationwide averages (equivalent to within-group market shares) and two-standard-deviation intervals for the predicted choice probabilities of a specific fuel type, calculated separately for each of the five daily mileage groups. To compute these statistics, we first sort the 16\,000 simulated consumers in 2023 (16 provinces $\times$ 1\,000 individuals) into five groups based on their daily mileage: 0--20, 20--40, 40--60, 60--80, and 80--100~km. Next, for each individual $i$, we predict the choice probability for fuel type $g$ as $\hat{s}_{ig} = \sum_{j \in \mathscr{J}_g} \hat{s}_{ij}$. We then calculate the nationwide averages and two-standard-deviation intervals of these probabilities for each mileage group. To ensure that the population share of each regional market is proportionally reflected, we conduct weighted resampling using market sizes as weights.
    \end{figure}
\clearpage

\newpage
\begin{figure}[htbp]
    \centering
    \caption{Market-level distribution of the marginal abatement return by fuel type}
    \includegraphics[width=0.8\textwidth]{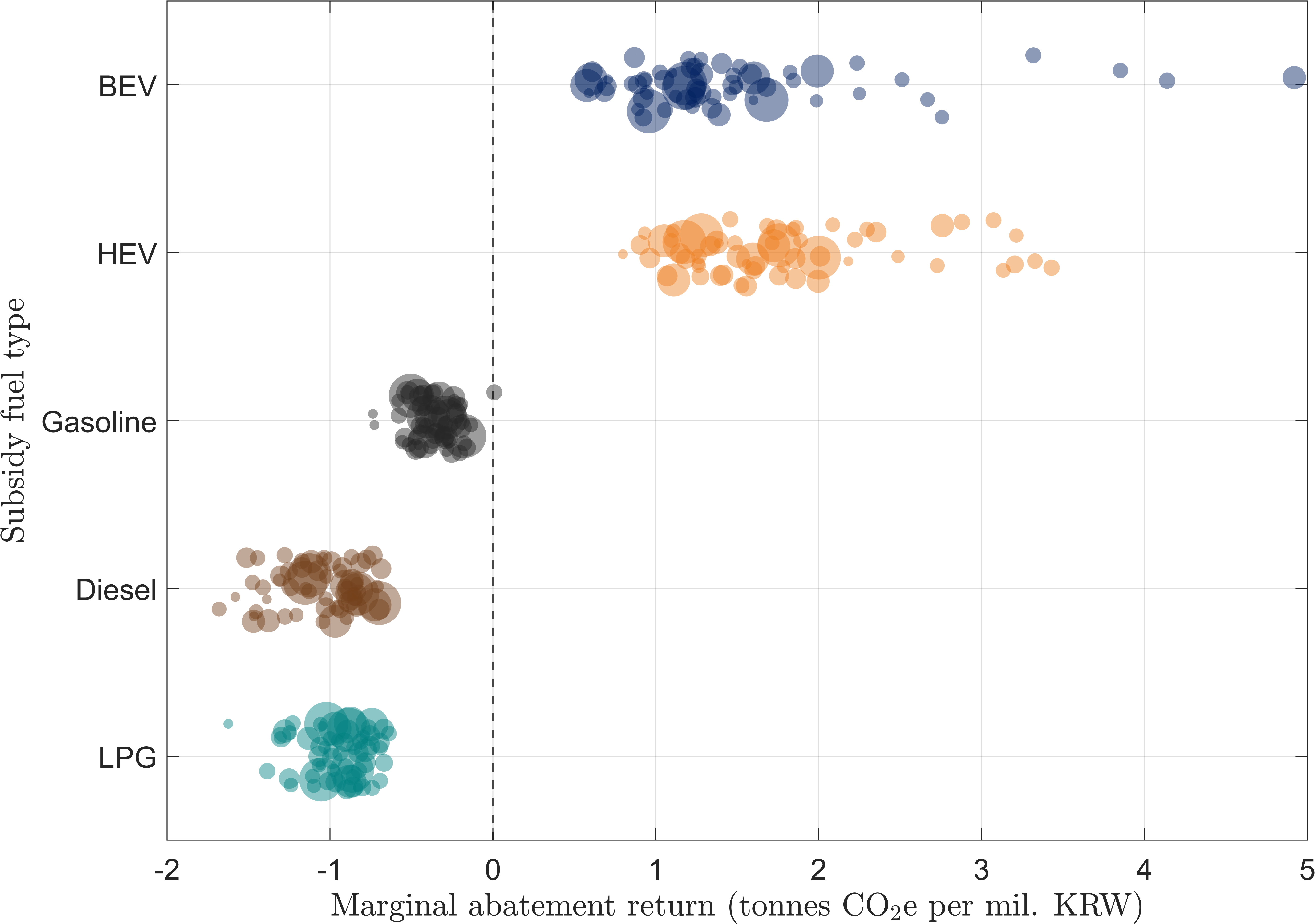}
    \label{fig: abatement_return_bubble_plot}
\tablenotes The figure plots the distributions of market-level marginal abatement returns by fuel type with the unrestricted subsidy pass-through. Each bubble represents the abatement return for a specific fuel-type in a given market, with bubble sizes scaled by the fuel-type’s total sales in each market.
\end{figure}
\clearpage

\end{document}